\documentclass{article}
\usepackage{graphicx}

\usepackage{array}
\usepackage{enumitem}
\usepackage{algorithm}
\usepackage{algpseudocode}
\usepackage[skip=3pt]{caption}

\usepackage{subcaption}

\usepackage{amsmath}
\allowdisplaybreaks[1]

\input{definitions}

\floatname{algorithm}{Mechanism}

\crefname{algorithm}{mechanism}{mechanisms}
\Crefname{algorithm}{Mechanism}{Mechanisms}

\title{Hallucinating Flows for Optimal Mechanisms}

\author{Marios Mertzanidis\thanks{Purdue University. Emails: mmertzan@purdue.edu, aterzogl@purdue.edu.}
\and Athina Terzoglou\footnotemark[1]}

\date{August 2024}

\begin{document}

\newcommand{\mm}[1]{{\color{magenta} Marios: #1}}
\newcommand{\at}[1]{{\color{purple} Athina: #1}}

\maketitle

\begin{abstract}
    Myerson’s seminal characterization of the revenue-optimal auction for a single item \cite{myerson1981optimal} remains a cornerstone of mechanism design. However, generalizing this framework to multi-item settings has proven exceptionally challenging. Even under restrictive assumptions, closed-form characterizations of optimal mechanisms are rare and are largely confined to the single-agent case \cite{pavlov2011optimal,hart2017approximate, daskalakis2018transport, GIANNAKOPOULOS2018432}, departing from the two-item setting only when prior distributions are uniformly distributed \cite{manelli2006bundling, daskalakis2017strong,giannakopoulos2018sjm}. In this work, we build upon the bi-valued setting introduced by Yao \cite{YAO_BIC_DSIC}, where each item's value has support 2 and lies in $\{a, b\}$. Yao's result provides the only known closed-form optimal mechanism for multiple agents. We extend this line of work along three natural axes, establishing the first closed-form optimal mechanisms in each of the following settings: (i) $n$ i.i.d. agents and $m$ i.i.d. items (ii) $n$ non-i.i.d. agents and two i.i.d. items and (iii) $n$ i.i.d. agents and two non-i.i.d. items. Our results lie at the limit of what is considered possible, since even with a single agent and m bi-valued non-i.i.d. items, finding the optimal mechanism is $\#P$-Hard \cite{daskalakis2014complexity, xi2018soda}.  We finally generalize the discrete analog of a result from~\cite{daskalakis2017strong}, showing that for a single agent with $m$ items drawn from arbitrary (non-identical) discrete distributions, grand bundling is optimal when all item values are sufficiently large. We further show that for any continuous product distribution, grand bundling achieves $\mathrm{OPT} - \epsilon$ revenue for large enough values.
\end{abstract}

\section{Introduction}

Designing truthful mechanisms that maximize revenue is a central objective in the field of mechanism design. In this work, we study the setting where an auctioneer seeks to sell $m$ items to $n$ agents. Each agent has a private valuation for each item, representing the maximum amount they are willing to pay. Agents submit bids for the items and, based on these bids, the auctioneer must determine an allocation and corresponding payments. Since agents are strategic, they will bid in a way that maximizes their expected utility over the randomness of the mechanism and the potential bids of other agents. The auctioneer’s goal is to design a mechanism that maximizes the expected revenue while accounting for the strategic nature of the buyers.

In the single-item setting, Myerson’s seminal work~\cite{myerson1981optimal} provides a complete characterization of the revenue-optimal auction, establishing a foundational result in mechanism design. However, despite the significance of this breakthrough and decades of subsequent research, surprisingly few results exist on optimal mechanisms for settings involving multiple items. A notable line of work by Cai et al.~\cite{CDW_1, CDW_3, CDW_2, CDW_4, CDW_5} introduced a general framework for computing exact or approximately optimal mechanisms in a variety of multidimensional settings. Their approach leverages the design of separation oracles combined with cutting-plane methods to compute revenue-maximizing mechanisms. While they provide some characterizations of the optimal mechanisms, these results are algorithmic in nature and do not yield closed-form solutions.

The difficulty of extending Myerson’s result to multi-item settings is evident from the scarcity of closed-form optimal mechanisms, even in highly constrained single-buyer settings. Manelli and Vincent~\cite{manelli2006bundling} identified conditions under which deterministic menu-type mechanisms are optimal for a single buyer, and applied these results to derive optimal mechanisms for two and three items under specific distributional assumptions. A subsequent line of work gradually expanded these results, characterizing optimal mechanisms for one buyer with two items under increasingly general settings~\cite{pavlov2011optimal, hart2017approximate, GIANNAKOPOULOS2018432, daskalakis2018transport}. These works, beyond their theoretical value, revealed the inherent complexity of optimal mechanisms by demonstrating that, even in simple settings, a continuum of lotteries may be required to extract full revenue.

Giannakopoulos and Koutsoupias \cite{giannakopoulos2018sjm} eventually broke the long-standing three item barrier (originally set by Manelli and Vincent) by employing a duality-based approach to characterize the revenue-maximizing mechanism for up to six items, assuming a single buyer whose valuations are i.i.d. uniformly distributed over $[0,1]$. Later, Daskalakis et al.\cite{daskalakis2017strong} developed a general duality framework that not only aids in deriving optimal mechanisms but also certifies their optimality. This framework unified and extended several prior results, including the case of two items under various distributional assumptions, and was used to show that grand bundling is the optimal mechanism for $m$ independent uniformly distributed items over $[c,c+1]$, when $c$ is sufficiently large. For a comprehensive overview of this line of work and the often counterintuitive structure of optimal multi-item mechanisms, we refer the reader to the following survey~\cite{daskalakis2015multi}.

It is important to note that the aforementioned results pertain to settings with continuous valuation distributions, which introduce technical challenges not present in discrete domains. Nevertheless, these results underscore the inherent complexity and significance of designing revenue-maximizing mechanisms in multi-item environments. The first (and, before our work, only) closed-form characterization of an optimal mechanism for a non-trivial multi-agent, multi-item setting is Yao's~\cite{YAO_BIC_DSIC}. Yao considers the case of $n$ agents and two items, where each agent independently values each item at either a low value~$a$ (with probability~$p$) or a high value~$b$ (with probability~$1 - p$). In this bi-valued setting, all values are drawn i.i.d.\ across agents and items. Yao provides a closed-form optimal mechanism under two distinct incentive compatibility notions, Bayesian Incentive Compatibility (BIC) and Dominant Strategy Incentive Compatibility (DSIC).

\paragraph{Further related work.} Due to the complexity of optimal mechanism design, a major line of research has focused on developing approximately optimal mechanisms~\cite{ChawlaHK07, Shuchi2010Copies, ChawlaMS15, RubinsteinW15, kleinberg2019matroid, BabaioffGN17, babaioff2020simple, KothariMSSW19, yao2014n, cai2019duality, chawla2016mechanism, alaei2014bayesian, bhawalkar2024mechanism, cai2017subadditive}. These mechanisms often extend beyond the additive setting, incorporate diverse allocation constraints, support online agent arrivals, and frequently admit simple forms. Another line of work, in order to mitigate the challenges of multi-dimensional mechanism design, has studied the so-called ``1.5-dimensional'' settings~\cite{devanur2020singleMinded, devanur2017budget, fiat2016fedex}, where an agent’s valuation for multiple items is determined by a common randomness. Optimal mechanisms have also been identified in non-additive settings, such as for unit-demand agents~\cite{haghpanah2015reverse, THIRUMULANATHAN201931}. Our techniques rely heavily on the flow interpretation of the dual linear program. This powerful connection has been previously observed in the literature: Cai et al.~\cite{cai2019duality} introduced the flow-based dual framework to prove approximate optimality of simple mechanisms, while Fu et al.~\cite{fu2018value} used it to show an infinite separation between the revenue achievable by BIC and DSIC mechanisms, even in settings with a single item and two correlated agents. Like many prior works, our results assume independently drawn valuations. This assumption is well-justified, as Hart and Nisan~\cite{HartNisan2013} showed that even for a single agent and two correlated items, the optimal mechanism may require an infinite menu. Subsequent research has only achieved success in correlated settings under assumptions, such as weak correlations~\cite{Cai21, makur2023robustness}. Finally, a recent line of work has explored computing optimal mechanisms using deep learning techniques~\cite{wang2024gemnet, dutting2024optimal}. In this context, analytically derived optimal mechanisms are particularly valuable as benchmarks for evaluating data-driven solutions. We are optimistic that our results will contribute in this direction as well.

\paragraph{Our Contribution.} We are interested in developing a general methodology for obtaining optimal mechanisms in discrete settings via strong duality. In \Cref{section: general Methodology}, we first recall the revenue-maximization primal and its dual interpretation as a flow network, where the nodes represent bidder-type profiles, and the edges carrying flow correspond to incentive constraints. That is a positive flow from $v_i$ to $v_i'$ means that the BIC constraint between $v_i$ and $v_i'$ is tight, and flow control enforces consistency. Our key insight is that any feasible flow on this graph induces a hierarchy-allocation mechanism—one that assigns each type a ``score" equal to its virtual value, and awards the item to the highest non-negative score—whose expected revenue under truthful bidding equals the flow’s objective value. By duality, this value serves as an upper bound on the revenue of any feasible mechanism. At the core of our results is the following (informal) theorem, 
\begin{informal}
    If a flow-induced mechanism is BIC and BIR, then it is revenue-optimal.   
\end{informal}
Our approach thus reduces the task of mechanism design to that of ``hallucinating'' feasible flows that induce truthful, individually rational mechanisms.

We use our developed methodology to extend Yao’s results along three distinct axes in \Cref{sec: Applications}. Each axis addresses a key limitation in the current literature, overcoming specific barriers to closed-form characterizations of revenue-maximizing mechanisms in more general settings. \emph{First axis (\Cref{sec: firstAxis}):} We extend Yao’s setting to $n$ agents and $m$ items, providing the \emph{first closed-form optimal mechanism} for arbitrarily many agents and items in a non-trivial family of instances. \emph{Second axis (\Cref{sec: secondAxis}):} We consider the case of $n$ agents and two items, but relax the assumption of identical distributions across agents. Specifically, each agent $i \in [n]$ has their own probability $p_i$, valuing each item at $a$ with probability $p_i$ and at $b$ with probability $1 - p_i$. This constitutes the \emph{first closed-form optimal mechanism} for a setting with \emph{non-identically distributed} bidders. \emph{Third axis (\Cref{sec: thirdAxis}):} We address the setting with $n$ agents and two \emph{non-i.i.d.} items. This is the \emph{first} mechanism to handle arbitrary numbers of agents and non-identically distributed items, and is the most intricate of the three axes due to the presence of multiple structural subcases. Our results approach the boundary of what is computationally feasible. Daskalakis et al.~\cite{daskalakis2014complexity}, and later Xi et al.~\cite{xi2018soda}, showed that computing the revenue-maximizing mechanism for a \emph{single} agent with $m$ items, where item values are interdependently distributed over just two rational values, is \#P-hard. Consequently, deriving a closed-form, optimal mechanism that generalizes Axis~1 and Axis~3 simultaneously appears highly unlikely.            

To derive our results, we develop several intermediate technical lemmas that may be of independent interest. The most significant and technically challenging component in each of the three settings is the derivation of the optimal dual flow. While we present the correct flows directly to the reader, many of them exhibit surprising structural properties. In fact, there exist numerous candidate flows, and identifying the optimal ones requires narrowing the search space by leveraging symmetries inherent in each setting. Through this process, we iteratively develop intuition through trial and error about how different flow choices influence the truthfulness constraints of the resulting mechanism.

In our analysis of the first axis (\Cref{sec: firstAxis}), the optimal flow arises naturally due to the high degree of symmetry across agents and items. This flow is identical for all agents and remains invariant with respect to the number of participants. In the second axis (\Cref{sec: secondAxis}), where agent distributions differ, such symmetry no longer holds; nonetheless, the optimal flow again turns out to be the same across all agents and independent of their number, a result that was both unexpected and elegant.

Even more surprising is the fact that the third axis (\Cref{sec: thirdAxis}) defies this pattern. Although the agents are identically distributed, the structure of the optimal dual flow varies significantly for different inputs. This setting lies closer to the boundary of computational intractability, as generalizations to $m$ items are known to lead to \#P-hardness, thus partially explaining the intricacies of this setting when compared to the previous ones. Letting $p$ and $q$ denote the probabilities that an agent values the first and second item at $a$, respectively, and $n$ be the number of agents, each instance can be represented as a point in a three-dimensional $(p, q, n)$ space. We identify seven distinct, well-defined regions within this space, each corresponding to a different optimal dual flow. This stands in sharp contrast to the first two axes, where the flow admits a unique, closed-form characterization regardless of the agent count or value distributions.

Although we cannot derive a general characterization for the single-agent, $m$-item setting, our methodology still yields valuable insights. In \Cref{sec: extensions}, we show that for any product discrete distribution $\Dcal$ supported on $\bigtimes_{j \in [m]} [c + v_{j}^{\mathrm{low}},\, c + v_{j}^{\mathrm{high}}]$, the optimal mechanism is grand bundling whenever $c$ is sufficiently large. Moreover, we provide explicit lower bounds on $c$ above which this result holds. This generalizes the discrete analog of a result from Daskalakis et al.~\cite{daskalakis2017strong}, who showed that for a single agent with $m$ independent items uniformly and continuously distributed over $[c, c+1]$, grand bundling is optimal when $c$ is large enough. While our result is in the discrete setting, as opposed to their continuous one, it holds for \emph{arbitrary product distributions}, rather than just the uniform case. To further bridge the gap between discrete and continuous settings, we employ standard discretization techniques. Specifically, we show that for any continuous product distribution supported on $[c, c+1]^m$ and any $\epsilon > 0$, there exists a threshold $c^*$ such that for all $c > c^*$, grand bundling achieves revenue at least $\mathrm{OPT} - \epsilon$.

\section{Preliminaries and Notation}

We consider a seller with $m$ indivisible, heterogeneous items to allocate among $n$ strategic agents. For any positive integer $m$, let $[m]=\{1,2,\ldots,m\}.$ Each agent $i \in [n]$ has a private valuation vector $v_i \in \Rm$, drawn independently from a known distribution $\Dcal_i$ over $\mathbb{R}^m$. We denote the support of $\Dcal_i$ by $\Vcal_i = \mathrm{supp}(\Dcal_i)$ and define the joint valuation space as $\Vcal = \bigtimes_{i \in [n]} \Vcal_i$, and the product distribution as $\Dcal = \bigtimes_{i \in [n]} \Dcal_i$.

We assume agents have \emph{additive preferences}: for any subset of items $S \subseteq [m]$, the value of agent~$i$ with type $v_i$ for the bundle $S$ is $\sum_{j \in S} v_{i,j}$. Agents are also \emph{quasi-linear}, meaning their utility is given by their total value for the allocated bundle minus the payment: $u_i(v_i) = \sum_{j \in S} v_{i,j} - p_i$, where $p_i$ is agent $i$'s payment.

For notational convenience, we write $\Pr[v_i]$ to denote the probability that agent $i$ has valuation $v_i$ under $\Dcal_i$, and $\Pr[v_{i,j}]$ to denote the marginal probability that agent $i$ has value $v_{i,j}$ for item $j$ when $v_i$ is sampled from $\Dcal_i$. Similarly, for a valuation profile $v = (v_1, \dots, v_n) \in \Vcal$, we write $\Pr[v]$ for the joint probability under the product distribution $\Dcal$.

A mechanism $\mathcal{M}$ is fully characterized by its allocation and payment rules, i.e., $\mathcal{M} = (x(\cdot), p(\cdot))$. The allocation rule $x: \Rmn \rightarrow [0,1]^{mn}$ maps reported valuation profiles to fractional allocations, where $x_{i,j}(v)$ denotes the probability that agent~$i$ receives item~$j$ given reported valuations $v = (v_1, \dots, v_n)$. The payment rule $p: \Rmn \rightarrow \mathbb{R}^n$ specifies the payment $p_i(v)$ made by each agent $i$ under report profile $v$. Agents are utility-maximizing: the expected utility of agent~$i$ with true valuation $v_i$ who reports $v_i'$ is defined as
\[
\mathbb{E}[u_i(v_i \rightarrow v_i')] := \mathbb{E}_{v_{-i} \sim \Dcal_{-i}}\left[ \sum_{j \in [m]} v_{i,j} \cdot x_{i,j}(v_i', v_{-i}) - p_i(v_i', v_{-i}) \right],
\]
where the expectation is over the other agents’ types drawn from $\Dcal_{-i}$. In particular, $\mathbb{E}[u_i(v_i \rightarrow v_i)]$ denotes the expected utility of agent~$i$ when reporting truthfully.
A mechanism is \emph{Bayesian Incentive Compatible (BIC)} if truthful reporting is a Bayesian Nash equilibrium—that is, for all $i \in [n]$ and all $v_i, v_i' \in \Vcal_i$,
\[
\mathbb{E}[u_i(v_i \rightarrow v_i)]\ge \mathbb{E}[u_i(v_i \rightarrow v_i')].
\]
A mechanism is \emph{Bayesian Individually Rational (BIR)} if each agent has non-negative expected utility when reporting truthfully: for all $i \in [n]$ and all $v_i \in \Vcal_i$,
\[
\mathbb{E}[u_i(v_i \rightarrow v_i)] \ge 0.
\]
The (expected) \emph{revenue} of a mechanism is defined as the expected sum of payments collected when agents draw their valuations from $\Dcal$ and report them truthfully:
\[
\mathrm{Rev}(\mathcal{M}) = \mathbb{E}_{v \sim \Dcal} \left[ \sum_{i=1}^n p_i(v) \right].
\]
We say that a mechanism is \emph{BIC-IR} if it satisfies both BIC and BIR.

The reduced form of a mechanism $\Mcal$ consists of two functions. The \emph{interim allocation} function, denoted by $\pi$, that specifies for each agent $i \in [n]$ and item $j \in [m]$ the probability $\pi_{i,j}(r_i)$ that agent $i$ receives item $j$ when reporting type $r_i$. This probability is taken over the randomness in the mechanism and over the other agents’ types $v_{-i}$, drawn from the prior $\Dcal_{-i}$. Similarly, the \emph{interim payment} $q_i(r_i)$ is the expected payment agent~$i$ makes when reporting $r_i$, again with the expectation over the randomness of $\mathcal{M}$ and $v_{-i} \sim \Dcal_{-i}$. Formally, the interim allocation is $\pi_{i,j}(r_i) = \EX{v_{-i} \sim \Dcal_{-i}}{x_{i,j}(r_i,v_{-i})}$ and the interim payment $q_i(r_i) = \EX{v_{-i} \sim \Dcal_{-i}}{p_{i}(r_i,v_{-i})}.$  It follows that the expected utility of agent~$i$ with true type $v_i$ who reports $r_i$ is given by:
\[
\mathbb{E}[u_i(v_i \rightarrow r_i)] = \sum_{j \in [m]} v_{i,j} \cdot \pi_{i,j}(r_i) - q_i(r_i).
\]

In the Bayesian setting, it is without loss of generality to assume that a mechanism charges agents according to their interim payments. That is, we can transform any BIC-IR mechanism into another mechanism with the same interim allocation and interim payments by directly charging agent~$i$ an amount $q_i(r_i)$ whenever she reports $r_i$. Thus, from this point forward, we let $p_i(r_i)$ denote both the actual and interim payment made by agent~$i$ upon reporting $r_i$.

A similar simplification does \emph{not} apply to the interim allocation. An interim allocation rule does not automatically correspond to a feasible ex-post allocation—i.e., one that never allocates the same item to multiple agents. Ensuring feasibility at the interim level requires satisfying an exponential number of constraints, known as \emph{Border’s conditions}, introduced by Border~\cite{Border1991, Border2007}. These conditions are both necessary and sufficient for an interim allocation to be implementable.

In this paper, we focus on a special class of mechanisms called \emph{hierarchy allocation mechanisms}. The concept of hierarchy allocation was originally introduced in the context of single-item auctions by Border~\cite{Border1991,Border2007}, and later extended to multi-item settings by Cai, Daskalakis, and Weinberg~\cite{CaiDaskalakisWeinberg2012}. In a hierarchy allocation mechanism, each agent-item pair is assigned a ``score", and the item is allocated to the agents with the highest non-negative score. Formally, we define a version of hierarchy allocation tailored to our setting as follows:

\smallskip
\noindent\fbox{%
    \parbox{0.99\textwidth}{%
\begin{definition}[Multi‑Item Hierarchy Allocation]\label{def:multi‑hierarchy}

In a Hierarchy Allocations mechanism, each agent $i \in [n]$ and item $j \in [m]$ has a function $H_{i,j}: \Vcal_i \rightarrow \R$. Assume that agents reported $v = (v_1, v_2, \dots, v_n)$ and let $H_{j}^{max} = \max_{i \in [n]} H_{i,j}(v_i)$:
\begin{enumerate}
    \item If $H_{j}^{max}<0$, item $j$ remains unallocated.
    \item If $H_{j}^{max}>0$, item $j$ is given uniformly at random to the set of agents where $H_{i,j}(v_i) = H_{j}^{max}$.
    \item If $H_{j}^{max}=0$, with probability $\delta$ item $j$ is given uniformly at random to the set of agents where $H_{i,j}(v_i) = 0$ and with probability $1-\delta$ it remains unallocated.
\end{enumerate}
\end{definition}
    }%
}
\smallskip

By definition, a hierarchy allocation mechanism is \emph{always feasible} in terms of item distribution, as each item is allocated to at most one agent. As a result, our analysis can focus solely on ensuring that the mechanism satisfies BIC and IR constraints.

Finally, we use standard notation throughout. For integer $m \ge 0$ and probability $0 \le p \le 1$, we denote by $B(m, p)$ the binomial distribution with $m$ independent Bernoulli trials, each succeeding with probability $p$. For any real number $x \in \mathbb{R}$, we define $[x]^+ := \max\{x, 0\}$ to denote the positive part of $x$.

\section{The General Methodology} \label{section: general Methodology}

In this section, we present our general methodology for designing revenue-optimal mechanisms in discrete, additive settings. Our approach builds on the duality framework introduced by Cai et al.~\cite{cai2019duality}, which connects the revenue maximization problem to a linear program and its dual, interpreted as a flow problem over the space of types. The central idea is that any feasible dual solution  (i.e., a flow) yields a mechanism with expected revenue equal to the flow’s objective value if the induced mechanism is BIC and IR. Therefore, to design an optimal mechanism, it suffices to identify a flow that (i) satisfies the dual constraints, and (ii) induces a BIC-IR mechanism. This reduction allows us to focus on constructing and analyzing flows rather than directly reasoning about the space of mechanisms. We start by presenting the LP that produces the revenue-maximizing mechanism.

\begin{equation}\label{LP:Primal}
\begin{aligned}
\max_{x, p} \quad &  \sum_{i \in [n]} \sum_{v_i \in \mathcal{V}}  \Pr[v_i] \cdot p_i(v_i) \\
\text{s.t.} \quad 
& \mathbb{E}\left[ u_i(v_i \rightarrow v_i) \right]  \ge  \mathbb{E}\left[ u_i(v_i \rightarrow v_i') \right], 
&& \forall i \in [n],\, (v_i, v'_i) \in \mathcal{V}_i^2 \\
& \mathbb{E}\left[ u_i(v_i \rightarrow v_i) \right] \ge 0, 
&& \forall i \in [n],\, v_i \in \mathcal{V}_i \\
& \sum_{i \in [n]} x_{i,j}(v) \le 1, 
&& \forall v \in \mathcal{V},\, j \in [m] \\
& x \ge 0
\end{aligned}
\end{equation}
By expanding the notation, we can verify that this is indeed a linear program, where the first set of constraints ensures that the mechanism is BIC, the second set of constraints ensures that it is BIR, while the last constraint ensures that each item is allocated only once. Unfortunately, solving the above LP requires exponential time since the allocation rule $x(\cdot)$ introduces exponentially many variables. Cai et al.~\cite{CDW_1} provide an efficiently computable separation oracle for the interim version of the above LP. These results can be used to calculate in polynomial time an $\epsilon$-approximation of the optimal mechanism. However, in this paper, we are not simply interested in calculating the optimal mechanism; our goal is to derive an interpretable, closed-form solution. It is easy to verify that the dual of the above LP is:
\begin{equation}\label{LP:Dual}
\begin{aligned}
\min_{\lambda, \mu, \kappa} \quad & \sum_{v \in \mathcal{V}} \sum_{j \in [m]} \kappa_j(v) \\
\text{s.t.} \quad 
& \Pr[v_i] + \sum_{v'_i \in \mathcal{V}_i} \lambda_i(v'_i, v_i) = \sum_{v'_i \in \mathcal{V}_i} \lambda_i(v_i, v'_i) + \mu_i(v_i), 
&& \forall i \in [n],\, v_i \in \mathcal{V}_i \\
& \kappa_j(v_i, v_{-i}) \ge \Pr[v_i, v_{-i}] \left( 
v_{i,j} - \frac{1}{\Pr[v_i]} \sum_{v'_i \in \mathcal{V}_i} \lambda_i(v'_i, v_i)(v'_{i,j} - v_{i,j}) 
\right), 
&& \forall i \in [n],\, j \in [m],\, (v_i, v_{-i}) \in \mathcal{V} \\
& \lambda, \mu, \kappa \ge 0
\end{aligned}
\end{equation}

We immediately observe that the first set of constraints corresponds to \emph{flow conservation} at each node, enforcing that the net flow into each type equals the net flow out. The objective value, given a feasible flow, is determined by the variables $\kappa$, which are constrained by the second set of inequalities. The right-hand side of these inequalities can be interpreted as \emph{virtual values}, capturing the revenue contribution of each type-flow pair. We formalize this intuition as follows:

\vspace{5pt}
\noindent\fbox{%
    \parbox{\textwidth}{%
        {\bf The Dual Flow.}\\ 
         Consider the dual LP. We can interpret the first constraint as a flow conservation constraint. Through this lens, for each agent $i$ we construct a graph where each type $v_i \in \Vcal_i$ is a node. For each $v_i \in \Vcal_i$ and $v_i' \in \Vcal_i$ if $\lambda_i(v_i', v_i)>0$ then there exists a directed edge from node $v_i'$ to node $v_i$ that carries $\lambda_i(v_i', v_i)$ amount of flow. We also have a source node $s$ that sends $\Pr[v_i]$ flow to each node $v_i \in \Vcal_i$, and a sink node $\bot$ to which each node $v_i \in \Vcal_i$ send $\mu_i(v_i)$ flow. For every feasible solution of the dual we will refer to $\lambda$ and $\mu$ as the Dual Flow.
    }%
}
\begin{definition}[Flow Decomposition]
    Let $\Pcal$ be the set of all simple paths from $s$ to $\bot$. Let $\Pcal_{v_i}$ be the set of all simple paths from $s$ where the first step is $v_i$, so $\Pcal = \bigcup_{v_i \in \Vcal_i}\Pcal_{v_i} $. For any flow from $s$ to $\bot$ represented by our dual variables $\lambda, \mu$, that has no cycles, we can decompose the flow over the simple paths from $s$ to $\bot$ such that for each path $\ell \in \Pcal$ there exists a flow $\xi_{\ell}$ such that (i) $\sum_{\ell_{v_i} \in \Pcal_{v_i}} \xi_{\ell_{v_i}} = \Pr[v_i]$, (ii) $\lambda_i(v'_i, v_i) =  \sum_{\ell \in \Pcal: (v'_i, v_i) \in \ell}  \xi_{\ell}$, and (iii) $\mu_i(v_i) =  \sum_{\ell \in \Pcal: (v_i, \bot) \in \ell}  \xi_{\ell}$.
\end{definition}
The existence of such a flow decomposition for acyclic flows follows from the well-known \emph{Flow Decomposition Theorem}. We are now ready to formally define the mechanism induced by a given feasible flow.

\vspace{5pt}
\noindent\fbox{%
    \parbox{\textwidth}{%
        {\bf The Flow Induced Mechanism.}\\ 
         Consider some feasible Dual Flow $\lambda, \mu$ that can be represented by a DAG. We will define the following mechanism which we will say is induced by the flow $\lambda, \mu$. The allocation rule will be a Hierarchy Allocation rule where for each $i \in [n]$, $j \in [m]$, and $v_i \in \Vcal_i$, $H_{i,j}(v_i) = v_{i,j} - \frac{1}{\Pr[v_i]}\sum_{v'_i \in \mathcal{V}_i} \lambda_i(v'_i, v_i)(v'_{i,j} - v_{i,j})$.  Let $\xi$ be a decomposition of the given flow. For any path $\ell \in \Pcal_{v_i}$ let $\ell = (s, v^1_i, v_{i}^2, \dots, v_{i}^{|\ell|-2},\bot)$, where $v^1_i = v_i$. The payment of agent $i \in [n]$ when reporting $v_i \in \Vcal_i$ will be:
         
         \[p_i(v_i) = \frac{1}{\Pr[v_i]}\sum_{\ell \in \Pcal_{v_i}} \xi_{\ell} \left(\sum_{j \in [m]} v_{i,j} \pi_{i,j}(v_i) - \sum_{z \in [|\ell|-3]}  \sum_{j \in [m]} \left( v_{i,j}^{z}-v_{i,j}^{z+1} \right)\pi_{i,j}(v_{i}^{z+1})\right)\]
         where the interim allocation rules $\pi_{i,j}(v_{i}^{z})$ are given from the hierarchy allocation rule and the distribution over agents' values (assuming truthful bidding).
    }%
}
\smallskip
 
 In analogy with Myerson's~\cite{myerson1981optimal} seminal work, $H_{i,j}(v_i)$ can be viewed as the virtual value of agent $i$ for item $j$ with real value $v_i$. The mechanism allocates the item to the highest non-negative virtual value and then charges an amount that ensures that the expected extracted revenue is equal to the objective of the flow. We crystallize this claim through the following theorem.
\begin{theorem} \label{thm: main Thm}
    For any feasible Dual Flow, if the Flow Induced Mechanism is BIC-IR then it is optimal.
\end{theorem}

\begin{proof}
    
Let $\lambda, \mu$ be the Dual Flow and $\xi$ its path decomposition. Assuming truthful bidding the revenue is:
\begin{align*}
    \mathrm{Rev} &= \sum_{i \in [n]} \sum_{v_i \in \Vcal_i} \Pr[v_i] p_i(v_i) \\
    &= \sum_{i \in [n]} \sum_{v_i \in \Vcal_i} \Pr[v_i]\frac{1}{\Pr[v_i]}\sum_{\ell \in \Pcal_{v_i}} \xi_{\ell} \left(\sum_{j \in [m]} v_{i,j} \pi_{i,j}(v_i) - \sum_{z \in [|\ell|-3]}  \sum_{j \in [m]} \left( v_{i,j}^{z}-v_{i,j}^{z+1} \right)\pi_{i,j}(v_{i}^{z+1})\right)\\
    &= \sum_{i \in [n]} \sum_{v_i \in \Vcal_i} \sum_{\ell \in \Pcal_{v_i}} \xi_{\ell} \left( \sum_{j \in [m]} v_{i,j} \pi_{i,j}(v_i) - \sum_{z \in [|\ell|-3]}  \sum_{j \in [m]} (v_{i,j}^{z} - v_{i,j}^{z+1}) \pi_{i,j}(v_{i}^{z+1})\right)\\
    &= \sum_{i \in [n]} \sum_{v_i \in \Vcal_i} \sum_{\ell \in \Pcal_{v_i}} \xi_{\ell} \left( \sum_{j \in [m]} v_{i,j} \sum_{v_{-i} \in \Vcal_{-i}} \Pr[v_{-i}]x_{i,j}(v_i, v_{-i}) \right.\\
    &\phantom{\sum_{i \in [n]} \sum_{v_i \in \Vcal_i} \sum_{\ell \in \Pcal_{v_i}} \xi_{\ell}(\sum_{j \in [m]} v_{i,j}} \left.- \sum_{z \in [|\ell|-3]}  \sum_{j \in [m]} (v_{i,j}^{z}-v_{i,j}^{z+1}) \sum_{v_{-i} \in \Vcal_{-i}} \Pr[v_{-i}]x_{i,j}(v_{i}^{z+1}, v_{-i})\right)\\
    &= \sum_{i \in [n]} \sum_{v_i \in \Vcal_i} \sum_{\ell \in \Pcal_{v_i}} \sum_{v_{-i} \in \Vcal_{-i}} \xi_{\ell} \left( \sum_{j \in [m]} v_{i,j} \frac{\Pr[v_i, v_{-i}]}{\Pr[v_i]}x_{i,j}(v_i, v_{-i}) \right.\\
    & \phantom{\sum_{i \in [n]} \sum_{v_i \in \Vcal_i} \sum_{\ell \in \Pcal_{v_i}} \sum_{v_{-i} \in \Vcal_{-i}} \xi_{\ell} (\sum_{j \in [m]} v_{i,j}} \left.- \sum_{z \in [|\ell|-3]}  \sum_{j \in [m]} (v_{i,j}^{z}-v_{i,j}^{z+1})  \frac{\Pr[v_{i}^{z+1}, v_{-i}]}{\Pr[v_{i}^{z+1}]} x_{i,j}(v_{i}^{z+1}, v_{-i})\right)\\
     &= \sum_{i \in [n]} \sum_{v_i \in \Vcal_i} \sum_{j \in [m]} \sum_{v_{-i} \in \Vcal_{-i}}  \left(  \sum_{\ell \in \Pcal_{v_i}} \xi_{\ell} \cdot v_{i,j} \frac{\Pr[v_i, v_{-i}]}{\Pr[v_i]}x_{i,j}(v_i, v_{-i}) \right. \\
     &\phantom{\sum_{i \in [n]} \sum_{v_i \in \Vcal_i} \sum_{j \in [m]} \sum_{v_{-i} \in \Vcal_{-i}} (  \sum_{\ell \in \Pcal_{v_i}} \xi_{\ell} \cdot v_{i,j} } \left.- \sum_{\ell \in \Pcal_{v_i}} \xi_{\ell} \sum_{z \in [|\ell|-3]}   (v_{i,j}^{z}-v_{i,j}^{z+1})  \frac{\Pr[v_{i}^{z+1}, v_{-i}]}{\Pr[v_i]} x_{i,j}(v_{i}^{z+1}, v_{-i})\right)\\
    &= \sum_{i \in [n]} \sum_{v_i \in \Vcal_i} \sum_{j \in [m]} \sum_{v_{-i} \in \Vcal_{-i}}  \left(  \sum_{\ell \in \Pcal_{v_i}} \xi_{\ell} \cdot v_{i,j} \frac{\Pr[v_i, v_{-i}]}{\Pr[v_i]}x_{i,j}(v_i, v_{-i}) \right.  \tag{Rearranging}\\
    & \phantom{\sum_{i \in [n]} \sum_{v_i \in \Vcal_i} \sum_{j \in [m]} \sum_{v_{-i} \in \Vcal_{-i}} (  \sum_{\ell \in \Pcal_{v_i}} \xi_{\ell} \cdot v_{i,j} } \left.- \frac{\Pr[v_i, v_{-i}]}{\Pr[v_i]}\sum_{v_i' \in \Vcal_i} \sum_{\ell \in \Pcal: (v_i',v_i) \in \ell}   x_{i,j}(v_i, v_{-i}) \xi_{\ell} (v_{i,j}'-v_{i,j})   \right) \\
    &= \sum_{i \in [n]} \sum_{v_i \in \Vcal_i} \sum_{j \in [m]} \sum_{v_{-i} \in \Vcal_{-i}} \Pr[v_i, v_{-i}]x_{i,j}(v_i, v_{-i}) \left(  \sum_{\ell \in \Pcal_{v_i}} \xi_{\ell} \cdot v_{i,j} \frac{1}{\Pr[v_i]} - \frac{1}{\Pr[v_i]}\sum_{v_i' \in \Vcal_i} \sum_{\ell \in \Pcal: (v_i',v_i) \in \ell} \xi_{\ell} (v_{i,j}'-v_{i,j})   \right) \\
    &= \sum_{i \in [n]} \sum_{v_i \in \Vcal_i} \sum_{j \in [m]} \sum_{v_{-i} \in \Vcal_{-i}} \Pr[v_i, v_{-i}]x_{i,j}(v_i, v_{-i}) \left( v_{i,j}  - \frac{1}{\Pr[v_i]}\sum_{v_i' \in \Vcal_i} \sum_{\ell \in \Pcal: (v_i',v_i) \in \ell} \xi_{\ell} (v_{i,j}'-v_{i,j})   \right) \tag{$ \sum_{\ell \in \Pcal_{v_i}} \xi_{\ell} = \Pr[v_i]$} \\
    &= \sum_{i \in [n]} \sum_{v_i \in \Vcal_i} \sum_{j \in [m]} \sum_{v_{-i} \in \Vcal_{-i}} \Pr[v_i, v_{-i}]x_{i,j}(v_i, v_{-i}) \left( v_{i,j}  - \frac{1}{\Pr[v_i]}\sum_{v_i' \in \Vcal_i} \lambda_i(v_i', v_i) (v_{i,j}'-v_{i,j})   \right) \tag{$\sum_{\ell \in \Pcal: (v_i',v_i) \in \ell} \xi_{\ell} = \lambda_i(v_i', v_i) $} \\\\
\end{align*}

However notice that given the Hierarchy allocation mechanism we have chosen, for every profile of bids $v \in \Vcal$ the item $j \in [m]$ will be allocated to the agents with the maximum virtual value for that item, $H_{i,j}(v_i) = v_{i,j} - \frac{1}{\Pr[v_i]}\sum_{v'_i \in \mathcal{V}_i} \lambda_i(v'_i, v_i)(v'_{i,j} - v_{i,j})$. If this maximum is negative, then the item will remain unallocated. Thus the expected revenue of the flow induced mechanism will be equal to the objective of our Dual LP. Since it is BIC-IR and each item is allocated to at most one agent the mechanism is feasible and thus it is optimal with the dual flow acting as a certificate of optimality.
\end{proof}

\section{Applications} \label{sec: Applications}

In what follows, we instantiate our general methodology in four distinct settings, deriving optimal mechanisms and revealing key structural properties in each.

\subsection{First Axis}\label{sec: firstAxis}

In this section, we present our analysis for the first axis. Recall that in this setting, there are $n$ agents and $m$ items. Each agent values each item independently at $a$ with probability $p$, and at $b$ with probability $1 - p$, with $a<b$. We begin by describing the proposed mechanism. We will then prove its optimality by showing that it is induced by a particular dual flow and satisfies truthfulness. Before stating the mechanism, we introduce a quantity that will be used extensively throughout the analysis: for any valuation vector $v_i$, we define $k_{v_i}$ as the number of items agent~$i$ values at the high value~$b$. Formally,
\[
k_{v_i} := \left| \left\{ j \in [m] \;\middle|\; v_{i,j} = b \right\} \right|,
\]

\begin{algorithm}[ht]
\caption{First Axis}\label{mech: firstAxis}
\begin{algorithmic}
\State For any agent $i \in [n]$, item $j \in [m]$, and valuation vector $v_i \in \Vcal_i$, if $v_{i,j} = b$ then $H_{i,j}(v_i) = b$. Else, if $v_{i,j} = a$ then:
\[H_{i,j}(v_i) = a - \frac{1}{(1-p)^{k_{v_i}} p^{m-k_{v_i}} } \cdot\frac{1}{(m-k_{v_i})\binom{m}{k_{v_i}}}\cdot \sum_{z=k_{v_i}+1}^m\binom{m}{z}(1-p)^zp^{m-z} \cdot (b-a)\]
\State Let $k^* = \arg \min_{k \in [0,m]}\left\{a > \frac{1}{(1-p)^k p^{m-k} } \cdot\frac{1}{(m-k)\binom{m}{k}}\cdot \sum_{z=k+1}^m\binom{m}{z}(1-p)^zp^{m-z} \cdot (b-a)\right\}$. Then the payment of agent $i$, for reporting $v_i \in \Vcal_i$ where $k_{v_i} \ge k^*$ is:

\begin{align*}
    p_{i}(v_i) &= k_{v_i}b\sum_{z=0}^{n-1}\left(\frac{1}{z+1}\Pr[B(n-1,1-p)=z]\right)\\
    &+(m-k_{v_i})ap^{n-1} \frac{\left(\Pr[B(m-1,1-p) \le k_{v_i}] \right)^n- \left( \Pr[B(m-1,1-p)<k_{v_i}] \right)^n}{n\Pr[B(m-1,1-p)=k_{v_i}]}\\
    &-(b-a)p^{n-1}\sum_{\tau = k^*}^{k_{v_i}-1}\frac{\left(\Pr[B(m-1,1-p)\le \tau] \right)^n- \left( \Pr[B(m-1,1-p)<\tau] \right)^n}{n\Pr[B(m-1,1-p)=\tau]}
\end{align*}
else if $k_{v_i} < k^*$,  $p_i(v_i) = k_{v_i}b\sum_{z=0}^{n-1}\left(\frac{1}{z+1}\Pr[B(n-1,1-p)=z]\right)$.
\State
\end{algorithmic}
\end{algorithm}

It is important to note that precomputing the full allocation and payment rules for all possible valuation profiles requires exponential time in the number of items. However, given a specific bid profile, the allocation and payment can be computed efficiently in polynomial time, making the mechanism practical for implementation. We proceed to describe the dual flow that induces the above mechanism.

\begin{definition}(\Cref{mech: firstAxis} Flow) \label{def: firstAxis flow}
Let for every $i \in [n]$ and every $v_i \in \Vcal_i$, $S(v_i) = \{v'_i \in \Vcal_i: \exists j \text{ such that } v_{i,j}=b, v'_{i,j}=a \text{ and } \forall j'\neq j, v_{i,j'}=v'_{i,j'}\} $ (i.e. the set of valuations obtained from $v_i$ by changing the value of exactly one item from $b$ to $a$). Then we define the \Cref{mech: firstAxis} Flow through the following recursive relationship:
\begin{itemize}
    \item For all  $v_i \in \Vcal_i$, $\mu_i(v_i) = \mathbf{1}[v_i = [a,a,\dots,a]]$, where $\mathbf{1}[\cdot]$ is the indicator function.
    \item For every $v_i \in \Vcal_i$ let. Then for any $v_i' \in \Vcal_i$, 
    
    \[\lambda_i(v_i,v_i')= \mathbf{1}[v_i' \in S(v_i) \;\&\; v_i \neq v_i'] \cdot \frac{1}{k_{v_i}}\left((1-p)^{k_{v_i}}p^{m-k_{v_i}}+ \sum\limits_{\tilde{v}_i \in \Vcal_i: v_i \in S(\tilde{v}_i)} \lambda_i(\tilde{v}_i, v_i) \right).\]
\end{itemize}
\end{definition}

To better understand the structure of the flow, it is helpful to visualize it as being organized into \emph{layers}, where layer~$k$ contains all valuation vectors in $\Vcal_i$ with exactly $k$ items valued at $b$ (and the remaining $m - k$ items valued at $a$). For any $v_i \in \Vcal_i$ in layer $k$, we define $S(v_i)$ to be the set of valuation vectors in layer $k-1$ that differ from $v_i$ in exactly one coordinate. We refer to $S(v_i)$ as the \emph{children} of $v_i$, and to $\{\tilde{v}_i \in \Vcal_i : v_i \in S(\tilde{v}_i)\}$ as its \emph{parents}. It is easy to see that $|S(v_i)| = k$, since there are $k$ items in $v_i$ valued at $b$ that can be changed to $a$ to form a child. Each node $v_i$ receives $(1 - p)^k p^{m - k}$ units of flow from the source (corresponding to its probability under the product distribution), as well as flow from its parents in layer $k + 1$, given by $\sum_{\tilde{v}_i \in \Vcal_i : v_i \in S(\tilde{v}_i)} \lambda_i(\tilde{v}_i, v_i)$. It then distributes its total incoming flow equally among all $k$ children in $S(v_i)$. \Cref{fig:4item_flow} illustrates these flow relationships for the case of $m = 4$ items.

\begin{figure}[ht]
    \centering
    \includegraphics[width=0.75\textwidth]{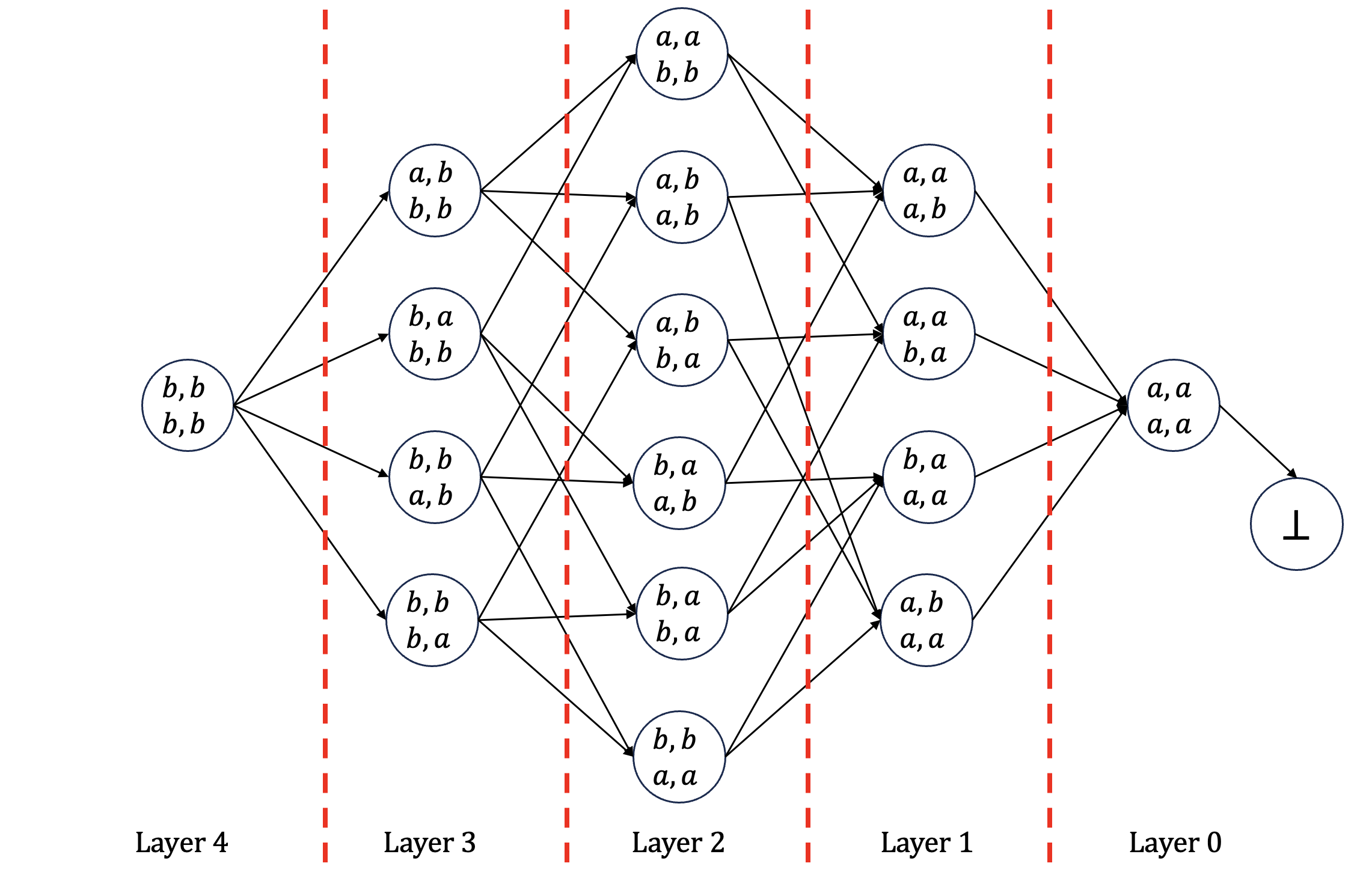}
    \caption{Flow relationships and layers when $m=4$.}
    \label{fig:4item_flow}
\end{figure}
Before we proceed, we need to present a few technical lemmas. The proofs are deferred to the~\Cref{apx: proofs-Axis}.

\begin{lemma} \label{lemma: technical sum binom}
$\sum_{j=1}^{n} \frac{1}{j}\binom{n-1}{j-1} q^{j-1}p^{n-j} = \frac{(p+q)^n- p^n}{nq}$
\end{lemma}


\begin{lemma}\label{lemma: Binomial Inequality}
    $\pr{B(m, p) = k | B(m, p) \ge k} \ge 1- \frac{p(m-k)}{(1-p)k}$.
\end{lemma}

Using the lemma above, we can now show that the values $ H_{i,j}(v_i) $ for items valued at $ a $ decrease as we move to lower layers. In other words, if two agents report value $ a $ for the same item, the one who reports more $ b $'s for the remaining items is given priority. This property is central to the structure of the flow: it allows us to incentivize bidders to report more $ b $'s, even at higher cost, by offering them preferential treatment for other items. Proving this behavior is nontrivial because, up to the middle layer, two opposing forces are at play. On the one hand, the total amount of flow entering each successive layer increases. On the other hand, the number of edges carrying that flow also grows, making it unclear whether the flow per edge increases as the layers decrease. The technical lemma resolves this tension by allowing us to establish that the flow per edge indeed increases, enabling the desired incentive structure. The proof of the following lemma is at~\Cref{apx: proofs-Axis}.

\begin{lemma}\label{lemma: tech-f(k)}
    Let $f(k) = a - \frac{1}{(1-p)^k p^{m-k} } \cdot\frac{1}{(m-k)\binom{m}{k}}\cdot \sum_{z=k+1}^m\binom{m}{z}(1-p)^zp^{m-z} \cdot (b-a)$. Then if $k\ge k'$, $f(k)\ge f(k')$ for any choice of $p \in [0,1]$, and $m \in \mathbb{Z}$.
\end{lemma}


Now we are ready to show the first main result of this section.

\begin{lemma}\label{lemma: first axis mechanism induced by flow}
    \Cref{mech: firstAxis} is induced by the flow defined in \Cref{def: firstAxis flow}.
\end{lemma}

\begin{proof}

    We begin by showing that the allocation rule defined in \Cref{mech: firstAxis} is induced by the flow described in \Cref{def: firstAxis flow}. To do so, we compute the exact amount of flow carried by each edge $(v_i', v_i)$ in the flow network. Recall the layered interpretation of the flow from \Cref{fig:4item_flow}, where nodes are grouped into layers according to the number of items valued at $b$. We will use a simple inductive argument to establish that all edges traversing from layer $k$ to layer $k-1$ carry the same amount of flow. Assume this holds for some layer $k$. By symmetry, each node in layer $k-1$ receives the same amount of flow from its parents in layer $k$ and the source. Moreover, each such node has exactly $k - 1$ children in layer $k - 2$, and it distributes its incoming flow equally among them. Thus, the flow carried by each edge from layer $k-1$ to layer $k-2$ is also equal, completing the inductive step. For the base case, note that there is only one node in layer $m$, which receives flow directly from the source and splits it equally among its $m$ children in layer $m - 1$. 

    As established earlier, each node in layer $k+1$ has exactly $k+1$ outgoing edges, and there are $\binom{m}{k+1}$ such nodes. Therefore, the total number of edges traversing from layer $k+1$ to layer $k$ is
    \[
    (k+1) \binom{m}{k+1} = (k+1) \cdot \frac{m!}{(m-k-1)!(k+1)!} = (m-k) \cdot \frac{m!}{(m-k)!k!} = (m-k) \binom{m}{k}.
    \]
    Each node in layer $z$ receives $(1-p)^z p^{m-z}$ units of flow from the source. Since only the node in layer $0$ sends flow to the sink, the total amount of flow traversing from layer $k+1$ to layer $k$ must equal the total flow introduced in layers $k+1, k+2, \dots, m$. That is, $\sum_{z = k+1}^m \binom{m}{z} (1-p)^z p^{m-z}$.Consequently, the flow carried by each edge from layer $k+1$ to layer $k$ is $\frac{1}{(m - k)\binom{m}{k}} \sum_{z = k+1}^m \binom{m}{z} (1 - p)^z p^{m - z}$.

    We now confirm that for all $i \in [n]$, $j \in [m]$, and $v_i \in \Vcal_i$, the virtual value is given by $H_{i,j}(v_i)$ as defined in the mechanism. First, consider the case where $v_{i,j} = b$. Then for any $v_i' \in \Vcal_i$, either $v_{i,j}' = b$ (so $v_{i,j}' - v_{i,j} = 0$), or $\lambda_i(v_i', v_i) = 0$ by the structure of the flow. Thus, all terms in the sum vanish, and we have $H_{i,j}(v_i) = b$. Now consider the case where $v_{i,j} = a$. In this case, for any $v_i'$ such that $\lambda_i(v_i', v_i) > 0$, we must have $v_{i,j}' = a$, except for a single predecessor $\tilde{v}_i \in \Vcal_i$ where $\tilde{v}_{i,j} = b$. By construction of the flow, exactly one such $\tilde{v}_i$ exists. Combining all of the above, we have:
\begin{align*}
    H_{i,j}(v_i) &= v_{i,j} - \frac{1}{\Pr[v_i]} \sum_{v_i' \in \Vcal_i} \lambda_i(v_i', v_i)(v_{i,j}' - v_{i,j}) \\
    &= a - \frac{1}{(1-p)^{k_{v_i}} p^{m-k_{v_i}}} \cdot \frac{1}{(m - k_{v_i}) \binom{m}{k_{v_i}}} \cdot \sum_{z = k_{v_i} + 1}^m \binom{m}{z}(1 - p)^z p^{m - z} \cdot (b - a)
\end{align*}

We now proceed to show that the payment rule of \Cref{mech: firstAxis} is induced by the flow defined in \Cref{def: firstAxis flow}. Let $\pi$ denote the interim allocation rule under truthful bidding. First, consider any $i \in [n]$, $v_i \in \Vcal_i$, and item $j \in [m]$ such that $v_{i,j} = b$. If exactly $z$ of the remaining $n - 1$ agents also report a value of $b$ for item $j$, then agent~$i$ receives the item with probability $1 / (z + 1)$. The probability that exactly $z$ from the remaining $n-1$ agents have value $b$ for item $j$ is $\Pr[B(n - 1, 1 - p) = z]$. Hence, the interim allocation in this case is $\pi_{i,j}(v_i) = \sum_{z = 0}^{n - 1} \frac{1}{z + 1} \cdot \Pr[B(n - 1, 1 - p) = z],$ which depends only on the fact that $v_{i,j} = b$, and not on $i$, $j$, or the rest of $v_i$. For ease of notation, define:
\[
\pi(b) := \sum_{z = 0}^{n - 1} \frac{1}{z + 1} \cdot \Pr[B(n - 1, 1 - p) = z].
\]

Now consider the case where $v_{i,j} = a$. Recall that
\[
k^* := \arg \min_{k \in [0, m]} \left\{ a > \frac{1}{(1 - p)^k p^{m - k}} \cdot \frac{1}{(m - k) \binom{m}{k}} \cdot \sum_{z = k + 1}^{m} \binom{m}{z} (1 - p)^z p^{m - z} \cdot (b - a) \right\}.
\]

If $k_{v_i} < k^*$, then $\pi_{i,j}(v_i) = 0$. Assume instead that $k_{v_i} \ge k^*$. In this case:
\begin{itemize}
    \item If there exists any agent $i' \in [n]$ such that $v_{i',j} = b$, then agent~$i$ does not receive item~$j$.
    \item If there exists $i' \in [n]$ such that $v_{i',j} = a$ but $k_{v_{i'}} > k_{v_i}$, then again agent~$i$ does not receive the item.
    \item Suppose there are exactly $z$ other agents $i'$ such that $v_{i',j} = a$ and $k_{v_{i'}} = k_{v_i}$, and the remaining $n - z - 1$ agents have $k_{v_{i''}} < k_{v_i}$. Then agent~$i$ receives item~$j$ with probability $1 / (z + 1)$.
\end{itemize}
The probability that any particular agent has $k_{v_i} = k$ given that his value for item $j$ is $v_{i,j}=a$ is $\Pr[B(m - 1, 1 - p) = k_{v_i}]$, while  the probability that any particular agent has $k_{v_i} < k$ given that his value for item $j$ is $v_{i,j}=a$ is $\Pr[B(m - 1, 1 - p) < k_{v_i}]$. Thus, the interim allocation is:
\begin{align*}
    \pi_{i,j}(v_i) &= p^{n - 1} \sum_{z = 0}^{n - 1} \frac{1}{z + 1} \binom{n - 1}{z} \left( \Pr[B(m - 1, 1 - p) = k_{v_i}] \right)^z \left( \Pr[B(m - 1, 1 - p) < k_{v_i}] \right)^{n - 1 - z} \\
    &= p^{n-1} \frac{\left(\Pr[B(m-1,1-p)<k_{v_i}]+\Pr[B(m-1,1-p)=k_{v_i}] \right)^n- \left( \Pr[B(m-1,1-p)<k_{v_i}] \right)^n}{n\Pr[B(m-1,1-p)=k_{v_i}]} \tag{\Cref{lemma: technical sum binom}} \\
    &=  p^{n-1} \frac{\left(\Pr[B(m-1,1-p)\le k_{v_i}] \right)^n- \left( \Pr[B(m-1,1-p)<k_{v_i}] \right)^n}{n\Pr[B(m-1,1-p)=k_{v_i}]}
\end{align*}

which depends only on $k_{v_i}$ and the fact that $v_{i,j} = a$. To simplify notation, we define:
\[
\pi(a, k) := 
\begin{cases}
p^{n-1} \frac{\left(\Pr[B(m-1,1-p)\le k_{v_i}] \right)^n- \left( \Pr[B(m-1,1-p)<k_{v_i}] \right)^n}{n\Pr[B(m-1,1-p)=k_{v_i}]} & \text{if } k \ge k^*, \\
0 & \text{if } k < k^*.
\end{cases}
\]

    Now consider any path $\ell \in \Pcal_{v_i}$. Each step along the path corresponds to a transition from one layer to the next, decreasing the number of $b$-valued items by one. Since the path starts at the source, passes through $v_i$ and each subsequent layer until it reaches layer $0$ and finally the sink, the total number of steps is given by $|\ell| = 3 + k_{v_i}$. Furthermore, for any two consecutive nodes on the path, $v_i^z$ and $v_i^{z+1}$, there exists exactly one item $j \in [m]$ such that $v_{i,j}^z = b$ and $v_{i,j}^{z+1} = a$. For all other items $j' \in [m] \setminus \{j\}$, we have $v_{i,j'}^z = v_{i,j'}^{z+1}$. Combining all these observations, we obtain:
    \begin{align*}
        p_i(v_i) &= \frac{1}{\Pr[v_i]}\sum_{\ell \in \Pcal_{v_i}} \xi_{\ell} \left(\sum_{j \in [m]} v_{i,j} \pi_{i,j}(v_i) - \sum_{z \in [1,|\ell|-3]}  \sum_{j \in [m]} \left( v_{i,j}^{z}-v_{i,j}^{z+1} \right)\pi_{i,j}(v_{i}^{z+1})\right) \\
        &=\frac{1}{\Pr[v_i]}\sum_{\ell \in \Pcal_{v_i}} \xi_{\ell} \left(\sum_{j \in [m]} v_{i,j} \pi_{i,j}(v_i) - \sum_{z \in [0,k_{v_i}-1]}  \left( b-a \right)\pi(a,z)\right)\\
        &= \sum_{j \in [m]} v_{i,j} \pi_{i,j}(v_i) - \sum_{z \in [0,k_{v_i}-1]}  \left( b-a \right)\pi(a,z) \\
        &= k_{v_i}b\pi(b)+ (m-k_{v_i})a\pi(a,k_{v_i})-\left( b-a \right)\sum_{z \in [0,k_{v_i}-1]}  \pi(a,z)\\
    \end{align*}
    Substituting the appropriate values for $\pi(b)$ and $\pi(a,k_{v_i})$ we get that the payment of \Cref{mech: firstAxis} is indeed induced by the flow described in $\Cref{def: firstAxis flow}$.  
\end{proof}

Now to establish the optimality of \Cref{mech: firstAxis}, we simply need to show that it is BIC-IR.
\begin{theorem}\label{thm: mech1 BIC}
    Consider a setting with $ n $ agents and $ m $ items, where each agent independently values each item at $ a $ with probability $ p $, and at $ b $ with probability $ 1 - p $. Then, \Cref{mech: firstAxis} is BIC-IR.
\end{theorem}

\begin{proof}
    We begin with the BIR (Bayesian Individual Rationality) part of the proof. Recall the following definitions established in the proof of \Cref{lemma: first axis mechanism induced by flow}:

\begin{itemize}
    \item $k^* := \arg \min_{k \in [0, m]} \left\{ a > \frac{1}{(1 - p)^k p^{m - k}} \cdot \frac{1}{(m - k)\binom{m}{k}} \cdot \sum_{z = k + 1}^{m} \binom{m}{z} (1 - p)^z p^{m - z} \cdot (b - a) \right\}$,
    \item $\pi(a, k) := 
    \begin{cases}
        p^{n-1} \frac{\left(\Pr[B(m-1,1-p)\le k_{v_i}] \right)^n- \left( \Pr[B(m-1,1-p)<k_{v_i}] \right)^n}{n\Pr[B(m-1,1-p)=k_{v_i}]}, & \text{if } k \ge k^*, \\
        0, & \text{if } k < k^*,
    \end{cases}$
    \item $\pi(b) := \sum\limits_{z = 0}^{n - 1} \frac{1}{z + 1} \Pr[B(n - 1, 1 - p) = z]$.
\end{itemize}
Now fix any agent $i \in [n]$ and a valuation $v_i \in \Vcal_i$. The expected utility of agent $i$ under truthful reporting is given by:
    \begin{align*}
        \mathbb{E}[u_i(v_i \rightarrow v_i)] &= \sum_{j \in [m]} v_{i,j} \pi_{i,j}(v_i) - p_i(v_i)\\
        &= k_{v_i}b\pi(b)+(m-k_{v_i})a\pi(a,k_{v_i})- \left(k_{v_i}b\pi(b)+(m-k_{v_i})a\pi(a,k_{v_i})-(b-a) \sum_{z \in [0,k_{v_i}-1]} \pi(a,z) \right) \\
        &= (b-a) \sum_{z \in [0,k_{v_i}-1]} \pi(a,z) \ge 0
    \end{align*}

The bulk of the proof is devoted to establishing that the mechanism is \emph{Bayesian Incentive Compatible (BIC)}. That is, for any agent $i \in [n]$ and any pair of valuations $v_i, v_i' \in \Vcal_i$, we must show $\mathbb{E}[u_i(v_i \rightarrow v_i)] \ge \mathbb{E}[u_i(v_i \rightarrow v_i')]$ where the expectation is taken over the types of the other agents and the randomness of the mechanism.

We analyze this by considering four cases, based on the relationship between $v_i$ and $v_i'$ in the flow graph. First, suppose that $k_{v_i} \ge k_{v_i'}$, i.e., $v_i'$ lies in the same layer or a lower layer than $v_i$ in the flow graph. Within this case, we consider two subcases.

Assume first that there exists a path $\ell \in \Pcal_{v_i}$ such that $v_i' \in \ell$, that is, $v_i'$ is reachable from $v_i$ in the flow graph. Then we have the following:
    \begin{align*}
         \mathbb{E}[u_i(v_i \rightarrow v'_i)] &= \sum_{j \in [m]} v_{i,j} \pi_{i,j}(v'_i) - p_i(v'_i)\\
         &= k_{v'_i}b\pi(b)+\left((m-k_{v_i})a + (k_{v_i}-k_{v'_i})b\right)\pi(a,k_{v_i'})\\ 
         & \phantom{=}- \left(k_{v_i'}b\pi(b)+(m-k_{v_i}')a\pi(a,k_{v_i}')-(b-a) \sum_{z \in [0,k_{v_i}'-1]} \pi(a,z) \right)\\
         &= (b-a) \sum_{z \in [0,k_{v_i}'-1]} \pi(a,z) +(k_{v_i}'-k_{v_i})a\pi(a,k_{v_i'}) + (k_{v_i}-k_{v'_i})b \pi(a,k_{v_i'})\\
         &= (b-a)\left(\sum_{z \in [0,k_{v_i'}-1]} \pi(a,z) +  (k_{v_i}-k_{v_i'}) \pi(a,k_{v_i'})\right)\\
         &= (b-a)\left(\sum_{z \in [0,k_{v_i'}-1]} \pi(a,z) +   \sum_{z \in [k_{v_i'},k_{v_i}-1]} \pi(a,k_{v_i'})\right)\\
         & \le (b-a)\left(\sum_{z \in [0,k_{v_i'}-1]} \pi(a,z) +   \sum_{z \in [k_{v_i'},k_{v_i}-1]} \pi(a,z)\right) \tag{if $x \le x'$ then $\pi(a,x) \le \pi(a,x')$ }\\
         &= (b-a)\left(\sum_{z \in [0,k_{v_i}-1]} \pi(a,z) \right)= \mathbb{E}[u_i(v_i \rightarrow v_i)]
    \end{align*}
    In the inequality, we use the fact that reporting a valuation of a higher $k$ results in a higher hierarchy and thus higher interim probability of being allocated the item. Using the same argument, we can also show that $\pi(b) \ge \pi(a,k)$ for all $k \in [0,m]$, which will become useful later on.

    Next, assume that $k_{v_i} \ge k_{v_i'}$ but there does not exist a path $\ell \in \Pcal_{v_i}$ such that $v_i' \in \ell$, that is, $v_i'$ is \emph{not} reachable from $v_i$ in the flow graph. This implies that $v_i'$ differs from $v_i$ in at least one coordinate $j_b \in [m]$ such that $v_{i,j_b}' = b$ while $v_{i,j_b} = a$. Since $k_{v_i} \ge k_{v_i'}$, there must also exist $j_a \in [m]$ such that $v_{i,j_a}' = a$ while $v_{i,j_a} = b$.Now consider a modified valuation $v_i'' \in \Vcal_i$ constructed as follows:
\[
v_{i,j}'' = 
\begin{cases}
v_{i,j}' & \text{for } j \in [m] \setminus \{j_a, j_b\}, \\
a & \text{for } j = j_b, \\
b & \text{for } j = j_a.
\end{cases}
\]
By construction, $v_i''$ and $v_i'$ differ only in a swap of $a$ and $b$ at positions $j_a$ and $j_b$, hence $k_{v_i''} = k_{v_i'}$. Thus,
    \begin{align*}
        \mathbb{E}[u_i(v_i \rightarrow v'_i)]- \mathbb{E}[u_i(v_i \rightarrow v''_i)] &= \sum_{j\in[m]}v_{i,j}(\pi_{i,j}(v'_{i,j})-\pi_{i,j}(v''_{i,j})) -(p_i(v_i')-p_i(v_i''))\\
        &= \sum_{j\in[m]}v_{i,j}(\pi_{i,j}(v'_{i,j})-\pi_{i,j}(v''_{i,j})) \tag{$p_i(v_i')=p_i(v_i'')$} \\
        &= v_{i,j_a}(\pi_{i,j_a}(v'_{i,j_a})-\pi_{i,j_a}(v''_{i,j_a})) + v_{i,j_b}(\pi_{i,j_b}(v'_{i,j_b})-\pi_{i,j_b}(v''_{i,j_b})) \\
        &= b(\pi(a,k_{v_i'})-\pi(b)) + a(\pi(b)-\pi(a,k_{v_i'})) \\
        &= (b-a)(\pi(a,k_{v_i'})-\pi(b)) \le 0 \tag{$\pi(a,k_{v_i'})\le\pi(b)$}
    \end{align*}
Let $\Bcal(v_i, \tilde{v}_i) = |\{j \in [m] : \tilde{v}_{i,j} = b \text{ and } v_{i,j} = a\}|$ denote the number of items whose value increases from $a$ in $v_i$ to $b$ in $\tilde{v}_i$. If $\Bcal(v_i, \tilde{v}_i) = 0$ and $k_{v_i} \ge k_{\tilde{v}_i}$, then by the definition of the flow, there exists a path $\ell \in \Pcal_{v_i}$ such that $\tilde{v}_i \in \ell$. In our case, $\Bcal(v_i, v_i') > 0$ and $\Bcal(v_i, v_i') - \Bcal(v_i, v_i'') = 1$. Therefore, we interpret $\Bcal(\cdot)$ as a measure of how far a valuation is from being reachable along a flow path from $v_i$ (when $k_{v_i} \ge k_{v_i'}$). By applying the above inequality iteratively, we can conclude that for any $v_i', v_i'' \in \Vcal_i$ with $k_{v_i'} = k_{v_i''}$ and $\Bcal(v_i, v_i') < \Bcal(v_i, v_i'')$, it holds that $\mathbb{E}[u_i(v_i \rightarrow v_i')] \ge \mathbb{E}[u_i(v_i \rightarrow v_i'')]$. Combining this with the fact that the agent has no incentive to misreport any valuation reachable from her true type completes the proof for this case. The remaining two cases follow by similar arguments.

    Now we will move to the case where $k_{v_i} < k_{v_i'}$. First, we will assume that there exists $\ell \in \Pcal_{v'_i}$ such that $v_i \in \ell$, that is, $v_i$ is reachable from $v'_i$ in the flow graph.
    \begin{align*}
         \mathbb{E}[u_i(v_i \rightarrow v'_i)] &= \sum_{j \in [m]} v_{i,j} \pi_{i,j}(v'_i) - p_i(v'_i)\\
         &= (k_{v_i}b+ (k_{v_i'}-k_{v_i})a)\pi(b)+(m-k_{v'_i})a\pi(a,k_{v_i'}) \\
         &\phantom{=}- \left(k_{v_i'}b\pi(b)+(m-k_{v_i}')a\pi(a,k_{v_i}')-(b-a) \sum_{z \in [0,k_{v_i}'-1]} \pi(a,z) \right)\\
         &= (b-a) \sum_{z \in [0,k_{v_i'}-1]} \pi(a,z)  - (k_{v_i'}-k_{v_i})(b-a) \pi(b)\\
         &= (b-a)\left(\sum_{z \in [0,k_{v_i'}-1]} \pi(a,z) -  (k_{v_i'}-k_{v_i}) \pi(b)\right)\\
         &= (b-a)\left(\sum_{z \in [0,k_{v_i'}-1]} \pi(a,z) -   \sum_{z \in [k_{v_i},k_{v'_i}-1]} \pi(b)\right)\\
         & \le (b-a)\left(\sum_{z \in [0,k_{v_i'}-1]} \pi(a,z) -   \sum_{z \in [k_{v_i},k_{v_i'}-1]} \pi(a,z)\right) \tag{$\pi(a,x) \le \pi(b)$ }\\
         &= (b-a)\left(\sum_{z \in [0,k_{v_i}-1]} \pi(a,z) \right)= \mathbb{E}[u_i(v_i \rightarrow v_i)]
    \end{align*}
    Finally, assume that there does not exist a path $\ell \in \Pcal_{v'_i}$ such that $v_i \in \ell$, while maintaining $k_{v_i} < k_{v_i'}$. This implies that there exists $j_a \in [m]$ such that $v_{i,j_a}' = a$ and $v_{i,j_a} = b$. Since $k_{v_i} < k_{v_i'}$, it must also be the case that there exists $j_b \in [m]$ such that $v_{i,j_b}' = b$ and $v_{i,j_b} = a$.  Now consider a valuation $v_i'' \in \Vcal_i$ such that for all $j \in [m] \setminus \{j_a, j_b\}$, we have $v_{i,j}'' = v_{i,j}'$, while $v_{i,j_b}'' = a$ and $v_{i,j_a}'' = b$. By construction, this implies that $k_{v_i''} = k_{v_i'}$. Thus:
    \begin{align*}
        \mathbb{E}[u_i(v_i \rightarrow v'_i)]- \mathbb{E}[u_i(v_i \rightarrow v''_i)] &= \sum_{j\in[m]}v_{i,j}(\pi_{i,j}(v'_{i,j})-\pi_{i,j}(v''_{i,j})) -(p_i(v_i')-p_i(v_i''))\\
        &= \sum_{j\in[m]}v_{i,j}(\pi_{i,j}(v'_{i,j})-\pi_{i,j}(v''_{i,j})) \tag{$p_i(v_i')=p_i(v_i'')$} \\
        &= v_{i,j_a}(\pi_{i,j_a}(v'_{i,j_a})-\pi_{i,j_a}(v''_{i,j_a})) + v_{i,j_b}(\pi_{i,j_b}(v'_{i,j_b})-\pi_{i,j_b}(v''_{i,j_b})) \\
        &= b(\pi(a,k_{v_i'})-\pi(b)) + a(\pi(b)-\pi(a,k_{v_i'})) \\
        &= (b-a)(\pi(a,k_{v_i'})-\pi(b)) \le 0 \tag{$\pi(a,k_{v_i'})\le\pi(b)$}
    \end{align*}
    Let $\Bcal(v_i, \tilde{v}_i) = \left| \left\{ j \in [m] : \tilde{v}_{i,j} = a \text{ and } v_{i,j} = b \right\} \right|$. If $\Bcal(v_i, \tilde{v}_i) = 0$ and $k_{v_i'} > k_{v_i}$, then by the definition of the flow, there exists a path $\ell \in \Pcal_{\tilde{v}_i}$ such that $v_i \in \ell$. In our current case, we have $\Bcal(v_i, v'_i) > 0$ and $\Bcal(v_i, v'_i) - \Bcal(v_i, v''_i) = 1$. This implies that $\Bcal(\cdot)$ serves as a measure of how far away the node $v_i$ is from being on a path rooted at $v'_i$ (under the assumption $k_{v_i} < k_{v_i'}$). By repeatedly applying the inequality implied by the above step, we can show that $\mathbb{E}\left[u_i(v_i \rightarrow \hat{v}'_i)\right] - \mathbb{E}\left[u_i(v_i \rightarrow \hat{v}''_i)\right] \le 0 $ for any $\hat{v}'_i, \hat{v}''_i \in \Vcal_i$ such that $k_{\hat{v}'_i} = k_{\hat{v}''_i} > k_{v_i}$ and $B(v_i, \hat{v}'_i) < B(v_i, \hat{v}''_i)$. Combining this with the fact that the agent has no incentive to misreport a value that can reach her true value completes the argument and proves the desired inequality.

\end{proof}

Finally, combining all the above results, we can derive the main result of this section:
\begin{theorem}\label{thm: revenue-first-axis}
    Consider a setting with $ n $ agents and $ m $ items, where each agent independently values each item at $ a $ with probability $ p $, and at $ b $ with probability $ 1 - p $. Then, \Cref{mech: firstAxis} is optimal and extracts, in expectation, revenue equal to:
    \begin{align*}
\mathbb{E}[\mathrm{Rev}] =\  m \bigg( & b \cdot (1 - p^n) \\
& + p^n \sum_{k=k^*}^{m-1} \bigg( \left(  \Pr[B(m-1,1-p) \le k]  \right)^n - \left(  \Pr[B(m-1,1-p) < k] \right)^n \bigg)
\\
& \quad \left. \cdot \left[ 
    a - \frac{1}{(1-p)^k p^{m-k}} \cdot \frac{1}{(m-k)\binom{m}{k}} 
    \sum_{z=k+1}^m \binom{m}{z} (1-p)^z p^{m-z} \cdot (b-a) 
    \right] 
\right)
\end{align*}
where $k^* = \arg \min_{k \in [0,m]}\left\{a > \frac{1}{(1-p)^k p^{m-k} } \cdot\frac{1}{(m-k)\binom{m}{k}}\cdot \sum_{z=k+1}^m\binom{m}{z}(1-p)^zp^{m-z} \cdot (b-a)\right\}$.
\end{theorem}
Notice that the revenue expression of the above theorem is a direct generalization of the revenue computed by Yao~\cite{YAO_BIC_DSIC} for the $m=2$ subcase.

\subsection{Second Axis} \label{sec: secondAxis}

In this section, we study the setting with $n$ agents with non-identical distributions (non-iid) and 2 items. Agent $i$ values each item independently at $a$ with probability $q_i$ and at $b$ with probability $1-q_i$. Without loss of generality, we assume that the probabilities are ordered $q_1\ge q_2\ge \ldots \ge q_n$. 

Similarly to the hierarchy mechanism in \Cref{sec: firstAxis}, the hierarchy allocation values $H_{i,j}$ are equal to the virtual values. By construction, the mechanism allocates each item to the agent with the highest non‐negative virtual value for that item, breaking ties uniformly at random. The allocation rule is therefore simple, but the payment rule is more intricate. That is because agents have different $q_i$'s, hence the virtual-value formula $H_{i,j}(v_i)$ is different for each agent. We show that the \Cref{mech: non-iid-agents} is optimal by showing, in \Cref{lemma: non-iid-agens-induced-mech}, that the flow in~\Cref{fig: flow-non-iid-agents} induces the mechanism. Then we show, in~\Cref{lemma: non-iid-agents-bic-conditions}, the necessary conditions the interim allocation probabilities must satisfy for the mechanism to be BIC. Finally, we show, in~\Cref{lemma: non-iid-agents-interim} that the probabilities that are induced by the hierarchy allocation satisfy the conditions given by the previous lemma. 

\begin{algorithm}[ht]
\caption{Optimal mechanism for non-iid agents}\label{mech: non-iid-agents}
\begin{algorithmic}
\State For any agents $i \in [n]$, with valuation $v_i\in\Vcal_i$, if $v_{i,j}=b$ then $H_{i,j}(v_i)=b$, else if $v_{i,j}=a$ then 
\[H_{i,j}(v_i) = a-\frac{1-q_i}{2q_i}(b-a)\quad \text{if }v_{i,-j}=b,\quad \& 
\quad H_{i,j}(v_i)=a-\frac{1-q_i^2}{2q_i^2}(b-a) \quad\text{if }v_{i,-j}=a\]

The payment of agent $i$ for each of $\{(b,b),(b,a),(a,b),(a,a)\}$ is
\begin{align*}
    p(b,b) &=2b\pi_i(b) - (b-a)(\pi_i(a,b) +  \pi_i(a,a))\\
    p(a,b) &= p(b,a) = b\pi_i(b) + a \pi_i(a,b) - (b-a)\pi_i(a,a)\\
    p(a,a) &=2a\pi_i(a)
\end{align*}

\State Let the probability agent \(i\) wins her item when she reports \(b\) be 

\[\pi_i(b) = \sum_{z = 1 }^{n} \frac{1}{z} \cdot \sum_{\substack{S \subseteq [n]\setminus\{i\} \\[4pt]|S| = z-1}} \;\prod_{k\in S} (1-q_k) \prod_{k \in [n]\setminus (S\cup \{i\})} q_k\]

\State The probability agent $i$ wins the item when she reports $a$ has the following three cases
     \If{$q_i \ge \sqrt{\frac{b-a}{a+b}}$ and $S_i^1=S_i^2=\emptyset$}
     \State\begin{equation*}
         \pi_{i}^{(1)}(a,b) =  \pi_{i}^{(1)}(b,a)=  \prod_{k\in [i]\setminus\{i\}} q_k \quad \& \quad \pi_i^{(1)}(a)=  \prod_{k\in S_i^3} q_k^2 \prod_{\ell \in S_i^4} q_\ell
     \end{equation*}
     \ElsIf {$\sqrt{\frac{b-a}{a+b}}>q_i \ge \frac{b-a}{a+b}$  or ($S_i^1=\emptyset$ and $ S_i^2 \neq \emptyset$) }
    \State\begin{equation*}
         \pi_{i}^{(2)}(a,b) =\pi_{i}^{(2)}(b,a) = \prod_{k\in S_i^2} q_k^2 \prod_{\ell \notin S_i^2} q_\ell  \qquad \& \qquad \pi_{i}^{(2)}(a) =0 
    \end{equation*}
     \Else{ $q_i < \frac{b-a}{a+b} $ or $S_i^1 \neq \emptyset$}
    \State \begin{equation*}
        \pi_{i}^{(3)}(a,b) =\pi_{i}^{(3)}(b,a) = \pi_{i}^{(3)}(a) =0 
    \end{equation*}
    \EndIf

 \State Where we define four disjoint sets of other players by comparing their $q_i$ to thresholds $\sqrt{q_k}$ and $q_k$:

\[
\begin{aligned}
S_i^1 &= \{\,k\neq i: q_k>\sqrt{q_i}\},&
S_i^2 &= \{\,k\neq i: q_i<q_k\le\sqrt{q_i}\},\\
S_i^3 &= \{\,k\neq i: q_i^2<q_k\le q_i\},&
S_i^4 &= \{\,k\neq i: q_k\le q_i^2\}.
\end{aligned}
\]

\end{algorithmic}
\end{algorithm}

For each agent, the flow is described by~\Cref{fig: flow-non-iid-agents}. Notably, the optimal flow for each agent is qualitatively identical, just parameterized by $q_i$. Thus, adding more bidders of various types does not introduce a fundamentally new structure into the mechanism.

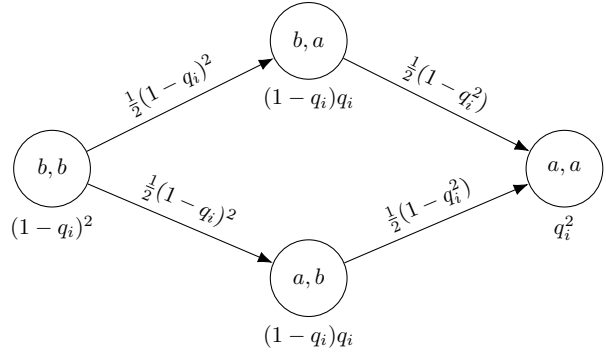
\begin{figure}[ht]
    \centering
    \begin{tikzpicture}[scale=0.85,  transform shape, node distance=2.5cm and 2.8cm,
    vertex/.style={draw, circle, minimum size=1.2cm},
    edge/.style={draw, -{Latex[length=2mm]}},
    ]

\node[vertex,label=below:{$(1-q_i)^2$}] (bb) {$b,b$};
\node[vertex, right=of bb, yshift=2cm, label=below:{$(1-q_i)q_i$}] (ba) {$b,a$};
\node[vertex, below=of ba,label=below:{$(1-q_i)q_i$}] (ab) {$a,b$};
\node[vertex, right=of bb, xshift=4cm,label=below:{$q_i^2$}] (aa) {$a,a$};

\path[edge] (ab) edge node[midway,sloped,above] {$\frac{1}{2}(1-q_i^2)$} (aa);
\path[edge] (ba) edge node[midway,sloped,above] {$\frac{1}{2}(1-q_i^2)$} (aa);
\path[edge] (bb) edge node[midway,sloped,above] {$\frac{1}{2}(1-q_i)^2$} (ba);
\path[edge] (bb) edge node[midway,sloped,above] {$\frac{1}{2}(1-q_i)^2$} (ab);

\end{tikzpicture}
    \caption{The Figure depicts the flow graph for a single agent $i$ with parameter $q_i$. The four nodes $(b,b), (b,a), (a,b), (a,a)$ represent the agent’s possible valuations for the two items. We omit the flow from $(a,a)\to\bot$, that is $\mu(a,a)=1$ }
    \label{fig: flow-non-iid-agents}
\end{figure}

\begin{lemma}\label{lemma: non-iid-agens-induced-mech}
The flow defined in~\Cref{fig: flow-non-iid-agents} induces the hierarchy allocation function and interim allocation probabilities presented in~\Cref{mech: non-iid-agents}.   
\end{lemma}

\begin{proof}
    We will break the proof of the lemma into two steps. First, we show that for each $i\in[n],j\in\{1,2\}$ and $v_i \in \{(b,b),(b,a),(a,b),(a,a)\}$, the hierarchy allocation function $H_{i,j}$ is equal to the virtual value, $H_{i,j}(v_i) = v_{i,j} - \frac{1}{\Pr[v_i]}\sum_{v'_i \in \mathcal{V}_i} \lambda_i(v'_i, v_i)(v'_{i,j} - v_{i,j})$. Then, we will calculate the interim allocation rule $\pi_{i,j}(\cdot)$.

    Let us first focus on the hierarchy allocation function $H_{i,j}(\cdot)$. Notice that in the flow, shown in~\Cref{fig: flow-non-iid-agents}, $\lambda_i(v_i',v_i)$ is strictly positive if $v_i'$ and $v_i$ differ in exactly one coordinate, which is $b$ in $v_i'$ and $a$ in $v_i$. Given this observation, we can see that if $v_{i,j}=b$ then $H_{i,j}=b$ since $\lambda_i(v_{i}',v_{i})(v_{i,j}'-v_{i,j})=0$ either because $v_{i,j}'=v_{i,j}=b$ or $\lambda_i(v_{i}',v_{i})=0$. On the other hand, if $v_{i,j}=a$ there exists exactly one node such that $v_{i,j}'=b$ such that $\lambda_i(v_{i}',v_{i})(v_{i,j}'-v_{i,j})>0$. Therefore,
    \begin{align*}
        H_{i,1}(a,b)=H_{i,2}(b,a) &=  v_{i,j} - \frac{1}{\Pr[v_i]}\sum_{v'_i \in \mathcal{V}_i} \lambda_i(v'_i, v_i)(v'_{i,j} - v_{i,j}) = a - \frac{1-q_i}{2q_i}(b-a)\\
        H_{i,1}(a,a)=H_{i,2}(a,a) &=  v_{i,j} - \frac{1}{\Pr[v_i]}\sum_{v'_i \in \mathcal{V}_i} \lambda_i(v'_i, v_i)(v'_{i,j} - v_{i,j}) =  a - \frac{1-q_i^2}{2q_i^2}(b-a)
    \end{align*}

    Our next step is to calculate the interim allocation rule assuming truthful bidding. Recall that the mechanism allocates each item to the agent with the highest non‐negative virtual value for that item, breaking ties uniformly at random. If agent $i$ reports $b$ for an item, then she has the highest virtual value for that item, and hence she splits the item with all the other agents that reported $b$ for that item. Therefore, probability she gets an item when she $b$ is $\pi_{i,j}(v_i) =\sum_{z = 1 }^{n} \frac{1}{z} \cdot \sum_{S\subseteq[n]\setminus \{i\},|S|=z-1} \prod_{k\in S} (1-q_k) \cdot \prod_{k \in [n]\setminus (S\cup \{i\})} q_k$. Intuitively, this probability captures all the subsets of size $z-1$ where all of those reported $b$ and the rest $a$. Note that this probability does not depend on the item $j$ or the valuation of the other item. For ease of notation, we will refer to this probability as $\pi_i(b)$. Now consider the case where agent $i$ reports $a$ for item $1$; by symmetry, the same analysis holds for item 2. Agent $i$ gets the item if, and only if, she has the highest non-negative virtual value for that item. Recall that if there exists an agent that reported $b$ for item $1$, then agent $i$ will not receive that item. One can easily verify the crucial points of the virtual values, 

\[
\begin{cases}
  q_i \ge \sqrt{\frac{b-a}{b+a}}
  &\implies  0 \le H_{i,1}(a,a) \le H_{i,1}(a,b),\\[6pt]
  \frac{b-a}{b+a}\le q_i \le \sqrt{\frac{b-a}{b+a}}
  &\implies H_{i,1}(a,a) \le 0 \le H_{i,1}(a,b),\\[6pt]
  q_i \le \frac{b-a}{b+a}
  &\implies  H_{i,1}(a,a) \le H_{i,1}(a,b) \le 0.
\end{cases}
\]

Observe the monotonicity property of the agent’s virtual value. That is, whenever $q_i \ge q_k$, we have $H_{i,1}(a,b)\ge H_{k,1}(a,b)$, and $H_{i,1}(a,a)\ge H_{k,1}(a,a)$. Moreover, if $q_i\ge \sqrt{q_k}$ then $H_{i,1}(a,a)\ge H_{k,1}(a,b).$

To determine whether agent $i$ has the highest virtual value among the agents who reported $a$ we partition the other agents into four disjoint sets, $S_i^1, S_i^2, S_i^3, S_i^4$, based on the relationship of their virtual values $H_{k,1}(\cdot,\cdot)$ to agent $i$'s virtual values $H_{i,1}(\cdot,\cdot)$.

Recall that each item is given to the agent with the highest virtual value. Agent $i$ is in the third case if all the virtual values are negative or $S_i^1 = \{k\in[n]\setminus\{i\}: q_k \in (\sqrt{q_i},1]\}$ is nonempty. The fact that $S_i^1$ is nonempty implies that there is an agent $k$ such that $H_{k,1}(a,a)>H_{i,1}(a,b)$, thus even if $H_{i,1}(a,b)$ is positive, there is an agent with a higher virtual value no matter what they report. Therefore, $\pi_{i,j}^{(3)}(v_i)=0$ when $v_{i,j}=a.$ Agent $i$ is in the second case if only the virtual values of $(a,a)$ are negative or if $S_i^1$ is empty and $S_i^2 = \{k\in[n]\setminus\{i\}: q_k \in (q_i,\sqrt{q_i}]\}$ is nonempty. The fact that the virtual values of $(a,a)$ are negative implies that $\pi_i^{(2)}=0$. The set $S_i^2$ contains all the agents $k$ that $H_{k,1}(a,a)\le H_{i,1}(a,b)$ and $H_{k,1}(a,b)>H_{i,1}(a,b)$. So when she reports $(a,b)$ or $(b,a)$, all agents in $S_i^2$ must report $(a,a)$ and the rest $a$ only for that item. This gives us the interim allocation probability $\pi_{i,j}^{(2)}(v_{i,j}) = \prod_{k\in S_i^2} q_k^2 \prod_{\ell \notin S_i^2} q_\ell$, when $v_{i,j}=a$. Finally, agent $i$ is in the first case if all the virtual values are positive and $S_i^1$ and $S_i^2$ are empty. Since $S_i^1$ and $S_i^2$ are empty implies that $H_i(a,b) > H_k(a,b)$ for all $k$; therefore, she has the highest virtual value when all other players report $a$ for that item. Hence, the probability $\pi_{i}^{(1)}(a,b) = \pi_{i}^{(2)}(b,a) =\prod_{k\in [n]\setminus\{i\}}.$ On the other hand, when she reports $(a,a)$ she takes the item if all agents in $S_i^3 = \{\,k\neq i: q_i^2<q_k\le q_i\}$ also report $a,a$, since $H_k(a,b) > H_i(a,a) > H_k(a,a)$ for $k\in S_i^3,$ and the rest report $a$ for the same item.  Hence $\pi_i^{(1)}(a)=  \prod_{k\in S_i^3} q_k^2 \prod_{\ell \notin S_i^4} q_\ell.$

\end{proof}

\begin{lemma}
    The payment rule induced by~\Cref{mech: non-iid-agents} is 
    \begin{align*}
    p(b,b) &=2b\pi_i(b) - (b-a)(\pi_i(a,b) +  \pi_i(a,a))\\
    p(a,b) &= p(b,a) = b\pi_i(b) + a \pi_i(a,b) - (b-a)\pi_i(a,a)\\
    p(a,a) &=2a\pi_i(a)
\end{align*}
\end{lemma}

\begin{proof}

To prove that the payment rule shown in~\Cref{mech: non-iid-agents} is induced by the flow in~\Cref{fig: flow-non-iid-agents}, we must for the flow decomposition $\xi$ defined on the set of simple paths $\Pcal_{v_i}$ for all $v_i$, the payment for agent $i$ is $p_i(v_i) = \frac{1}{\Pr[v_i]}\sum_{\ell \in \Pcal_{v_i}} \xi_{\ell} \left(\sum_{j \in [2]} v_{i,j} \pi_{i,j}(v_i) - \sum_{z \in [|\ell|-3]}  \sum_{j \in [m]} \left( v_{i,j}^{z}-v_{i,j}^{z+1} \right)\pi_{i,j}(v_{i}^{z+1})\right).$ Starting with the payment of $v_i=(a,a)$, we can easily see that there is only one path from $v_i$ to $\bot$ with flow $q_i^2$: 
\begin{align*}
    p_i(a,a) &=\frac{1}{q_i^2} q_i^2 \left(2a\pi_i(a) \right)=2a\pi_i(a) 
\end{align*}
Following similar logic, we can show that $p_i(a,b)=p_i(b,a)$, since there is only one path from $(a,b)$ ($(b,a)$ respectively) to $\bot$ with flow $q_i(1-q_i)$: 
\begin{align*}
    p_i(b,a) = p_i(a,b) &= \frac{1}{q_i(1-q_i)} q_i(1-q_i) \left(b\pi_i(b) + a \pi_i(a,b) - (b-a)\pi_i(a)\right)\\
    &=b\pi_i(b) + a \pi_i(a,b) - (b-a)\pi_i(a)
\end{align*}
Finally, the payment for $(b,b)$ can be calculated from the two simple paths that exist between $b,b$ and $\bot$ with flow $1/2(1-q_i)^2$, that is $\{s,(b,b),(a,b),(a,a),\bot\},\{s,(b,b),(b,a),(a,a),\bot\}$: 
\begin{align*}
    p_i(b,b) &= \frac{1}{\Pr[v_i]}\sum_{\ell \in \Pcal_{v_i}} \xi_{\ell} \left(\sum_{j \in \{1,2\}} v_{i,j} \pi_{i,j}(v_i) - \sum_{z \in [1,|\ell|-3]}  \sum_{j \in [m]} \left( v_{i,j}^{z}-v_{i,j}^{z+1} \right)\pi_{i,j}(v_{i}^{z+1})\right) \\
    &=\frac{1}{(1-q_i)^2}\sum_{\ell \in \Pcal_{(b,b)}} \frac{1}{2}(1-q_i)^2 \left(2b\pi_i(b) - (b-a)\pi_i(a,b) - (b-a) \pi_i(a)\right)\\
    &= 2b\pi_i(b) - (b-a)(\pi_i(a,b) +  \pi_i(a))
\end{align*}
\end{proof}

\begin{lemma}
\label{lemma: non-iid-agents-bic-conditions}
Under the payments of \Cref{lemma: non-iid-agens-induced-mech}, the mechanism in
\Cref{mech: non-iid-agents} is BIR. The mechanism is BIC provided that for every agent $i$
\[
\pi_i(a)\;\le\;\pi_i(a,b)\;\le\;\pi_i(b)
\]
\end{lemma}

\begin{theorem}\label{lemma: non-iid-agents-interim} \Cref{mech: non-iid-agents} is BIC, satisfying the necessary conditions from \Cref{lemma: non-iid-agents-bic-conditions}.  
\end{theorem}

\begin{proof}
Recall that \Cref{mech: non-iid-agents} has three payment and allocation rules depending on the probability $q_i$ of agent $i$. To prove that the mechanism is BIC, we must show that $\pi_i(a)\le \pi_i(a,b)\le\pi_i(b)$. We start from the trivial case where agent $i$ is in the third case since it is easy to see that $\pi_i(b)\ge  \pi_i^{(3)}(a,b) = \pi_i^{(3)}(a) = 0.$ 

When agent $i$ is in the second case, it suffices to show that $\pi_i(b)\ge  \pi_i^{(2)}(a,b)$ since $\pi_i^{(2)}(a) = 0.$ Recall from \Cref{lemma: non-iid-agens-induced-mech} that when she reports $b$ she takes that item with probability $\pi_i(b)=\sum_{z = 1 }^{n} \frac{1}{z} \cdot \sum_{S \subseteq [n]\setminus\{i\}, |S| = z-1} \;\prod_{k\in S} (1-q_k) \prod_{k \in [n]\setminus (S\cup \{i\})} q_k$, and when she reports $(a,b)$ or $(b,a)$ she gets the item she reported $a$ with probability $\pi_i^{(2)}(a,b)=\pi_i^{(2)}(b,a) = \prod_{k\in S_i^2} q_k^2 \prod_{\ell \notin S_i^2} q_\ell$. Rearrange $\prod_{k\in S_i^2} q_k \prod_{\ell \in [n]\setminus\{i\}} q_\ell$. Notice that the event when everyone reports $a$ for one item ($\prod_{\ell \in [n]\setminus\{i\}} q_\ell$) is included as an event in $\pi_i(b)$ for $z=n$. Multiplying by $\prod_{k\in S_i^2} q_k \le 1$ makes the probability even smaller. Therefore, $\pi_i(b) \ge \pi_i^{(2)}(a,b)$. 

Finally, when agent $i$ is in the first case and all the interim allocation probabilities are non-negative. From \Cref{lemma: non-iid-agens-induced-mech}, recall that $\pi_i(b)$ is the same as in the previous case, $\pi_i^{(1)}(a,b) = \pi_i^{(1)}(b,a) =\prod_{k\in [n]\setminus\{i\}} q_k $, and $\pi_i^{(1)}(a) =  \prod_{k\in S_i^3} q_k^2 \prod_{\ell \notin S_i^4} q_\ell $. Using similar arguments as in the previous case, we can see that $\pi_i(b)\ge \pi_i(a,b)$ since the event where everyone reports $a$ is included in $\pi_i(b).$ Notice that $ \pi_i^{(1)}(a) = \prod_{k\in S_i^3} q_k^2 \prod_{\ell \in S_i^4} q_\ell =  \prod_{k\in S_i^3} q_k \cdot \pi_i(a,b) $. Therefore, $\pi_i^{(1)}(b,a)\ge \pi_i^{(1)}(a)$ which concludes the proof.
\end{proof}

\subsection{Third Axis} \label{sec: thirdAxis}

In this section, we consider $n$ i.i.d. bidders and two items, where each bidder’s valuation for item~1 is $a$ with probability $p$ and $b$ with probability $1-p$, and independently their valuation for item~2 is $a$ with probability $q$ and $b$ with probability $1-q$, where $b>a.$ We show that the optimal mechanism is induced by the simple flow shown in \Cref{fig: flow-non-iid-items-split}, where $x$ depends on the values of $a,b,p,$ and $q$. These values split the optimal mechanism into seven regions. In this setting, since the agents are iid, we remove the subscript $i$, and the set of profiles for all agents is $\Vcal = \{(b,b),(b,a),(a,b),(a,a)\}$. By symmetry, we focus on the case where $p\ge q.$ We start by showing that the virtual values presented in table~\ref{tab: virtual-values-non-iid-items} are induced by the flow shown in~\Cref{fig: flow-non-iid-items-split}.
\begin{figure}[ht]
  \centering
  \begin{subfigure}{0.45\textwidth}
    \centering
    \begin{tikzpicture}[scale=0.7,  transform shape,  node distance=2.5cm and 2.5cm,
    vertex/.style={draw, circle, minimum size=1.2cm},
    edge/.style={draw, -{Latex[length=2mm]}},
    ]

\node[vertex, label=below:{$(1-p)q$}] (ba) {$b,a$};
\node[vertex, yshift=-4cm, label=below:{$p(1-q)$}] (ab) {$a,b$};
\node[vertex, left=of ab,yshift=2cm, label=below:{$(1-p)(1-q)$}] (bb) {$b,b$};

\node[vertex, right=of ab, yshift=2cm, label=below:{$pq$}] (aa) {$a,a$};

  \path[edge] (ab) edge
    node[midway,sloped,above] {$p(1-q)+x$}
    (aa);

  \path[edge] (ba) edge
    node[midway,sloped,above] {$1-p-x$}
    (aa);

  \path[edge] (bb) edge
    node[midway,sloped,above] {$(1-p)(1-q)-x$}
    (ba);

  \path[edge] (bb) edge
    node[midway,sloped,above] {$x$}
    (ab);

\end{tikzpicture}
    \caption{Node $(b,b)$ splits its flow between $(a,b)$ and $(b,a)$.}
    \label{fig: flow-non-iid-items-split}
  \end{subfigure}%
  ~ 
  \begin{subfigure}{0.55\textwidth}
  \centering
        \begin{tabular}{|c|c|c|}
    \hline
     & Item 1 & Item 2 \\ \hline
    $(b,b)$ & $b$ & $b$ \\
    $(b,a)$ & $b$ & $a - \tfrac{(1-p)(1-q)-x}{(1-p)q}(b-a)$ \\
    $(a,b)$ & $a - \tfrac{x}{p(1-q)}(b-a)$ & $b$ \\
    $(a,a)$ & $a - \tfrac{1-p-x}{pq}(b-a)$ 
          & $a - \tfrac{p(1-q)+x}{pq}(b-a)$ \\ \hline
  \end{tabular}
  \vspace{0.7cm}
  \caption{Virtual values parametrized by  $x.$}
  \label{tab: virtual-values-non-iid-items}
  \end{subfigure}
  \caption{Dual flow parametrized by $x\in[0,(1-p)(1-q)]$ and the corresponding virtual values.}
  \label{fig: flows-and-virtual-values}
\end{figure}

\begin{lemma}
    The virtual values shown in Table~\ref{tab: virtual-values-non-iid-items} are induced by the parametrized flow shown in~\Cref{fig: flow-non-iid-items-split}.
\end{lemma}

\begin{proof}
    Recall that the virtual values are given by 
    \[H_j(v) = v_j - \frac{1}{\Pr[v]}\sum_{v'\in\Vcal}\lambda(v',v)(v'_j-v_j)\]
    
    First, we consider the case where $v_j=b$. Then for any $v'\in \Vcal,$ either $v'_j=b$ or $\lambda(v',v)=0$ by construction of the flow. Hence, all terms in the sum are zero, making $H_j(v)=b.$ Now consider the case where $v_j=a$. Here exactly one node $v'$ such that $\lambda(v',v)>0$ and $v'_j-v_j=b-a$. We can easily confirm that on the path $(b,b)\to(b,a)$ with flow $(1-p)(1-q)-x$, the virtual value becomes $H_2(b,a)=a-\frac{(1-p)(1-q)-x}{(1-p)q}(b-a)$. For the path $(b,b)\to(a,b)$ with flow $x$ the virtual value becomes $H_1(a,b) = a -\frac{x}{p(1-q)}(b-a)$. Similarly, for the path $(b,a)\to(a,a)$ with flow $1-p-x$ the virtual value is $H_1(a,a) = a \frac{1-p-x}{pq}(b-a).$ Finally, for the path $(a,b)\to(a,a)$ with flow $p(1-q)+x$ the virtual value is $H_2(a,a)=a-\frac{p(1-q)+x}{pq}(b-a).$   
\end{proof}

Next, the following technical lemma states that the virtual values are monotone for each item. The proof is shown in~\Cref{apx: proofs-thr-axis}.
\begin{lemma}\label{lemma: non-iid-items-monotone}
For the virtual values shown in~\Cref{tab: virtual-values-non-iid-items}, for all $x\in[0,(1-p)(1-q)],$ 
    \[H_1(a,b)\ge H_1(a,a)\quad \&\quad H_2(b,a) \ge H_2(a,a) \] 
\end{lemma}


In this setting, the number of agents $n,$ the probabilities $p,q$, and the valuations $a,b$. dictate the optimal mechanism. We show that there are seven distinct regions where the optimal mechanism slightly varies. Before we dive into the different regions, we calculate the payments that are induced by the parametrized flow shown in~\Cref{fig: flow-non-iid-items-split}. Recall that the technique we use to deduce the payments is by decomposing the flow into its simple components. Since the node $(b,b)$ is the only node that splits its incoming flow, we get that only the payment of $(b,b)$ depends on $x$, all other payments have a closed form. We deferred the proof to~\Cref{apx: proofs-thr-axis}.

\begin{lemma}\label{lemma: non-iid-items-payments}
The payment rule induced by the parametrized flow shown in~\Cref{fig: flow-non-iid-items-split} is

\begin{align*}
    p(b,b) &= b(\pbb + \qbb) - (b - a)(\qab + \paa)\\
    &\qquad\qquad\qquad+\frac{x}{(1 - p)(1 - q)} (b - a) \left( \qab - \pab + \paa - \qaa \right)\\
    p(b,a) &= a \qab + b \pbb - (b-a) \paa \\
    p(a,b) &= a \pab + b \qbb - (b-a)\qaa\\
    p(a,a) &= a \cdot \left( \paa + \qaa \right)
\end{align*}

\end{lemma}

Next, we show the necessary conditions the interim probabilities must satisfy in order for the mechanism to be BIC. From complementary slackness, we know that if a dual variable is strictly positive, then the corresponding constraint in the primal is tight. In our case, when a flow variable $\lambda(v,v')>0$ then the BIC constraint holds with equality, $\Ex{v\to v} = \Ex{v\to v'}$. This is particularly useful when the parameter of the flow $x<(1-p)(1-q)$ since $\lambda((b,b),(b,a))>0$ and thus $\Ex{(b,b)\to(b,b)} = \Ex{(b,b)\to(b,a)}$. We can easily calculate $\Ex{(b,b)\to(b,b)} = (b - a)(\qab + \paa)  - \frac{x}{(1 - p)(1 - q)} (b - a) \left( \qab - \pab + \paa - \qaa \right).$  Similarly, $\Ex{(b,b)(b,a)} = b(\pbb+\qab) - p(b,a) =(b - a)(\qab + \paa)$. Therefore,  for any $x<(1-p)(1-q)$ we get that 
\begin{equation}\label{eq: BIC-variantI-eq}
    \qab - \pab + \paa - \qaa =0
\end{equation}
From the previous observation, we can define two variants of payment in the case of $(b,b)$. Variant~$I$, when $x<(1-p)(1-q)$ and $\qab - \pab + \paa - \qaa = 0$, the payment is reduced to 
\begin{equation}
    p_1(b,b) = b(\pbb+\qbb)-(b-a)(\pab+\qaa).\label{eq: non-iid-items-p1(b,b)}
\end{equation}
Variant~$II$, when $x=(1-p)(1-q)$ and the BIC constraint is not tight, we get the payment of $(b,b)$ is 
\begin{equation}
p_2(b,b)= b(\pbb+\qbb) -(b-a)(\pab+\qaa).\label{eq: non-iid-items-p2(b,b)}
\end{equation}

This leads us to the following lemma that gives us the necessary conditions for the mechanism to be BIC in each variant.  We rely heavily on this to show that our mechanism is BIC in all the regions. 

\begin{lemma}\label{lemma: non-iid-items-BIC}
The mechanism is Bayesian individual rational (BIR). The mechanism is Bayesian incentive‐compatible (BIC) if, and only if, 
\begin{itemize}
    \item Variant $I$ $(\qab - \pab + \paa - \qaa = 0)$:
    \begin{align}    
\label{eq: BIC-variantI}
\qbb\ge\qab&\ge\qaa, \nonumber \\
\qab + \paa &\ge \pab +\qaa, \\
\pbb+\qaa &\ge \qab + \paa\nonumber,\\
\pbb &\ge \paa.\nonumber
\end{align}
\item Variant $II$ $(x=(1-p)(1-q))$:
    \begin{align}
\label{eq: BIC-variantII}
\pbb\ge\pab&\ge\paa, \nonumber \\
\pab + \qaa &\ge \qab +\paa,\nonumber \\
\qbb+\paa &\ge \pab + \qaa,\\
\qbb &\ge \qaa.\nonumber
\end{align}
\end{itemize}

\end{lemma}

Now that we have presented the parameterized mechanism and the necessary BIC conditions, we define the interim allocation probabilities when the virtual values are positive and distinct. That is $H_1(a,b)>H_1(a,a)>0$ and $H_2(b,a)>H_2(a,a)>0$. Recall that the hierarchical allocation mechanism awards an item to the agent with the highest non-negative virtual value, breaking ties uniformly at random.  When the virtual value is zero, the mechanism flips a biased coin (probability $\delta$ of ``heads"), and only if it comes up heads does it allocate uniformly at random. This reduced-probability trick is essential for the optimal mechanism, as we will see in the region-by-region analyses.

\vspace{5pt}
\noindent\fbox{%
    \parbox{\textwidth}{%
    \smallskip
        For the first item, if she reports $b$, then
        \begin{equation}\label{eq: pbb}
            \pbb=\sum_{\ell=1}^n \frac{1}{\ell}\binom{n-1}{\ell-1}(1-p)^{\ell-1}p^{n-\ell}= \frac{1 - p^n}{n(1-p)},
        \end{equation} 
        where the last equality follows from \Cref{lemma: technical sum binom}. By symmetry, if she reports $b$ for the second item,
        \begin{equation}\label{eq: qbb}
            \qbb = \sum_{\ell=1}^n \frac{1}{\ell}\binom{n-1}{\ell-1}(1-q)^{\ell-1}q^{n-\ell}=\frac{1 - q^n}{n(1-q)}.
        \end{equation}
        When an agent reports $(a,b)$, she shares the first item with everyone who reported $(a,b)$, given that the remaining agents have reported $(a,a)$. This gives us the following 
        \begin{equation}\label{eq: pab}
            \pab = \sum_{\ell=1}^n \frac{1}{\ell}\binom{n-1}{\ell-1} (p(1-q))^{\ell-1}(pq)^{n-\ell} = p^{n-1}\qbb.
        \end{equation}
        Likewise reporting $(b,a)$, yields
         \begin{equation}\label{eq: qab}
             \qab = \sum_{\ell=1}^n \frac{1}{\ell}\binom{n-1}{\ell-1} ((1-p)q)^{\ell-1}(pq)^{n-\ell} = q^{n-1}\pbb.
         \end{equation}
         Finally, reporting $(a,a)$ wins only when all other agents have also reported $(a,a)$, that is 
         \begin{equation}\label{eq: paa}
             \paa=\qaa= \frac{1}{n}(pq)^{n-1}.
         \end{equation}
         
    }%
}
\vspace{5pt}

 It is important to note that the interim probabilities are different in each region depending on the sign of the virtual value. However, these are going to be crucial for the analysis of the regions. Before we proceed, we show the following technical lemma about the interim probabilities. 

\begin{lemma}\label{lemma: non-iid-items-probabilities-monotonicity}
Under the definitions in 
\Cref{eq: pbb,eq: qbb,eq: pab,eq: qab,eq: paa}, if $p\ge q$ then
\[
  \pbb\;\ge\;\qbb,
  \quad
  \pab\;\ge\;\qab,\]
\[\pbb\;\ge\;\pab\;\ge\;\paa,
  \quad
  \qbb\;\ge\;\qab\;\ge\;\qaa,
\]
\end{lemma}

Next, we provide a road map of the regions. Concretely, we partition the $(p,q)$-parameter space into seven regions based on the sign of flow-induced virtual values.  In~\Cref{tab:thresholds} we show the four critical $x$-values at which each $H_j(v)$ crosses zero. This table will help us with the transition from one region to the next. Then, the bullet points below summarize each region's definition and the corresponding value of the flow.

\begin{table}[ht]
\centering
\begin{tabular}{|c|c|}
\hline
Virtual‐value nonnegativity & Threshold on $x$ \\ 
\hline
$H_1(a,b)\ge0$ & $x\le \tfrac{a}{b-a}\,p(1-q)$ \\[3pt]
$H_2(b,a)\ge0$ & $x\ge (1-p)(1-q)-\tfrac{a}{b-a}\,(1-p)q$ \\[3pt]
$H_1(a,a)\ge0$ & $x\ge 1-p-\tfrac{a}{b-a}\,p\,q$ \\[3pt]
$H_2(a,a)\ge0$ & $x\le \tfrac{a}{b-a}\,p\;q-p(1-q)$ \\
\hline
\end{tabular}
\caption{Threshold values of $x$ at which each $H_j(v)$ becomes nonnegative.}
\label{tab:thresholds}
\end{table}

\begin{itemize}
\item \textbf{Region 1 (red):} the virtual values of items valued at $a$ are negative. 

\item \textbf{Region 2 (green):} only the virtual values of $(a,a)$ are negative, that is $H_1(a,a),H_2(a,a)<0$. For the mechanism to be BIC, $H_2(b,a) \ge H_1(a,b) = 0$.  


\item \textbf{Region 3 (orange):} In this region $H_1(a,b),H_2(b,a)$ are positive. For the mechanism to be optimal $H_1(a,a)=0$ and $H_2(a,a)<0$. From the BIC constraints, in this region, there is an additional condition given by the interim allocation probabilities. That is $\pab-\qab\le \paa.$ We see in more detail in~\Cref{subsec: region3}.

\item \textbf{Region 4 (purple):} In this region, the optimal flow is making $H_2(b,a)=0$. The virtual values for the first item are always positive. That is, $H_1(a,b),H_1(a,a)>0$. The only negative virtual value is $H_2(a,a).$ The same BIC condition must hold, $\pab-\qab\le \paa.$ We see in more detail in~\Cref{subsec: region4}.  

\item \textbf{Region 5 (brown):} In this region, the optimal flow makes $H_2(a,a)=0.$ All other virtual values are positive. The same BIC condition must hold, $\pab-\qab\le \paa.$ 

\item \textbf{Region 6 (black):} This region is combining Regions 3,4 and 5, when the condition does not hold. That is when $\pab-\qab> \paa.$  Only $H_2(a,a)$ is negative.

\item \textbf{Region 7 (blue):} In this region, all the virtual values are positive. 

\end{itemize}

\begin{figure}[ht]
\centering
    \begin{subfigure}{0.4\textwidth}
        \centering
    \includegraphics[width=\linewidth]{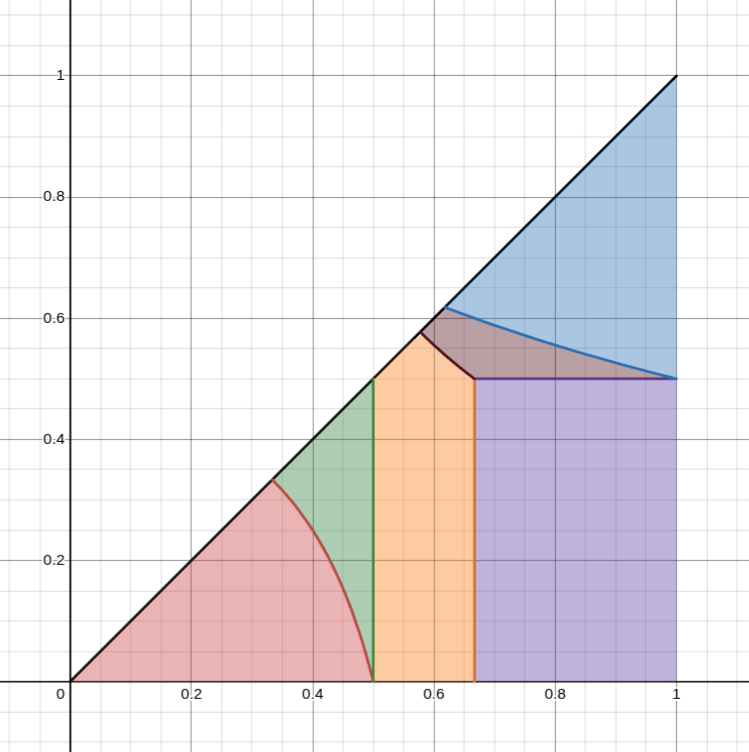}
    \caption{Regions $1-5\;\&\;7$   }
    \label{fig: regionGraph}
    \end{subfigure}
    \hfill
    \begin{subfigure}{0.4\textwidth}
        \centering
    \includegraphics[width=\linewidth]{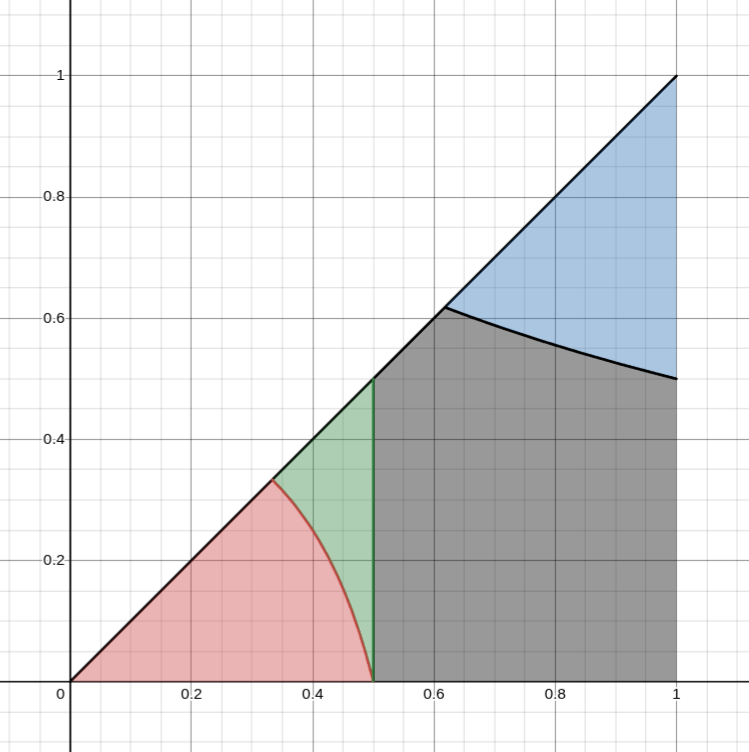}
    \caption{ Regions $1,2,6,7$ }
    \label{fig: regionGraphMerged}
    \end{subfigure}
    \caption{The regions defined for $p\ge q$ and $b>a$.  Red: Region 1, Green: Region 2, Orange: Region 3, Purple: Region 4, Brown: Region 5, Black: Region 6, Blue: Region 7. This graph is for $a=1$ and $b=2$ (\href{https://www.desmos.com/calculator/392a84bf6c}{link} to the graph parametrized by $a$ and $b$). }
    \label{fig: regions}
\end{figure}

Intuitively, we can see that the opposite BIC constraint is necessary in Variant~$(I)$ ($\qab + \paa \ge \pab +\qaa$) and in Variant~$(II)$ ($\qab + \paa \le \pab +\qaa$). In the regions that have the tension from the BIC constraint, we have that $H_2(a,a)\le 0$. Therefore, the condition becomes $\qab + \paa \le \pab$.  When the mechanism is under Variant~$(II)$ with optimal flow is $x=(1-p)(1-q)$ the virtual values are positive besides $H_2(a,a)<0$, and thus the interim allocation probabilities are defined according to~\Cref{eq: pbb,eq: qbb,eq: pab,eq: qab,eq: paa} and they cannot be modified. Hence, whenever $\qab + \paa \ge \pab$, and the optimal flow can be selected for any $0\le x\le (1-p)(1-q)$, depending on the region, one virtual value is set to zero so that the BIC constraint is tight, as we can see in~\Cref{eq: BIC-variantI-eq}. Recall that when a virtual value is equal to zero, the mechanism first flips a biased coin with probability $\delta$ and then allocates the item uniformly at random among the agents with zero virtual value. In~\Cref{fig: regions}, we present the different regions for $a=1$ and $b=2.$ 


Next, we present each region of the mechanism separately. For each region, we define the boundaries, then we show the optimal flow and the corresponding virtual values(Table~\ref{tab: virtual-values-non-iid-items}). Then prove that the flow induces the mechanism where we calculate the interim allocation probabilities, and we apply \Cref{lemma: non-iid-items-payments} to find the payments. Finally, we show that the induced mechanism is BIC, using~\Cref{lemma: non-iid-items-probabilities-monotonicity}. We formally present the proofs in~\Cref{apx: regions}.

\subsubsection{Region 1.}\label{subsec: region1}

In Region 1, all four virtual values at types with an $a$-report are negative. By~\Cref{lemma: non-iid-items-monotone} it follows that $H_1(a,a)\le H_1(a,b)$ and $H_2(a,a)\le H_2(b,a).$ Hence, it suffices to check $H_1(a,b),H_2(b,a).$ It is easy to verify $H_1(a,b)\le0$ whenever $x\ge\frac{a}{b-a}p(1-q)$ and $H_2(b,a)\le0$ whenever $x\le(1-p)(1-q)-\frac{a}{b-a}(1-p)q$. Hence, for all $x$ in the feasible range is $\frac{a}{b-a}p(1-q)\le x \le (1-p)(1-q)- \frac{a}{b-a}(1-p)q$ the mechanism is optimal. Equivalently, for the bounds to be consistent, it must be $\frac{(1-p)(1-q)}{1-pq}> \frac{a}{b}$ (boundary of region 1).  Since in this region $x<(1-p)(1-q)$, the mechanism falls under Variant~$I$.

\begin{algorithm}[ht]
\caption{Mechanism for Region 1}\label{mech: Region1}
\begin{algorithmic}
\State The hierarchy allocation function is the same for all agents. That is for all $i\in [n]$ and $v\in \Vcal$, we have $H_{i,j} (v) = H_j(v)$. Let $v = (v_1,v_2)$, if $v_j=b$ then $H_j(v)=b$, else if $v_j=a$ then $H_j(v)<0.$

\noindent The payment function for reporting $v\in \Vcal$ is 
\begin{itemize}
    \item $p^{(1)}(b,b) = b(\pbbr{1} + \qbbr{1})$
    \item $p^{(1)}(b,a) = b\pbbr{1} $
    \item $p^{(1)}(a,b) = b\qbbr{1}$
    \item $p^{(1)}(a,a) = 0$
\end{itemize}

where $\pbbr{1} = \frac{1-p^n}{n(1-p)} $, $\qbbr{1} =  \frac{1-q^n}{n(1-q)}$, and $\pabr{1} = \qabr{1} = \paar{1} = \qaar{1} = 0$
\end{algorithmic}
\end{algorithm}

\begin{lemma} \label{lemma: region1-induce}
In Region 2 (i.e., $\frac{a}{b-a}\,p(1-q)\le (1-p)(1-q)-\frac{a}{b-a}(1-p)q$), the flow with any $x\in(\frac{a}{b-a}p(1-q),(1-p)(1-q)-\frac{a}{b-a}(1-p)q]),$ induces \Cref{mech: Region1}.  
\end{lemma}




\begin{lemma}\label{lemma: region1-BIC}
\Cref{mech: Region1} is BIC.
\end{lemma}


\subsubsection{Region 2.}\label{subsec: region2}

Region 2 starts exactly where Region 1 ends, namely when $\frac{(1-p)(1-q)}{1-pq}\le\frac{a}{b}.$ Equivalently, $(1-p)(1-q)-\frac{a}{b-a}(1-p)q\le\frac{a}{b-a}\,p(1-q).$ As we increase $x$, the virtual value $H_2(b,a)$ turns positive before $H_1(a,b)$ does. The BIC constraint for Variant~$I$ (\Cref{eq: BIC-variantI-eq}), $\qab - \pab + \paa - \qaa =0$ together with $\pi_1(a,a)=\pi_2(a,a)=0$ (since $H_1(a,a),H_2(a,a)<0$) forces $\pab=\qab.$  The only way to satisfy both inequalities is to set $H_1(a,b)=0,$ which fixes $x=\frac{a}{b-a}\,p(1-q)$.  Finally, for the flow to be feasible, $\frac{a}{b-a}\,p(1-q) \le (1-p)(1-q)$, equivalently $p<\frac{(b-a)}{b}$.

\begin{algorithm}[ht]
\caption{Mechanism for Region 2}\label{mech: Region2}
\begin{algorithmic}
\State The hierarchy allocation function is the same for all agents. That is for all $i\in [n]$ and $v\in \Vcal$, we have $H_{i,j} (v) = H_j(v)$. If $v_j=b$ then $H_j(v)=b$, if $v=(b,a)$ then $H_2(v)>0$, if $v = (a,b)$, then $H_1(v)=0$,  else if $v=(a,a)$ then $H_1(v),H_2(v)<0$. The payment function for reporting $v\in \Vcal$ is 
\begin{itemize}
    \item $p^{(2)}(b,b) = b(\pbbr{2} + \qbbr{2}) - (b-a)\pabr{2}$
    \item $p^{(2)}(b,a) = b\pbbr{2} + a\qabr{2}$
    \item $p^{(2)}(a,b) = a\pabr{2} + b\qbbr{2}$
    \item $p^{(2)}(a,a) = 0$
\end{itemize}

\noindent where $\pbbr{2} = \frac{(1-p^n)}{n(1-p)} $, $\qbbr{2} =  \frac{(1-q^n)}{n(1-q)}$, $\pabr{2} = \qabr{2} = q^{n-1} \pbbr{2} $, and $\paar{2} = \qaar{2} = 0 $
\end{algorithmic}
\end{algorithm}

\begin{lemma}\label{lemma: region2-induce}
In Region 2 (i.e., $(1-p)(1-q)-\frac{a}{b-a}(1-p)q\le\frac{a}{b-a}\,p(1-q)$ and $p<\frac{b-a}{b}$), the flow with $x=\frac{a}{b-a}p(1-q)$ induces \Cref{mech: Region2}.  
\end{lemma}
 

 

\begin{lemma}\label{lemma: region2-BIC}
\Cref{mech: Region2} is BIC.
\end{lemma}



\subsubsection{Region 3.}\label{subsec: region3}

Region 3 starts at $p\ge \frac{b-a}{b},$ again where Region 2 ends. The next virtual value to turn positive is $H_1(a,a)$. The virtual values $H_1(a,b),H_2(b,a)\ge 0$ and $H_2(a,a)<0$, implying that $p\le\frac{b}{b+a}$ and $pq\ge \frac{b-a}{b+a}.$  Under Variant I, Lemma \ref{lemma: non-iid-items-BIC} gives the BIC constraint $\qab +\paa \ge \pab + \qaa$. Since in Region 3 $H_2(a,a)<0$ for all $x\le(1-p)(1-q)$, we have $\qaa=0,$ and the constraint reduces to $\qab +\paa \ge \pab$, where $\qab,\paa,\pab$  are given by \Cref{eq: pab,eq: qab,eq: paa}. However, from \Cref{eq: BIC-variantI-eq} we get that the constraint must be tight. To achieve this, we choose the flow such that $H_1(a,a)=0$, so that we can reduce $\paar{3}=\pab-\qab$ by selecting the appropriate $\delta_3$. Combining the above, we get that $x=1-p-\frac{a}{b-a}pq.$ The flow is feasible since $0\le1-p-\frac{a}{b-a}pq\le (1-p)(1-q).$

\begin{algorithm}[ht]
\caption{Mechanism for Region 3}\label{mech: Region3}
\begin{algorithmic}
\State The hierarchy allocation function is the same for all agents. That is for all $i\in [n]$ and $v\in \Vcal$, we have $H_{i,j} (v) = H_j(v)$. If $v_j=b$ then $H_j(v)=b$, if $v = (a,b)$, then $H_1(v)>0$, if $v=(b,a)$ then $H_2(v)>0$, else if $v=(a,a)$ then $H_1(v)=0$ and $H_2(v)<0$. The payment function for reporting $v\in \Vcal$ is 
\begin{itemize}
    \item $p^{(3)}(b,b) = b(\pbbr{3} + \qbbr{3}) - (b-a)\pabr{3}$
    \item $p^{(3)}(b,a) = b\pbbr{3} + a\qabr{3} - (b-a)\paar{3} $
    \item $p^{(3)}(a,b) = b\qbbr{3} + a\pabr{3} $
    \item $p^{(3)}(a,a) = a \paar{3}$
\end{itemize}

\noindent where $\pbbr{3} = \frac{p(1-p^n)}{n(1-p)} $, $\qbbr{3} =  \frac{q(1-q^n)}{n(1-q)}$, $\pab = p^{n-1} \qbbr{3}$, $\qabr{3} = q^{n-q} \pbbr{3} $,  $\paar{3} = \pabr{3}-\qabr{3} $, and $ \qaar{3} = 0 $
\end{algorithmic}
\end{algorithm}

\begin{lemma}\label{lemma: region3-induce} 
In Region 3 (i.e.,$\frac{b-a}{b}\le p\le\frac{b}{b+a}$ and $pq \ge \frac{b-a}{b+a}$), the flow with $x=1-p-\frac{a}{b-a}pq$ induces \Cref{mech: Region3}.  
\end{lemma}





\begin{lemma}\label{lemma: region3-BIC}
\Cref{mech: Region3} is BIC.
\end{lemma}


    

\subsubsection{Region 4.}\label{subsec: region4}

Region 4 starts from the point where the virtual $H_2(b,a)$ switches from positive to negative after Region 3. The boundary, therefore, is $p\ge\frac{b}{b+a}.$ By the monotonicity of the virtual values we have $H_1(a,b)>H_1(a,a)>0$ (\Cref{lemma: non-iid-items-monotone}).The only negative virtual value is $H_2(a,a)<0.$ To maintain $H_2(a,a)<0$ we must have $q\le \frac{b-a}{b}.$ In this region, the same BIC constraint, as in Region 3, must hold $\qab+\paa\ge\pab$ since the mechanism is under Variant~$I$.  However, \Cref{eq: BIC-variantI-eq} enforces the inequality to be tight. To achieve equality, the optimal flow induces $H_2(b,a)=0$ with flow $x=(1-p)(1-q)-\frac{a}{b-a}(1-p)q$, allowing to control $\qabr{4}$ by appropriately selecting the probability $\delta_4$.
\smallskip
\begin{algorithm}[ht]
\caption{Mechanism for Region 4}\label{mech: Region4}
\begin{algorithmic}
\State The hierarchy allocation function is the same for all agents. That is for all $i\in [n]$ and $v\in \Vcal$, we have $H_{i,j} (v) = H_j(v)$. If $v_j=b$ then $H_j(v)=b$, if $v = (a,b)$, then $H_1(v)>0$, if $v=(b,a)$ then $H_2(v)=0$, else if $v=(a,a)$ then $H_1(v)>0$ and $H_2(v)<0$.  

The payment function for reporting $v\in \Vcal$ is 
\begin{itemize}
    \item $p^{(4)}(b,b) = b(\pbbr{4} + \qbbr{4}) - (b-a)\pabr{4}$
    \item $p^{(4)}(b,a) = b\pbbr{4} + a\qabr{4} - (b-a)\paar{4} $
    \item $p^{(4)}(a,b) = b\qbbr{4} + a\pabr{4} $
    \item $p^{(4)}(a,a) = a \paar{4}$
\end{itemize}

\noindent where $\pbbr{4} = \frac{p(1-p^n)}{n(1-p)} $, $\qbbr{4} =  \frac{q(1-q^n)}{n(1-q)}$, $\pab = p^{n-1} \qbbr{4}$, $\paar{4}=\frac{1}{n}(pq)^{n-1}$ $\qabr{4} = \pabr{4}-\paar{4}$ and $ \qaar{4}= 0 $
\smallskip
\end{algorithmic}
\end{algorithm}
\smallskip

\begin{lemma}\label{lemma: region4-induce}
In Region 4 (i.e., $p\ge\frac{b}{b+a}$ $ q\le\frac{b-a}{b}$), the flow with $x=(1-p)(1-q)-\frac{a}{b-a}(1-p)q$ induces \Cref{mech: Region3}.  
\end{lemma}




\begin{lemma}\label{lemma: region4-BIC}
\Cref{mech: Region4} is BIC.
\end{lemma}


    

\subsubsection{Region 5.}\label{subsec: region5}

Region 5 starts from the point where the virtual $H_2(a,a)$ switches from negative to positive after Regions 3 and 4. That is $pq\ge\frac{b-a}{b+a}$ from Region 3 and $q\ge\frac{b-a}{b}$ from Region 4. Following the pattern of the previous regions, the optimal flow induces $H_2(a,a)=0.$ By construction $x=\frac{a}{b-a}pq-p(1-q),$ for the flow to be feasible $\frac{a}{b-a}pq-p(1-q)> (1-p)(1-q).$ Beyond that point, all virtual values are positive. Equivalently $\frac{1-q}{pq}> \frac{a}{b-a}.$  In this region, the same BIC constraint, as in Regions 3 and 4, must hold $\qab+\paa\ge\pab$ since the mechanism is under Variant~$I$.  However, \Cref{eq: BIC-variantI-eq} enforces the constraint to be tight  $\qab+\paa=\pab +\qaa$. Notice that in all previous regions $\qaa=0$, here we choose $\qaar{5}$ such that the equality is met, by appropriately selecting the probability $\delta_5$.

\begin{algorithm}[ht]
\caption{Mechanism for Region 5}\label{mech: Region5}
\begin{algorithmic}
\State The hierarchy allocation function is the same for all agents. That is for all $i\in [n]$ and $v\in \Vcal$, we have $H_{i,j} (v) = H_j(v)$. If $v_j=b$ then $H_j(v)=b$, otherwise $H_1(a,b)>H_1(a,a)>0$, if $v=(b,a)$ then $H_2(v)>0$, else if $v=(a,a)$ then $H_2(v)=0$.  The payment function for reporting $v\in \Vcal$ is 
\begin{itemize}
    \item $p^{(5)}(b,b) = b(\pbbr{5} + \qbbr{5}) - (b-a)(\qabr{5}+\paar{5})$
    \item $p(b,a) = b\pbbr{5} + a\qabr{5} - (b-a)\paar{6} $
    \item $p(a,b) = b\qbbr{5} + a\pabr{5} - (b-a)\qaar{6} $
    \item $p(a,a) = a (\paar{5}+\qaar{5})$
\end{itemize}

\noindent where $\pbbr{5} = \frac{p(1-p^n)}{n(1-p)} $, $\qbbr{5} =  \frac{q(1-q^n)}{n(1-q)}$, $\pabr{5} = p^{n-1} \qbbr{5}$, $\qab = q^{n-1} \pbbr{5} $,\\$\paar{5} = \frac{1}{n}(pq)^{n-1} $ and $\qaar{5} = \qabr{5}-\pabr{5}+\paar{5}.$

\end{algorithmic}
\end{algorithm}

\begin{lemma} \label{lemma: region5-induce}
In Region 5 (i.e., $pq\ge\frac{b-a}{b+a}$, $ q\ge\frac{b-a}{b}$ and $\frac{1-q}{pq}>\frac{a}{b-a}$), the flow with $x=\frac{a}{b-a}pq-p(1-q)$ induces \Cref{mech: Region3}.  
\end{lemma}




\begin{lemma}\label{lemma: region5-BIC}
\Cref{mech: Region5} is BIC.
\end{lemma}


    

\subsubsection{Region 6.}\label{subsec: region6}

We define Region 6 as the union of Regions 3 and 4 when the condition $\pab-\qab \le \paa$ does not hold. Notice that this condition violates the BIC constraints of Variant~$I$. Hence, the optimal mechanism switches to Variant~$II,$ where $x=(1-p)(1-q)$. First, we note that the hierarchy allocation mechanism induced by the flow is not immediately BIC. Then, we transform it into a modified hierarchy allocation mechanism that is BIC with the same revenue. By duality, we know that any feasible dual solution serves as an upper bound of the expected revenue. Since we have matching objective values, the solutions are optimal.

Next, we define the hierarchy allocation mechanism induced by the flow for $x=(1-p)(1-q).$ We use the notation $\pi_j^{(*)}(v)$ to denote the interim probabilities. Observe that the only negative virtual value is $H_2(a,a)$, making $\qaar{*}=0$.  Using Table~\ref{tab: virtual-values-non-iid-items} we can see that $H_2(b,a)=a$, and that $H_1(a,b)=H_1(a,a)=a-\frac{1-p}{p}(b-a).$ By the definition of the hierarchy allocation mechanism, the item is allocated uniformly at random to the highest non-negative virtual value. Therefore, $\qabr{*} =\qab$ from \Cref{eq: qab}. On the other hand, the fact that $H_1(a,b)=H_1(a,a)$ implies that $\pabr{*}=\paar{*}.$ When she reports $a$ for the first item, it is allocated whenever all other agents have reported $a$ independently from their report for the second item. Therefore, $\pabr{*}=\paar{*}=\frac{1}{n}p^{n-1}.$ Finally, since $H_j(v)=b$ when $v_j=b$, we get that $\pbbr{*}=\pbb$ (\Cref{eq: pbb}) and $\qbbr{*}=\qbb$ (\Cref{eq: qbb}). Applying \Cref{lemma: non-iid-items-payments}, and the fact that we are under Variant~$II$, we get that  
\begin{equation}\label{eq: induced-mech-region6}
    \begin{aligned}
p^{(*)}(b,b)&=b\bigl(\pbbr{*}+\qbbr{*}\bigr)
-(b-a)\bigl(\pabr{*}+\qaar{*}\bigr)\\
p^{(*)}(b,a)&=b\,\pbbr{*}+a\,\qabr{*}-(b-a)\,\paar{*},\\
p^{(*)}(a,b)&=a\,\pabr{*}+b\,\qbbr{*}-(b-a)\,\qaar{*},\\
p^{(*)}(a,a)&=a\bigl(\paar{*}+\qaar{*}\bigr),\\
\end{aligned}
\end{equation}

The BIC constraint violated by the flow-induced mechanism according to~\Cref{lemma: non-iid-items-BIC} (\Cref{eq: BIC-variantII}) is $\pab + \qaa \ge \qab +\paa$. For $\pabr{*}=\paar{*}$ the constraint would require $\qaar{*}\ge\qabr{*}$ but $\qaar{*}=0$.

Observe that in the definition of~\Cref{mech: Region6}, the only modification we do to the hierarchy allocation function is that $H_1(a,b)>H_1(a,a)$.This allows us to separate the allocation probabilities and follow \Cref{eq: pbb,eq: qbb,eq: pab,eq: qab,eq: paa}, for all interim probabilities besides  $\qaar{6}=0,$ since $H_2(a,a)<0.$

\begin{algorithm}[ht]
\caption{Mechanism for Region 6}\label{mech: Region6}
\begin{algorithmic}
\State The hierarchy allocation function is the same for all agents. That is for all $i\in [n]$ and $v\in \Vcal$, we have $H_{i,j} (v) = H_j(v)$. If $v_j=b$ then $H_j(v)=b$, if $v = (a,b)$, then $H_1(v)=$, if $v=(b,a)$ then $H_2(v)=0$, else if $v=(a,a)$ then $H_1(v) = \frac{1}{n}(pq)^n$ and $H_2(v)<0$. The payment function for reporting $v\in \Vcal$ is 
\begin{itemize}
    \item $p^{(6)}(b,b) = b(\pbbr{6} + \qbbr{6}) - (b-a)\pabr{6}$
    \item $p^{(6)}(b,a) = b\pbbr{6} + a\qabr{6} - (b-a)\paar{6}$
    \item $p^{(6)}(a,b) = b\qbbr{6} + a\pabr{6}  $
    \item $p^{(6)}(a,a) = a \paar{6}$
\end{itemize}

\noindent where $\pbbr{6} = \frac{p(1-p^n)}{n(1-p)} $, $\qbbr{6} =  \frac{q(1-q^n)}{n(1-q)}$, $\pabr{6} = p^{n-1} \qbbr{6}$, $\qabr{6} = q^{n-1} \pbbr{6} $,
\\$\paar{6} = \frac{1}{n}(pq)^{n-1} $ and $\qaar{6} = 0$
\end{algorithmic}
\end{algorithm}

We start by showing that the mechanism is BIC.

\begin{lemma}\label{lemma: region6-BIC}
\Cref{mech: Region6} is BIC.
\end{lemma}




    

The next step in showing that~\Cref{mech: Region6} is optimal is to show that the revenue is equal to the revenue of the flow-induced mechanism.

\begin{lemma}\label{lemma: equal-revenue}
The revenue of~\Cref{mech: Region6} is equal to the revenue of the revenue of the flow-induced mechanism, with payment identity according to~\Cref{eq: induced-mech-region6}.  
\end{lemma}

Therefore, the hierarchy allocation mechanism is optimal, since we showed that a feasible primal solution has the same objective value as a feasible dual solution.

\subsubsection{Region 7.}\label{subsec: region7}

In the last region, all the virtual values are positive. Similar to region 6, the optimal flow is $x=(1-p)(1-q)$ under Variant~$II$. For $H_2(a,a)$ to be positive $\frac{1-q}{pq}\le \frac{a}{b-a}$. We run into the same issue as in Region 6, where the flow-induced mechanism is not BIC. Following the same logic as in Region 6, the interim allocation probailities of the mechanism are: $\pbbr{*}=\pbb$ (\Cref{eq: pbb}) since $H_j(v)=b$ when $v_j=b$ and $\qbbr{*}=\qbb$ (\Cref{eq: qbb}), $\pabr{*}=\paar{*}=\frac{1}{n}p^{n-1}$ since  $H_1(a,b)=H_1(a,a)$. The only difference from Region 6 is that $0<H_2(a,a) < H_2(b,a)$, making $\qaar{*} =\frac{1}{n}(pq)^{n-1}$, as in \Cref{eq: paa}.

We define the modified hierarchy mechanism by using $H_1(a,b)>H_1(a,a)$, which induces the interim from \Cref{eq: pbb,eq: qbb,eq: pab,eq: qab,eq: paa}, also shown in~\Cref{mech: Region7}. We can apply directly~\Cref{lemma: equal-revenue} since the only difference between~\Cref{mech: Region7} and the flow-induced mechanism are the interim probabilities of $\pab$ and $\paa.$

\begin{algorithm}[ht]
\caption{Mechanism for Region 7}\label{mech: Region7}
\begin{algorithmic}
\State The hierarchy allocation function is the same for all agents. That is for all $i\in [n]$ and $v\in \Vcal$, we have $H_{i,j} (v) = H_j(v)$. If $v_j=b$ then $H_j(v)=b$, if $v = (a,b)$, then $H_1(v)=$, if $v=(b,a)$ then $H_2(v)=0$, else if $v=(a,a)$ then $H_1(v) = \frac{1}{n}(pq)^n$ and $H_2(v)<0$. The payment function for reporting $v\in \Vcal$ is 
\begin{itemize}
    \item $p^{(7)}(b,b) = b(\pbbr{7} + \qbbr{7}) - (b-a)(\pabr{7}+\paar{7})$
    \item $p^{(7)}(b,a) = b\pbbr{7} + a\qabr{7} - (b-a)\paar{6} $
    \item $p^{(7)}(a,b) = b\qbbr{7} + a\pabr{7} - (b-a)\paar{7} $
    \item $p^{(7)}(a,a) = 2a \paar{7}$
\end{itemize}

\noindent where $\pbbr{7} = \frac{p(1-p^n)}{n(1-p)} $, $\qbbr{7} =  \frac{q(1-q^n)}{n(1-q)}$, $\pabr{7}= p^{n-1} \qbbr{7}$, $\qab = q^{n-1} \pbbr{7} $, and \\$\paar{7} = \qaar{7} = \frac{1}{n}(pq)^{n-1} $

\end{algorithmic}
\end{algorithm}

\begin{lemma}\label{lemma: region7-BIC}
\Cref{mech: Region7} is BIC.
\end{lemma}




    
\subsection{\texorpdfstring{$m$}{m} non-i.i.d. items and \texorpdfstring{$n=1$}{n=1} agent} \label{sec: extensions}
In this section, we consider the setting with a single agent ($n = 1$) and $m$ items. For each item $j \in [m]$, there exists a set of possible values given by 
\[
\Vcal_{j} = \{v_{j,1} + c, v_{j,2} + c, \dots, v_{j,\rho_j} + c\},
\]
where $0 < v_{j,1} < v_{j,2} < \dots < v_{j,\rho_j}$ and $c > 0$. The agent's overall valuation space is then $\Vcal = \bigtimes_{j \in [m]} \Vcal_j$. For each item $j$, the agent's value is drawn independently from a discrete distribution $\Dcal_j$ supported on $\Vcal_j$, and we assume that $\Pr[v_{j,1} + c] \ge \delta$ for some $\delta > 0$ and all $j \in [m]$. Recall that finding the optimal mechanism for any choice of $c$ is \#P-hard.

Given the product distribution $\Dcal = \bigtimes_{j \in [m]} \Dcal_j$, we show that there exists a threshold $c^*$ such that for all $c > c^*$, the revenue-optimal mechanism is grand bundling. This result generalizes the discrete analog of one of the results of Daskalakis et al.~\cite{daskalakis2017strong}, which considers $m$ i.i.d.\ uniformly distributed items over $[c, c+1]$, to arbitrary (non-identical) discrete distributions.

We emphasize that our result does \emph{not} generalize the continuous version of~\cite{daskalakis2017strong}; instead, it generalizes the discrete analog. To further bridge the gap between discrete and continuous settings, we apply standard discretization techniques and show that for any continuous product distribution $\Dcal$ and any $\epsilon > 0$, there exists a $c^*$ such that for all $c > c^*$, grand bundling achieves revenue at least $\mathrm{OPT} - \epsilon$.

\begin{algorithm}[!ht]
\caption{Grand Bundling}\label{mech: grandBundling}
\begin{algorithmic}
\State For any $v \in \Vcal$, and $j \in [m]$, $H_{j}(v)=1$ (i.e. the agent always receives all items).

\State For any $v \in \Vcal$ the agent pays $p(v)= cm+\sum_{j \in [m]}v_{j,1}$ (i.e. the price of the grand bundle).
\end{algorithmic}
\end{algorithm}

\begin{definition}(\Cref{mech: grandBundling} flow) \label{def: grandBundling flow}
    Let $\tilde{v}= [c+v_{1,1}, c+v_{2,1}, \dots, c+v_{m,1}]$. Then we define the \Cref{mech: grandBundling} flow as follows:
    \begin{itemize}
        \item  for any $v = [c+v_{1,\sigma_1}, c+v_{2,\sigma_2}, \dots, c+v_{m,\sigma_m}]\in \Vcal$, $\mu(v)=\mathbf{1}[v = \tilde{v}]$, where $\mathbf{1}[\cdot]$ is the indicator function.
        \item For any $v \in \Vcal-\{\tilde{v}\}$, $v' \in \Vcal$, $\lambda(v,v')= \mathbf{1}[v' = \tilde{v} \;\&\; v \neq v' ] \prod_{j \in [m]} \Pr[c+v_{j,\sigma_j}]$.
    \end{itemize}
\end{definition}
The intuitive interpretation of the above flow is that all nodes send their entire flow to node $\tilde{v}$ and node $\tilde{v}$ sends the entire flow it receives to the sink. Now, we will once again use our methodology to prove the main result of this section. First, we will show that for sufficiently large $c$, \Cref{mech: grandBundling} is induced by the flow described in \Cref{def: grandBundling flow}. Show truthfulness and individual rationality for \Cref{mech: grandBundling} is trivial. Combining the above with \Cref{thm: main Thm} we can easily prove the main result of this section.

\begin{lemma}\label{lemma: lower bound on c}
    For $c\ge \frac{v_{max}-v_{min}}{\delta^m}$  \Cref{mech: grandBundling} is induced by the flow described in \Cref{def: grandBundling flow}, where $v_{min} = \min_{j \in [m]} v_{j,1}$ and  $v_{max} = \max_{j \in [m]} v_{j,\rho_j}$.
\end{lemma}

\begin{proof}
    As before let $\tilde{v}= [c+v_{1,1}, c+v_{2,1}, \dots, c+v_{m,1}]$. For any $v \in \Vcal - \{\tilde{v}\}$, and $j \in [m]$, the mechanism induced by the flow of \Cref{def: grandBundling flow} has $H_{j}(v)= v_j - \frac{1}{\Pr[v]}\sum_{v' \in \mathcal{V}} \lambda_i(v', v)(v'_{j} - v_{j}) = v_j>0$. Since there is only one agent, whenever she has a positive $H_{j}(v)$ (virtual value) for the item, she takes it. Thus, for all $v \in \Vcal - \{\tilde{v}\}$ the agent receives all items, which is consistent with \Cref{mech: grandBundling}. Also, we have:
    \begin{align*}
        H_{j}(\tilde{v}) &= \tilde{v}_j - \frac{1}{\Pr[\tilde{v}]}\sum_{v' \in \mathcal{V}} \lambda_i(v', \tilde{v})(v'_{j} - \tilde{v}_{j}) \\
        &= c+v_{j,1} - \frac{1}{\Pr[\tilde{v}]}\sum_{v' \in \mathcal{V}} \lambda_i(v', \tilde{v})(v'_{j,\sigma_{j}}+c - (c+v_{j,1})) \\
        &= c+v_{j,1} - \frac{1}{\Pr[\tilde{v}]}\sum_{v' \in \mathcal{V}} \lambda_i(v', \tilde{v})(v'_{j,\sigma_{j}} -v_{j,1}) \\
        &\ge c - \frac{1}{\Pr[\tilde{v}]}\sum_{v' \in \mathcal{V}} \lambda_i(v', \tilde{v})(v_{max}-v_{min}) \tag{$v_{j,1}\ge0$ and $\lambda_i(v', \tilde{v}) \ge 0$}\\
        &= c-\frac{1}{\Pr[\tilde{v}]}(v_{max}-v_{min}) \tag{$\sum_{v' \in \mathcal{V}} \lambda_i(v', \tilde{v}) \le 1$}\\
        &\ge c-\frac{v_{max}-v_{min}}{\delta^m} \ge 0 \tag{$\Pr[\tilde{v}] \ge \delta^m$}
    \end{align*}
    Thus, given our assumption on $c$, allocating all items to the agent when their value is $\tilde{v}$ is also consistent with the flow-induced mechanism. Since we have shown that the allocation rule is consistent with the flow for all $v \in \Vcal$ we simply need to show that the payment is consistent as well to conclude the proof. Node $\tilde{v}$ sends all its flow directly to the sink. Thus, for the flow-induced payment, we have:
    \begin{align*}
        p_i(\tilde{v}) &= \frac{1}{\Pr[\tilde{v}]}\sum_{\ell \in \Pcal_{\tilde{v}}} \xi_{\ell} \left(\sum_{j \in [m]} \tilde{v}_{j} \pi_{j}(\tilde{v}) - \sum_{z \in [1,|\ell|-3]}  \sum_{j \in [m]} \left( \tilde{v}_{1,j}^{z}-v_{1,j}^{z+1} \right)\pi_{j}(v_{1}^{z+1})\right)\\
        &= \sum_{j \in [m]} \tilde{v}_{j} \pi_{j}(\tilde{v}) - \sum_{z \in [1,|\ell|-3]}  \sum_{j \in [m]} \left( \tilde{v}_{1,j}^{z}-v_{1,j}^{z+1} \right)\pi_{j}(v_{1}^{z+1}) \tag{$|\Pcal_{\tilde{v}}|=1$ and $\xi_{\ell} = \Pr[\tilde{v}]$}\\
        &= \sum_{j \in [m]} \tilde{v}_{j} \pi_{j}(\tilde{v}) \tag{$|\ell|=3$}\\
        &= \sum_{j \in [m]} \tilde{v}_{j} = cm+\sum_{j \in [m]}v_{j,1}
    \end{align*}
    For any other $v \in \Vcal-\{\tilde{v}\}$ we again have that there is only one path that includes this node and it is $\ell = (s,v, \tilde{v},\bot)$. Thus, for the flow-induced payment, we have:
    \begin{align*}
        p_i(v) &= \frac{1}{\Pr[v]}\sum_{\ell \in \Pcal_{v}} \xi_{\ell} \left(\sum_{j \in [m]} v_{j} \pi_{j}(v) - \sum_{z \in [1,|\ell|-3]}  \sum_{j \in [m]} \left( v_{1,j}^{z}-v_{1,j}^{z+1} \right)\pi_{j}(v_{1}^{z+1})\right)\\
        &= \sum_{j \in [m]} v_{j} \pi_{j}(v) - \sum_{z \in [1,|\ell|-3]}  \sum_{j \in [m]} \left( v_{1,j}^{z}-v_{1,j}^{z+1} \right)\pi_{j}(v_{1}^{z+1}) \tag{$|\Pcal_{v}|=1$ and $\xi_{\ell} = \Pr[v]$}\\
        &= \sum_{j \in [m]} v_{j} \pi_{j}(v) -\sum_{j \in [m]} \left( v_{j}-\tilde{v}_{j} \right)\pi_{j}(\tilde{v})  \tag{$\ell=(s,v, \tilde{v},\bot)$}\\
        &=\sum_{j \in [m]} v_{j}  -\sum_{j \in [m]} \left( v_{j}-\tilde{v}_{j} \right)  \tag{$\pi_{j}(v)=1$, for all $j \in [m]$, and $v \in \Vcal$}\\
        &= \sum_{j \in [m]} \tilde{v}_{j} = cm+\sum_{j \in [m]}v_{j,1}
    \end{align*}
    Thus, the payment is consistent with the flow as well. This concludes the proof.
\end{proof}

\begin{theorem} \label{thm: optimalBundling}
    For $c\ge \frac{v_{max}-v_{min}}{\delta^m}$ \Cref{mech: grandBundling} is optimal and always extracts $cm+\sum_{j \in [m]}v_{j,1}$ revenue, where $v_{min} = \min_{j \in [m]} v_{j,1}$ and  $v_{max} = \max_{j \in [m]} v_{j,\rho_j}$.
\end{theorem}

\begin{proof}
The proof follows directly from the structure of \Cref{mech: grandBundling}. The mechanism is truthful, as the allocation and payment do not depend on the agent’s reported values. It is individually rational because the payment is equal to the minimum possible value the agent could have for the grand bundle, namely $cm + \sum_{j \in [m]} v_{j,1}$. 

Optimality follows from \Cref{thm: main Thm} and the fact that the mechanism is induced by a feasible flow, as established in \Cref{lemma: lower bound on c}. Therefore, the mechanism is optimal, and we extract a revenue of exactly $cm + \sum_{j \in [m]} v_{j,1}$ from every agent type.
\end{proof}

\subsubsection{Extension to Continuous Distributions}

Although our results can be generalized for items with different supports, to keep simplicity and consistency with existing literature, assume that the value of each item $v_j$ is drawn interdependently from a continuous distribution $\Dcal_j$ that is supported in $[c,c+1]$.

\begin{theorem}
    For any $\epsilon>0$, if $c>\left(\frac{m}{\epsilon} \right)^m$, selling the grand bundle at a price of $mc$ extracts $\mathrm{OPT}-\epsilon$ revenue.
\end{theorem}

\begin{proof}
    First, we need to define two distances between distributions, the Total Variation distance and the Wasserstein Distance.
    \begin{definition}[Total Variation Distance]
The \emph{total variation (TV) distance} between any two probability distributions $P$ and $Q$ on a sample space $\Omega$ is defined as
\[ 
d_{TV}(P,Q)\triangleq \sup_{E \subseteq \Omega} | P(E) - Q(E) | \, ,\]
where the supremum is over all Borel measurable subsets $E \subseteq \Omega$, and $P(E)$ (resp. $Q(E)$) denotes the probability of the event $E$ with respect to the distribution $P$ (resp. $Q$).
\end{definition}
We will define the Wasserstein Distance similar to the definition of \cite{cai2021eBIC}, casting it on our setting.

\begin{definition}[Wasserstein Distance]
The \emph{Wasserstein Distance} between any two probability distributions $P$ and $Q$ is defined as the smallest expected $\ell_1$ distance over all couplings. Formally,
\[
d_w(P, Q) = \min_{\gamma \in \Pi(P,Q)} \int  \norm{v-v'}_1 \, d \gamma(v, v').
\]
where $\Pi(P,Q)$ is the set of all couplings consistent with $P,Q$.
\end{definition}

Now consider the product distribution $\Dcal = \bigtimes_{j\in [m]}\Dcal_j$ from which the agent's values for the items are drawn. First we will create a new distribution $\hat{\Dcal}$ that is the discretized version of $\Dcal$. Partition $[c,c+1]$ into $1/\epsilon'$ non-overlapping equal parts. Let $I_z = [c+\epsilon'z, c+\epsilon'(z+1)]$ for $z \in [1/\epsilon'-1]$. For each $\Dcal_j$ let $\hat{\Dcal}_j$ be the discrete distribution, supported on $[c,c+\epsilon', c+2\epsilon'+,\dots,c+1-\epsilon']$ such that:
\[ \Pr_{X \sim \hat{\Dcal}_j}[X=c+\epsilon'z]= \Pr_{X\sim \Dcal_j}[X \in I_z] \]
and let $\hat{\Dcal} = \bigtimes_{j \in [m]}\hat{\Dcal}_j$. Now consider the coupling where we first sample $X_j \sim \Dcal_j$ and then set $\hat{X}_j = c+ \epsilon'z$ where $z \in [1/\epsilon'-1]$ is such that $X_j \in I_z$, for all $j \in [m]$. This is a valid coupling between $\Dcal$ and $\hat{\Dcal}$. Also $\norm{[X_1, X_2, \dots, X_m]- [\hat{X}_1, \hat{X}_2, \dots, \hat{X}_m]}_1 \le \epsilon'm$ by construction. Thus:
\[d_w(D, \hat{D}) = \min_{\gamma \in \Pi(D, \hat{D})} \int  \norm{v-v'}_1 \, d \gamma(v, v') \le \int  \epsilon'm \, d \gamma(v, v') = \epsilon'm\]
Now we can use the following result:
\begin{corollary}[Corollary 2 \cite{cai2021eBIC}]
If $d_w(\Dcal_i, \Dcal_i') \le \kappa$ for all $i \in [n]$, let $\mathrm{OPT}(\Dcal)$ and $\mathrm{OPT}(\Dcal')$ be the optimal revenue achievable by any BIC and IR mechanism with respect to $\Dcal$ and $\Dcal'$ respectively. Then
\[
|\mathrm{OPT}(\Dcal) - \mathrm{OPT}(\Dcal')| \le O(n \cdot \sqrt{\kappa}).
\]
\end{corollary}
Thus $\mathrm{OPT}(\Dcal) \le \mathrm{OPT}(\hat{\Dcal}) + O(\sqrt{\epsilon'm})$. However, we are not ready to use \Cref{thm: optimalBundling} just yet since we do not have any lower bound on $\Pr_{X \sim \hat{D}_j}[X=c]$. For each $j \in [m]$ we will create a new discrete distribution $\hat{D}_j'$ as follows. Let $z^* = \min_{z \in [1/\epsilon'-1]}\{\Pr_{X \sim \hat{D}_j}[X\le c+z\epsilon']\ge \delta\}$. $\hat{D}_j'$ will be supported on $[c+z^*\epsilon',  c+1-\epsilon']$, and $\Pr_{X \sim \hat{D}'_j}[X = c+z\epsilon'] =  \Pr_{X \sim \hat{D}_j}[X = c+z\epsilon']$ for all $z \in [ z^*+1, 1-\epsilon']$ and $\Pr_{X \sim \hat{D}'_j}[X = c+z^*\epsilon'] = \sum_{z \in [z^*]} \Pr_{X \sim \hat{D}_j}[X = c+z\epsilon']$. In other words, we construct $\hat{D}'_j$ by dropping the leftmost part of $\hat{D}_j$ and appending all the probability mass onto the smallest value in the new support, such that the smallest value has probability mass at least $\delta$. By construction $d_{TV}(\hat{D}'_j,\hat{D}_j) \le \delta$ and by the tensorization of TV distance $d_{TV}(\hat{D}',\hat{D}) \le m\delta$, where $\hat{D}' = \bigtimes_{j \in [m]} \hat{D}'_j$. 

Now consider the optimal incentive compatible and individual rational mechanism $\Mcal$ under $\hat{D}$. $\Mcal$ is incentive compatible and individually rational for $\hat{D}'$ as well since $supp(\hat{D}') \subseteq supp(\hat{D})$. Now we can use the following result from \cite{makur2023robustness}:
\begin{lemma}[Lemma 2 \cite{makur2023robustness}]
Let $P$ and $Q$ be two arbitrary probability distributions supported on $\Tcal$ and let $\Mcal$ be any mechanism. Assuming truthful bidding, for all objective functions $\Ocal(.,.) \in [a,b]$, letting $V = b-a$, it holds that $\mathbb{E}_{t \sim P}[\Ocal(t,\M(t))] - \mathbb{E}_{t' \sim Q}[\Ocal(t',\M(t'))] \le V\, d_{\mathsf{TV}}(P,Q)$.
\end{lemma}
Translated to our setting we get that $OPT(\hat{\Dcal}) \le \mathbb{E}_{v \sim \hat{\Dcal}'}[Rev(\Mcal(v))]+ d_{TV}(\hat{D}'_j,\hat{D}_j) \le OPT(\hat{\Dcal}')+ m\delta$, where $Rev(\Mcal(v))$ is the revenue extracted from mechanism $\Mcal$ when valuation $v$ is reported, and $OPT(\hat{\Dcal})$ and $OPT(\hat{\Dcal}')$ are the expected revenue extracted by the optimal mechanism with respect to $\hat{\Dcal}$ and $\hat{\Dcal}'$. 

Combining the above we get that $OPT(\Dcal) \le OPT(\hat{\Dcal}')+ O(\sqrt{m\epsilon'})+m\delta$. We also have that from \Cref{thm: optimalBundling} that if $c>\frac{1}{\delta^{m}}$, $OPT(\hat{\Dcal}') = cm$. Thus for $c>\frac{1}{\delta^{m}}$, $OPT(\Dcal) \le cm +  O(\sqrt{m\epsilon'})+m\delta$. By taking $\epsilon' \rightarrow 0$ and setting $\delta = \epsilon/m$ we get that if $c> \left(\frac{m}{\epsilon} \right)^m$ then $OPT(\Dcal) \le cm+ \epsilon$ where $cm$ is equal to the revenue of selling the grand bundle at a price of $cm$ under $\Dcal$.
\end{proof}

\section*{Acknowledgments}

Marios Mertzanidis and Athina Terzoglou are supported in part by an NSF CAREER award CCF-2144208, and a research award from the Herbert Simon Family Foundation. The authors would like to thank Alexandros Psomas for the valuable discussions and suggestions in the early stages of this project.

\bibliographystyle{alpha}
\bibliography{refs}

\clearpage
\appendix
\section{Missing proofs from \Cref{sec: firstAxis}} \label{apx: proofs-Axis}

\begin{proof}[Proof of \Cref{lemma: technical sum binom}]
\begin{align*}
\sum_{j=1}^n \frac{1}{j} \binom{n-1}{j-1} q^{j-1} p^{n-j}&= \sum_{j=1}^n \frac{1}{j} \frac{(n-1)!}{(j-1)! (n-j)!} q^{j-1} p^{n-j} \\
&= \sum_{j=1}^n \frac{1}{n} \frac{n!}{j! (n-j)!} q^{j-1} p^{n-j} \\
&= \sum_{j=1}^n \frac{1}{n} \binom{n}{j} q^{j-1} p^{n-j} \\
&= \frac{1}{nq} \sum_{j=0}^n  \binom{n}{j} q^j p^{n-j} - \frac{p^n}{nq} \\
&= \frac{1}{nq} \left(p + q\right)^{n} - \frac{p^n}{nq} \\
&= \frac{(p+q)^n- p^n}{nq}.
\end{align*}
\end{proof}

\begin{proof}[Proof of \Cref{lemma: Binomial Inequality}]
    For $k=m$, this is trivially true. By opening $\pr{B(m, p) = k | B(m, p) \ge k} = \frac{\binom{m}{k}(1-p)^{m-k}p^k}{\sum_{i=k}^m \binom{m}{i}(1-p)^{m-i}p^i}$. Proving our statement is equivalent to showing that $\left(\frac{p(m-k)}{(1-p)}-k\right)\sum_{i=k}^m \binom{m}{i}(1-p)^{m-i}p^i + k\binom{m}{k}(1-p)^{m-k}p^k = \frac{p(m-k)}{(1-p)}\sum_{i=k}^m \binom{m}{i}(1-p)^{m-i}p^i - k\sum_{i=k+1}^m \binom{m}{i}(1-p)^{m-i}p^i  \ge 0$. So let 
    \[g(k) = \frac{p(m-k)}{(1-p)}\sum_{i=k}^m {m \choose i}(1-p)^{m-i}p^i - k\sum_{i=k+1}^m {m \choose i}(1-p)^{m-i}p^i \]
    Thus we simply need to show that $g(k)\ge 0$ for all $k \in [m-1]$. We will show that $g(k) \ge g(k+1)$ for all $k \in [m-2]$.
    \begin{align*}
        &g(k)-g(k+1) = \\
        &\frac{p(m-k)}{(1-p)}\sum_{i=k}^m {m \choose i}(1-p)^{m-i}p^i - k\sum_{i=k+1}^m {m \choose i}(1-p)^{m-i}p^i \\
        &- \left( \frac{p(m-(k+1))}{(1-p)}\sum_{i=k+1}^m {m \choose i}(1-p)^{m-i}p^i - (k+1)\sum_{i=k+2}^m {m \choose i}(1-p)^{m-i}p^i \right)\\
        &= \frac{p(m-k)}{(1-p)}\sum_{i=k+1}^m {m \choose i}(1-p)^{m-i}p^i + \frac{p(m-k)}{(1-p)}{m \choose k}(1-p)^{m-k}p^k\\
        &- \frac{p(m-k)}{(1-p)}\sum_{i=k+1}^m {m \choose i}(1-p)^{m-i}p^i + \frac{p}{1-p}\sum_{i=k+1}^m {m \choose i}(1-p)^{m-i}p^i \\
        &+ k \sum_{i=k+2}^m {m \choose i}(1-p)^{m-i}p^i + \sum_{i=k+2}^m {m \choose i}(1-p)^{m-i}p^i \\
        &-k\sum_{i=k+2}^m {m \choose i}(1-p)^{m-i}p^i - k {m \choose k+1}(1-p)^{m-(k+1)}p^{k+1}\\
        &= \frac{p(m-k)}{(1-p)}{m \choose k}(1-p)^{m-k}p^k -k {m \choose k+1}(1-p)^{m-(k+1)}p^{k+1}  \\
        &+ \frac{p}{1-p}\sum_{i=k+1}^m {m \choose i}(1-p)^{m-i}p^i+ \sum_{i=k+2}^m {m \choose i}(1-p)^{m-i}p^i\\
        &\ge \frac{p(m-k)}{(1-p)}{m \choose k}(1-p)^{m-k}p^k -k {m \choose k+1}(1-p)^{m-(k+1)}p^{k+1}\\
        &= (1-p)^{m-(k+1)}p^{k+1} \left( (m-k){m \choose k} - k {m \choose k+1} \right)\\
        &= (1-p)^{m-(k+1)}p^{k+1} \left((m-k) \frac{m!}{(m-k)!k!} - k \frac{m!}{(m-(k+1))!(k+1)!} \right)\\
        &= (1-p)^{m-(k+1)}p^{k+1} \frac{m!}{(m-(k+1))!k!} \left( 1 - \frac{k}{k+1}\right)\ge 0
    \end{align*}
Thus we simply need to prove that $g(m-1)\ge0$:
\begin{align*}
    g(m-1) &= \frac{p}{(1-p)}\sum_{i=m-1}^m {m \choose i}(1-p)^{m-i}p^i - (m-1)\sum_{i=m}^m {m \choose i}(1-p)^{m-i}p^i\\
    &= \frac{p}{(1-p)} \left({m \choose m-1}(1-p)p^{m-1} + {m \choose m}p^m \right) -(m-1){m \choose m}p^m \\
    &= mp^m + \frac{p}{1-p}p^m - m p^m + p^m\\
    &= \frac{p^{m}}{1-p} \ge 0
\end{align*}
\end{proof}

\begin{proof}[Proof of~\Cref{lemma: tech-f(k)}]
\begin{align*}
    f(k)&= a-\frac{1}{(1-p)^k p^{m-k} } \cdot\frac{1}{(m-k)\binom{m}{k}}\cdot \sum_{z=k+1}^m\binom{m}{z}(1-p)^zp^{m-z} \cdot (b-a) \\
    &= a-\frac{1}{(1-p)^k p^{m-k} } \cdot\frac{1}{(m-k)\binom{m}{k}}\cdot (b-a)\cdot \left(\sum_{z=k}^m\binom{m}{z}(1-p)^zp^{m-z} -\binom{m}{k}(1-p)^kp^{m-k}\right)\\
    &= a-\frac{1}{(1-p)^k p^{m-k} } \cdot\frac{1}{(m-k)\binom{m}{k}}\cdot (b-a)\sum_{z=k}^m\binom{m}{z}(1-p)^zp^{m-z} \left(1 -\frac{\binom{m}{k}(1-p)^kp^{m-k}}{\sum_{z=k}^m\binom{m}{z}(1-p)^zp^{m-z}}\right)\\
    &\ge a-\frac{1}{(1-p)^k p^{m-k} } \cdot\frac{1}{(m-k)\binom{m}{k}}\cdot (b-a)\sum_{z=k}^m\binom{m}{z}(1-p)^zp^{m-z} \left(1 -\left(1-\frac{(1-p)(m-k)}{pk}\right)\right) \tag{\Cref{lemma: Binomial Inequality}}\\
    &= a-\frac{1-p}{(1-p)^k p^{m-k}p } \cdot\frac{m-k}{(m-k)k\binom{m}{k}}\cdot \sum_{z=k}^m\binom{m}{z}(1-p)^zp^{m-z} \cdot (b-a)\\
    &=  a-\frac{1}{(1-p)^{k-1} p^{m-k+1} } \cdot\frac{(k-1)!(m-k)!}{m!}\cdot \sum_{z=k}^m\binom{m}{z}(1-p)^zp^{m-z} \cdot (b-a)\\
    &=  a-\frac{1}{(1-p)^{k-1} p^{m-k+1} } \cdot\frac{(k-1)!(m-k+1)!}{(m-k+1) \cdot m!}\cdot \sum_{z=k}^m\binom{m}{z}(1-p)^zp^{m-z} \cdot (b-a)\\
    &=a-\frac{1-p}{(1-p)^k p^{m-k}p } \cdot\frac{1}{(m-k+1)\binom{m}{k-1}}\cdot \sum_{z=k}^m\binom{m}{z}(1-p)^zp^{m-z} \cdot (b-a) = f(k-1)
\end{align*}
\end{proof}

\begin{proof}[Proof of~\Cref{thm: revenue-first-axis}]

 The fact that the mechanism is optimal follows directly from \Cref{thm: main Thm}, in conjunction with \Cref{thm: mech1 BIC} and \Cref{lemma: first axis mechanism induced by flow}. Since the mechanism is optimal, it follows that the flow defined in \Cref{def: firstAxis flow} fully characterizes an optimal solution to the dual program. Consequently, we can compute the expected revenue extracted by the optimal mechanism via the objective value attained by this dual solution. 
 
 For compactness of notation, define
\[
f(k) = a - \frac{1}{(1-p)^k p^{m-k}} \cdot \frac{1}{(m-k)\binom{m}{k}} \sum_{z=k+1}^m \binom{m}{z} (1-p)^z p^{m-z} \cdot (b-a).
\]
Using this, we have:

 \begin{align*}
    \mathbb{E}[\mathrm{Rev}] &= \sum_{v \in \mathcal{V}} \Pr[v] \sum_{j \in [m]}  \left[\max_{i \in [n]}\left\{
v_{i,j} - \frac{1}{\Pr[v_i]} \sum_{v'_i \in \mathcal{V}_i} \lambda_i(v'_i, v_i)(v'_{i,j} - v_{i,j}) \right\}\right]^+\\
&= \sum_{j \in [m]} \left(\Pr\left[b= \max_{i \in [n]}\left\{
v_{i,j} - \frac{1}{\Pr[v_i]} \sum_{v'_i \in \mathcal{V}_i} \lambda_i(v'_i, v_i)(v'_{i,j} - v_{i,j}) \right\}\right]\cdot b\right. \\
&\phantom{=} \left. + \sum_{k=k^*}^{m-1} \Pr\left[f(k)= \max_{i \in [n]}\left\{
v_{i,j} - \frac{1}{\Pr[v_i]} \sum_{v'_i \in \mathcal{V}_i} \lambda_i(v'_i, v_i)(v'_{i,j} - v_{i,j}) \right\}\right] \cdot f(k)\right)
\end{align*}
But we know that for an item $j \in [m]$ if there exists an agent $i \in [n]$ with $v_{i,j} = b$, then, 
\[\max_{i \in [n]}\left\{
v_{i,j} - \frac{1}{\Pr[v_i]} \sum_{v'_i \in \mathcal{V}_i} \lambda_i(v'_i, v_i)(v'_{i,j} - v_{i,j}) \right\} = b.\] 
Thus, 
\[\Pr\left[b= \max_{i \in [n]}\left\{
v_{i,j} - \frac{1}{\Pr[v_i]} \sum_{v'_i \in \mathcal{V}_i} \lambda_i(v'_i, v_i)(v'_{i,j} - v_{i,j}) \right\}\right]= 1 - \Pr[\forall i \in [n],\, v_{i,j}=a ]=1-p^n\]
If for all $i \in [n]$, $v_{i,j}=a$ then:
\[\arg\max_{i \in [n]}\left\{
v_{i,j} - \frac{1}{\Pr[v_i]} \sum_{v'_i \in \mathcal{V}_i} \lambda_i(v'_i, v_i)(v'_{i,j} - v_{i,j}) \right\} = \arg\max_{i \in [n]}\left\{ k_{v_i} \right\}\]
the player with the most $b$'s in his valuation will have the highest virtual value. Thus:
\begin{align*}
    &\Pr\left[f(k)= \max_{i \in [n]}\left\{
v_{i,j} - \frac{1}{\Pr[v_i]} \sum_{v'_i \in \mathcal{V}_i} \lambda_i(v'_i, v_i)(v'_{i,j} - v_{i,j}) \right\}\right]\\
&= \Pr\left[\forall i\in [n], v_{i,j}=a \text{ and }\exists i\in [n], k_{v_i}=k \text{ and }  \forall i\in [n], k_{v_i} \le k\right]\\
&= \Pr[\forall i\in [n], v_{i,j}=a] \cdot Pr\left[\exists i\in [n], k_{v_i}=k \text{ and }  \forall i\in [n], k_{v_i} \le k \; | \; \forall i\in [n], v_{i,j}=a\right] \\
&= p^n \sum_{z=1}^{n}  Pr\left[ |i\in [n]: k_{v_i}=k|= z \text{ and }  |i\in [n]: k_{v_i}<k|= n-z \; | \; \forall i\in [n], v_{i,j}=a\right] \\
&= p^n \sum_{z=1}^{n} \binom{n}{z} 
    \left( \Pr[B(m-1,1-p) = k] \right)^z 
    \left( \Pr[B(m-1,1-p) < k] \right)^{n-z} \\
&= p^n \left(\sum_{z=0}^{n} \binom{n}{z} 
    \left( \Pr[B(m-1,1-p) = k] \right)^z 
    \left( \Pr[B(m-1,1-p) < k] \right)^{n-z}  - \left( \Pr[B(m-1,1-p) < k] \right)^{n}\right) \\
&= p^n \bigg( 
    \left( \Pr[B(m-1,1-p) = k] + \Pr[B(m-1,1-p) < k]\right)^n   - \left( \Pr[B(m-1,1-p) < k] \right)^{n}\bigg) \\
&= p^n \bigg( 
    \left( \Pr[B(m-1,1-p) \le k] \right)^n   - \left( \Pr[B(m-1,1-p) < k] \right)^{n}\bigg)
\end{align*}
Combining all of the above we get the desired result.
\end{proof}

\section{Missing proofs from \Cref{sec: secondAxis}}\label{apx: proofs-sec-axis}

\begin{proof}[Proof of~\Cref{lemma: non-iid-agents-bic-conditions}]

We must show for each agent $i$, all types $v_i\in\{a,b\}^2$ and all misreports $v_i'\in\{a,b\}^2$,
\[
\Ex{\divutil{v_i}{v_i'}} \;\le\; \Ex{\divutil{v_i}{v_i}}.
\]

By definition, we see that the expected value of truth-telling and deviating can be found by
\[
\Ex{\divutil{v_i}{v_i}}
= \sum_{j\in\{1,2\}} v_{ij}\,\pi_{ij}(v_i) \;-\; p_i(v_i),
\quad
\Ex{\divutil{v_i}{v_i'}}
= \sum_{j\in\{1,2\}} v_{ij}\,\pi_{ij}(v_i') \;-\; p_i(v_i').
\]

Recall the payment for each reported profile,
\begin{align*}
    p_i(b,b) &= 2b\pi_i(b) - (b-a)(\pi_i(a,b) +  \pi_i(a,a))\\
    p_i(b,a) &= p_i(a,b) = b\pi_i(b) + a \pi_i(a,b) - (b-a)\pi_i(a,a)\\
    p_i(a,a) &= 2a \cdot \pi_i(a,a)
\end{align*}

We need to check for all four true types and all their possible misreports. By symmetry, the profiles $(b,a)$ and $(a,b)$ are equivalent; therefore, we are going to show the inequalities only for $(a,b)$.

\medskip\noindent
\textbf{Case 1: }  Agent's $i$ true type $v_i = (b,b)$. Her truthful expected utility is:
\begin{align*}
\Ex{\divutil{(b,b)}{(b,b)}}& =  2b\,\pi_i(b)- (2b\pi_i(b) - (b-a)(\pi_i(a,b) +  \pi_i(a,a)))=(b-a)\bigl(\pi_i(a,b)+\pi_i(a,a)\bigr).
\end{align*}

The expected utility for deviating to either $(a,b)$ or $(a,a)$ is 
\begin{align}
\Ex{\divutil{(b,b)}{(a,b)}}
&=b(\pi_i(b)+\pi_i(a,b))-p_i(a,b)\notag\\
&=b(\pi_i(b)+\pi_i(a,b)) - (b\pi_i(b) + a \pi_i(a,b) - (b-a)\pi_i(a,a))  \notag\\
&=(b-a)(\pi_i(a,b)+\pi_i(a,a)),  \label{eq: bb-ab}\\
\Ex{\divutil{(b,b)}{(a,a)}}
&=2b\,\pi_i(a,a)-p_i(a,a) \notag\\
&= 2b\,\pi_i(a,a) - 2a\,\pi_i(a,a)\notag\\
&= 2(b-a)\pi_i(a,a). \label{eq: bb-aa}
\end{align}

It is easy to see that the deviation $(b,b)\to(b,a)$, \Cref{eq: bb-ab}, has the same utility as truth-telling. Hence, for the deviation $(b,b)\to(a,a)$ to be not profitable, from \Cref{eq: bb-aa} it must be that $\pi_i(a,a) \le \pi_i(a,b).$

\medskip\noindent
\textbf{Case 2:}   Agent's $i$ true type $v_i = (a,b)$. Her truthful expected utility is:
\begin{align*}
\Ex{\divutil{(a,b)}{(a,b)}} & = a\,\pi_i(a,b) + b\pi_i(b) - (b\pi_i(b) + a \pi_i(a,b) - (b-a)\pi_i(a,a))=(b-a)\,\pi_i(a,a).
\end{align*}

In this case, we need to check all possible deviations to $(b,b),(b,a)$, and $(a,a)$. The expected utilities are: 
\begin{align}
\Ex{\divutil{(a,b)}{(b,b)}} &=(a+b)\,\pi_i(b)-p_i(b,b)\notag\\
&=(a+b)\,\pi_i(b)- \left(2b\,\pi_i(b) - (b-a)\bigl(\pi_i(a,b) + \pi_i(a,a)\bigr)\right)\notag\\
& = (b-a)\left(\pi_i(a,b) + \pi_i(a,a) - \pi_i(b)\right)\label{eq: ab-bb}\\
\Ex{\divutil{(a,b)}{(b,a)}} &=a\,\pi_i(b)+b\,\pi_i(a,b)-p_i(b,a)\notag\\
&=a\,\pi_i(b) + b\,\pi_i(a,b) - \left(b\,\pi_i(b) + a\,\pi_i(a,b) - (b-a)\,\pi_i(a,a)\right)\notag\\
&=(b-a)\left(\pi_i(a,b) + \pi_i(a,a) - \pi_i(b)\right)\label{eq: ab-ba}\\
\Ex{\divutil{(a,b)}{(a,a)}} &=(a + b)\,\pi_i(a,a) - p_i(a,a)\notag\\
&= (a + b)\,\pi_i(a,a) - 2a\,\pi_i(a,a) \notag\\
&= (b - a)\,\pi_i(a,a) \label{eq: ab-aa}
\end{align}

All three expected utilities must be at most $(b-a)\pi_i(a,a)$. Notice again that the deviation $(a,b)\to(a,a)$, \Cref{eq: ab-aa}, produces the same expected utility as truthtelling. From \Cref{eq: ab-bb},\Cref{eq: ab-ba} for incentive compatibility we must have $\pi_i(a,b)\le\pi_i(b)$.

\medskip\noindent
\textbf{Case 3:}   Agent's $i$ true type $v_i = (a,a)$. Her truthful expected utility is:
\begin{align*}
 \Ex{\divutil{(a,a)}{(a,a)}} &= 2a\,\pi_i(a,a)-2a\,\pi_i(a,a)=0 
\end{align*}

The expected utility for deviating to either $(b,b)$ or $(a,b)$ is 
\begin{align}
\Ex{\divutil{(a,a)}{(b,b)}}
&=a(\pi_i(b)+\pi_i(a,b))-p_i(a,b)\notag\\
&= 2a\,\pi_i(b) - \left(2b\,\pi_i(b) - (b-a)(\pi_i(a,b) + \pi_i(a,a))\right) \notag\\
&= (b - a)\left(\pi_i(a,b) + \pi_i(a,a) - 2\,\pi_i(b)\right),  \label{eq: aa-bb}\\
\Ex{\divutil{(a,a)}{(a,b)}}
&=a(\pi_i(b)-p_i(a,b) \notag\\
&= a\,(\pi_i(a,b) + \pi_i(b)) - \left(b\,\pi_i(b) + a\,\pi_i(a,b) - (b-a)\,\pi_i(a,a)\right)\notag\\
&=  (b - a)\left(\pi_i(a,a) - \pi_i(b)\right). \label{eq: aa-ab}
\end{align}

In this case, for the mechanism to be incentive compatible, the expected utility of the deviation must be negative. That is, form \Cref{eq: aa-bb} we get $\pi_i(a,b) + \pi_i(a,a) \le 2\,\pi_i(b)$ and from  \Cref{eq: aa-ab} we get $\pi_i(a,a) \le \pi_i(b)$.

Notice that combining the conditions from Case 1 and Case 2 we need $\pi_i(a,a)\;\le\;\pi_i(a,b)\;\le\;\pi_i(b)$. Thus, implies $\pi_i(a,b) + \pi_i(a,a) \le 2\,\pi_i(b)$. This concluded the proof.
\end{proof}

 \begin{proof}[Proof of \Cref{lemma: non-iid-items-payments} ]

We want to show that for all $v\in\Vcal$, and the
flow decomposition $\xi$, the payment rule is given by 
\[
p(v) = \frac{1}{\Pr[v]}\sum_{\ell \in \Pcal_{v}} \xi_{\ell} \left(\sum_{j \in \{1,2\}} v_{j} \pi_{j}(v)\quad - \sum_{z \in [1,|\ell|-3]}  \sum_{j \in \{1,2\}} \left( v_{j}^{z}-v_{j}^{z+1} \right)\pi_{j}(v^{z+1})\right)
\]

Recall the flow shown in~\Cref{fig: flow-non-iid-items-split}. Notice that all nodes, except $(b,b)$, have a unique path to $\bot$, since they have one outgoing edge. Hence, the probability $\Pr[v]$ is equal to the corresponding path flow of the decomposition, simplifying the payment rule to  $p(v) = \sum_{j \in \{1,2\}} v_{j} \pi_{j}(v) - \sum_{z \in [1,|\ell|-3]}  \sum_{j \in \{1,2\}} \left( v_{j}^{z}-v_{j}^{z+1} \right)\pi_{j}(v^{z+1}).$

We can quickly verify that the payment is given by 
\begin{align*}
    p(a,a) &= a \cdot \left( \paa + \qaa \right)\\
    p(a,b) &= a \pab + b \qbb - (b-a) \qaa\\
    p(b,a) &= a \qab + b \pbb - (b-a) \paa
\end{align*}



    

Finally, we are interested in the payment of node $(b,b)$. Recall the flow shown in~\Cref{fig: flow-non-iid-items-split} parametrized by $x$. There are two simple paths from $(b,b)$ to $\bot$. That is $\{(b,b), (b,a), (a,a)\}$ with flow $(1 - p)(1 - q) - x$, and $\{(b,b), (a,b), (a,a)\}$ with flow $x$. Therefore, the payment is:

\begin{align*}
p(b,b) &= \frac{1}{\Pr[(b,b)]} \sum_{\ell \in P_{(b,b)}} \xi_\ell\left( \sum_{j \in \{1, 2\}} b \pi_j(b,b) - \sum_{z \in [2]} \sum_{j \in \{1,2\}} (v_j^z - v_j^{z+1}) \pi_j(v^{z+1}) \right) \\
&= \frac{1}{(1 - p)(1 - q)} \left( (1 - p)(1 - q) - x \right) \left( b(\pbb + \qbb) - (b - a)(\qab + \paa) \right) \\
&\quad\quad + \frac{1}{(1 - p)(1 - q)} \cdot x \left( b(\pbb + \qbb) - (b - a)(\pab + \qaa) \right) \\
&= b(\pbb + \qbb) - (b - a)(\qab + \paa) \\
&\quad\quad + \frac{x}{(1 - p)(1 - q)} (b - a) \left( \qab - \pab + \paa - \qaa \right)
\end{align*}
     
 \end{proof}

\section{Missing proofs of~\Cref{sec: thirdAxis}}\label{apx: proofs-thr-axis}

\begin{proof}[Proof of \Cref{lemma: non-iid-items-monotone}]
A straightforward calculation shows
\[
H_1(a,b)-H_1(a,a)
=\Bigl[\tfrac{1-p-x}{p\,q}-\tfrac{x}{p\,(1-q)}\Bigr](b-a)
=\frac{(1-p)(1-q)-x}{p\,q\,(1-q)}\,(b-a)\;\ge0,
\]
since $x\le(1-p)(1-q)$ and $b>a$.  Likewise,
\[
H_2(b,a)-H_2(a,a)
=\Bigl[\tfrac{p(1-q)+x}{p\,q}-\tfrac{(1-p)(1-q)-x}{(1-p)\,q}\Bigr](b-a)
=\frac{x}{p\,q\,(1-p)}\,(b-a)\;\ge0,
\]
since $x\ge0$ and $b>a$.
\end{proof}

\begin{proof}[Proof of~\Cref{lemma: non-iid-items-probabilities-monotonicity}]
Recall that 
\[
  \pbb=\frac{1-p^n}{n(1-p)},\quad
  \qbb=\frac{1-q^n}{n(1-q)},
  \]
  \[\pab=\frac{1}{n}p^{n-1}\frac{1-q^n}{1-q},
  \quad
  \qab=\frac{1}{n}q^{n-1}\frac{1-p^n}{1-p}.
\]
\[\paa = \qaa = \frac{1}{n}(pq)^{n-1}\]
First, we observe that $\frac{1-p^n}{1-p}=\sum_{\ell=0}^{n-1}p^\ell$. We start by computing the difference of $\pbb,\qbb$
\[
\pbb-\qbb
=\frac{1}{n}\sum_{\ell=0}^{n-1}\bigl(p^\ell - q^\ell\bigr).
\]
Since $p\ge q$ gives $p^\ell\ge q^\ell$ for all $\ell$, so each term in the sum is nonnegative.  Therefore
$\pbb\ge\qbb$.

Next, we compute the difference of $\pab,\qab$:
\[
\pab-\qab=p^{\,n-1}\frac{1}{n}\sum_{\ell=0}^{n-1}q^\ell - q^{\,n-1}\frac{1}{n}\sum_{\ell=0}^{n-1}p^\ell
=\frac{1}{n}\sum_{\ell=0}^{n-1}\bigl(p^{\,n-1}q^\ell - q^{\,n-1}p^\ell\bigr)
=\frac{1}{n}\sum_{\ell=0}^{n-1}(pq)^\ell\bigl(p^{\,n-1-\ell}-q^{\,n-1-\ell}\bigr).
\]
Since $p\ge q$ implies $p^{\,n-1-\ell}\ge q^{\,n-1-\ell}$ for each $0\le \ell\le n-1$, every term in the sum is nonnegative.  Hence $\pab-\qab\ge0.$

We move on to calculating the difference between $\pbb$ and $\pab$
\[\pbb-\pab=\frac{1-p^n}{n(1-p)} -\frac{1}{n}p^{n-1}\frac{1-q^n}{1-q} = \frac{1}{n}\sum_{\ell=0}^{n-1}p^\ell - \frac{1}{n}p^{n-1}\sum_{\ell=0}^{n-1}q^\ell \]
Since $q\le p \le 1$, for all $\ell$ we have $p^\ell\ge q^\ell$ and $p^{n-1}\le 1$ we get $\sum_{i=0}^{n-1}p^i \ge \sum_{i=0}^{n-1}q^i \ge p^{n-1}\sum_{i=0}^{n-1}q^i.$ Thus,  $\pbb-\pab\ge0.$

We move on to calculating the difference between $\qbb$ and $\qab$
\[\qbb-\qab=\frac{1-q^n}{n(1-q)} -\frac{1}{n}q^{n-1}\frac{1-p^n}{1-p} = \frac{1}{n}\sum_{\ell=0}^{n-1}q^\ell - \frac{1}{n}q^{n-1}\sum_{\ell=0}^{n-1}p^\ell = \sum_{\ell=0}^{n-1}q^\ell (1-q^{n-\ell-1}p^\ell) \]
Since $0 \ge q\le p \le 1$, for all $\ell$ we have $q^\ell\ge 0$ and $ 1-q^{n-\ell-1}p^\ell \ge 0$. Thus,  $\qbb-\qab\ge0.$

Finally, we can easily see that 
\[\pab =\frac{1}{n}p^{n-1}\frac{1-q^n}{1-q} = \frac{1}{n}p^{n-1}\sum_{\ell=0}^{n-1}q^\ell \ge \frac{1}{n}p^{n-1}q^{n-1} =\paa \]

Therefore, $\pab \ge \paa.$ By symmetry $\qab\ge \qaa.$ 
\end{proof}

\begin{proof}[Proof of~\Cref{lemma: non-iid-items-BIC}]

In this case, the agents are iid, and therefore, we drop the subscript $i$. To show that the mechanism is BIR, it suffices to show that reporting truthfully has non-negative utility. 

To show that the mechanism is BIC we must show that the expected utility of misreporting is at most the utility of truth-telling. Consider a true profile $v=(v_1,v_2)\in\{a,b\}^2$ and a misreport $v'=(v_1',v_2')\in\{a,b\}^2$. Recall, her utility when reporting truthfully is $\Ex{\divutil{v}{v}}
=\sum_{j\in\{1,2\}} v_j\,\pi_j(v) - p(v),$ and her utility when misreporting is $\Ex{\divutil{v}{v'}}
=\sum_{j\in\{1,2\}} v_j\,\pi_j(v') - p(v')$. BIC demands $\Ex{\divutil{v}{v'}}\le\Ex{\divutil{v}{v}}$.

The payments for each type are given by the following formulas 
\[
\begin{aligned}
p(a,a)&=a\bigl(\paa+\qaa\bigr),\\
p(a,b)&=a\,\pab+b\,\qbb-(b-a)\,\qaa,\\
p(b,a)&=b\,\pbb+a\,\qab-(b-a)\,\paa,
\end{aligned}
\]
while for $p(b,b)$ we have two possible variants:
\[
\begin{aligned}
\text{Variant I:}\quad
p^1(b,b)&=b\bigl(\pbb+\qbb\bigr)
-(b-a)\bigl(\qab+\paa\bigr),\\
\text{Variant II:}\quad
p^2(b,b)&=b\bigl(\pbb+\qbb\bigr)
-(b-a)\bigl(\pab+\qaa\bigr).
\end{aligned}
\]

We check all four true types.  In each case, we display the truthful utility and the three deviations, then read off the required inequalities.

\medskip\noindent
\textbf{Case 1:} $v=(a,a)$.  Truth‐telling:
\[
\Ex{\divutil{(a,a)}{(a,a)}} 
=2a(\paa+\qaa)-p(a,a)
=0.
\]
Deviations:
\begin{align*}
\Ex{\divutil{(a,a)}{(a,b)}}
&=a\,\pab+a\,\qbb-p(a,b)
=(b-a)(\qaa-\qbb),\\
\Ex{\divutil{(a,a)}{(b,a)}}
&=a\,\pbb+a\,\qab-p(b,a)
=(b-a)(\paa-\pbb),\\
\Ex{\divutil{(a,a)}{(b,b)}}_{I}
&=a\,\pbb+a\,\qbb-p^1(b,b)=(b-a)(\paa+ \qab -\pbb-\qbb)\\
\Ex{\divutil{(a,a)}{(b,b)}}_{II}
&=a\,\pbb+a\,\qbb-p^2(b,b)=(b-a)(\pab+ \qaa -\pbb-\qbb)
\end{align*}
For the deviations not to be profitable, the above expected utilities must be negative. Hence, we have
\begin{align}
\qbb\ge\qaa,&\quad\pbb\ge\paa,\label{eq: case1}\\
(I)\;\pbb+\qbb&\ge\paa+\qab,\label{eq: case1-variantI}\\
(II)\;\pbb+\qbb &\ge \pab+\qaa,\label{eq: case1-variantII}
\end{align}

\medskip\noindent
\textbf{Case 2:} $v=(a,b)$.  Truth‐telling:
\[
\Ex{\divutil{(a,b)}{(a,b)}}
=a\,\pab+b\,\qbb-p(a,b)
=(b-a)\,\qaa\ge0\tag{$b>a,\qaa\ge0$}\] 
Deviations:
\begin{align*}
\Ex{\divutil{(a,b)}{(a,a)}}
&=a\,\paa+b\,\qaa-p(a,a)=(b-a)\qaa,\\
\Ex{\divutil{(a,b)}{(b,a)}}
&=a\,\pbb+b\,\qab-p(b,a)
=(b-a)(\qab + \paa - \pbb),\\
\Ex{\divutil{(a,b)}{(b,b)}}_{ I}
&=a\,\pbb+b\,\qbb-p^1(b,b) = (b-a)(\qab + \paa - \pbb)\\
\Ex{\divutil{(a,b)}{(b,b)}}_{ II}
&=a\,\pbb+b\,\qbb-p^2(b,b)=(b-a)(\pab+\qaa - \pbb)
\end{align*}
Ensuring these utilities are $\le(b-a)\,\qaa$ gives us the following constraints
\begin{align}
    \pbb+\qaa&\ge\qab+\paa,\label{eq: case2}\\
    (II)\;\pbb&\ge\pab,\label{eq: case2-variantII}
\end{align}

It is easy to see that the deviation $(a,b)\to(a,a)$ has the same expected utility. Notice that the deviation $(a,b)\to(b,a)$ and $((a,b)\to(b,b))_{ I}$ have the same BIC restrictions, and hence we only include the first that must hold in all cases, not just in variation $I$.

\medskip\noindent
\textbf{Case 3:} $v=(b,a)$.  Truth‐telling:
\[
\Ex{\divutil{(b,a)}{(b,a)}}
=b\,\pbb+a\,\qab-p(b,a)
=(b-a)\,\paa\ge0 \tag{$b>a,\paa\ge0$}\]
Deviations:
\begin{align*}
\Ex{\divutil{(b,a)}{(a,a)}}&=b\paa+a\qaa - p(a,a)=(b-a)\paa,\\
\Ex{\divutil{(b,a)}{(a,b)}} 
&=b\pab +a\qbb - p(a,b) = (b-a)(\pab+\qaa-\qbb),\\
\Ex{\divutil{(b,a)}{(b,b)}}_{ I}
&=b\pbb+a\qbb - p^1(b,b) = (b-a)(\qab+\paa -\qbb) \\
\Ex{\divutil{(b,a)}{(b,b)}}_{ II}
&=b\pbb+a\qbb - p^1(b,b) = (b-a)(\pab+\qaa -\qbb).
\end{align*}
To satisfy the BIC constraints, the utilities must be at most $(b-a)\paa$, which implies
\begin{align}
    \qbb+\paa&\ge\pab+\qaa,\label{eq: case3}\\
    (I)\;\qbb&\ge\qab,\label{eq: case3-variantI}
\end{align}
Notice again that the deviation along an edge with a positive flow ($(b,a)\to(a,a)$) has the same expected utility, hence we do not write the constraint. Furthermore, 
the deviations $(b,a)\to(a,b)$ and $((b,a)\to(b,b))_{ II}$ have the same BIC restrictions, and thus we only include the general one. 

\medskip\noindent
\textbf{Case 4:} $v=(b,b)$.  Truth‐telling has two forms, depending on the payment for reporting $(b,b)$,
\[
\Ex{\divutil{(b,b)}{(b,b)}}_{ I}
=b\,\pbb+b\,\qbb-p^1(b,b)
=(b-a)(\qab+\paa)\ge0,
\tag{$b>a,\qab,\paa\ge0$}\]
\[
\Ex{\divutil{(b,b)}{(b,b)}}_{ II}
=b\,\pbb+b\,\qbb-p^2(b,b)
=(b-a)(\pab+\qaa)\ge0,\tag{$b>a,\pab,\qaa\ge0$}
\]

Deviations:
\begin{align*}
\Ex{\divutil{(b,b)}{(a,a)}}&=b(\paa+\qaa) - p(a,a) = (b-a)(\paa+\qaa),\\
\Ex{\divutil{(b,b)}{(a,b)}} 
&=b(\pab+\qbb)-p(a,b) = (b-a)(\pab+\qaa),\\
\Ex{\divutil{(b,b)}{(b,a)}} 
&=b(\pbb+\qab) - p(b,a) = (b-a)(\qab+\paa).
\end{align*}
These give us two sets of inequalities. When the payment is according to Variant~$I$ and the expected utility when reporting truthfully is $(b-a)(\qab+\paa)$, we get the BIC conditions are
\begin{equation}
   \qab\ge \qaa, \quad  \qab + \paa \ge \pab+\qaa.\label{eq: case4-variantI} 
\end{equation}
When the payment is according to Variant~$II,$ and the truthful expected utility is $(b-a)(\pab+\qaa)$ we get the following BIC conditions
\begin{equation}
\pab\ge \paa, \quad  \pab+\qaa \ge \qab + \paa. \label{eq: case4-variantII} 
\end{equation}

\medskip
Collecting all four cases, we can derive all the necessary equations for each variant. First for Variant~$I$, from \Cref{eq: case1,eq: case3-variantI,eq: case4-variantI}, we get that $\qbb\ge\qab\ge\qaa,$ and from \Cref{eq: case1} we get $\pbb\ge\paa$. \Cref{eq: case1-variantI} is implied by \Cref{eq: case1,eq: case3-variantI}. Furthermore, \Cref{eq: case3} is implied by \Cref{eq: case4-variantI,eq: case3-variantI}. Hence, the necessary conditions for Variant~$I$ are $\pbb\ge\paa,\;\qbb\ge\qab\ge\qaa,\;\pbb+\qaa\ge\qab+\paa,$ and $  \qab + \paa \ge \pab+\qaa.$

Similarly, collecting all the necessary equations for Variant~$II,$ from \Cref{eq: case1,eq: case2-variantII,eq: case4-variantII} we get $\pbb\ge\pab\ge\paa,$ and from \Cref{eq: case1} we get $\qbb\ge\qaa$. \Cref{eq: case2-variantII} is implied by \Cref{eq: case4-variantII,eq: case2-variantII}. Also, \Cref{eq: case1-variantII} is implied by \Cref{eq: case4-variantII,eq: case2-variantII}. Thus, all the conditions for Variant~$II$ are $\qbb\ge\qaa,\pbb\ge\pab\ge\paa,\pab + \qaa \ge \qab +\paa,$ and $\qbb+\paa \ge \pab + \qaa$. Which completes the proof.
\end{proof}

\subsection{Missing Proofs from \Cref{subsec: region1}-\Cref{subsec: region7}}\label{apx: regions}

\begin{proof}[Proof of\Cref{lemma: region1-induce}]

Consider any $x\in(\frac{a}{b-a}p(1-q),(1-p)(1-q)-\frac{a}{b-a}(1-p)q)$. By Table~\ref{tab: virtual-values-non-iid-items} and the flow in~\cref{fig: flow-non-iid-items-split}, one checks  $H_1(a,b) = a - 
\frac{x}{p(1-q)}(b-a) \le 0$ and $H_2(b,a) =  a - \frac{(1-p)(1-q)-x}{(1-p)q}(b-a)\le 0$. By~\Cref{lemma: non-iid-items-monotone}, we then have $H_1(a,b)\le H_1(a,b) \le 0$ and $H_2(a,a)\le H_2(b,a)\le 0$. Hence, whenever $v_j=a$, $H_j(v)<0$ and therefore $\pi_j(v)=0$. On the other hand, following \Cref{eq: pbb} (\Cref{eq: qbb} respectively) we have that if $v_1=b$, then $\pbbr{1} = \frac{1-p^n}{n(1-p)}$ and if $v_2=b$, then $\qbbr{1} = \frac{1-q^n}{n(1-q)}.$

Applying \Cref{lemma: non-iid-items-payments}, we can quickly verify that the payment of~\Cref{mech: Region1} is induced by the flow.
\end{proof}

\begin{proof}[Proof of~\Cref{lemma: region1-BIC}]
 First, we can easily confirm that the mechanism is according to the Variant~$I$ defined in \Cref{lemma: non-iid-items-BIC}. Hence, it suffices to show that the interim allocation probabilities satisfy the set of inequalities~(\ref{eq: BIC-variantI}). Since $\pabr{1},\qabr{1},\paar{1},\qaar{1} = 0$, all the inequalities are trivially true. 
\end{proof}

\begin{proof}[Proof of~\Cref{lemma: region2-induce}]
By construction $x=\frac{a}{b-a}p(1-q)$ is exactly the threshold that makes $H_1(a,b)$ equal zero, and monotonicity (Lemma~\ref{lemma: non-iid-items-monotone}) gives $H_1(a,a)\le0$.  Next, using \Cref{tab:thresholds}, one checks
\[
x\;\ge\;(1-p)(1-q)-\frac{a}{b-a}(1-p)q
\quad\Longrightarrow\quad
H_2(b,a)>0.
\]
\[
x<\tfrac{a}{b-a}\,p\;q-p(1-q)
\quad\Longrightarrow\quad
q<\tfrac{b}{b+a}
\quad\Longrightarrow\quad
H_2(a,a)<0.
\]

where the last part holds since $q<p<(b-a)/b<b/(b+a).$ The interim allocation probabilities of $\pbbr{2},\qbbr{2},\qabr{2}$ are given by \Cref{eq: pbb,eq: qbb,eq: qab}, since the corresponding virtual values are positive. Also, we have that $\paar{2}=\qaar{2}=0$, due to the fact that $H_1(a,a),H_2(a,a)<0.$ Finally, \Cref{eq: BIC-variantI-eq} dictates that $\pabr{2}=\qabr{2}.$  Applying \Cref{lemma: non-iid-items-payments} gives us the matching payments.
 
\end{proof}

\begin{proof}[Proof \Cref{lemma: region2-BIC}]
    In Region 2, the flow $x=\frac{a}{b-a}p(1-q)$ is less than $(1-p)(1-q)$.Therefore, the mechanism falls under Variant~$I$. The BIC constraints according to \Cref{lemma: non-iid-items-BIC}
    \begin{align}    
    \qbb\ge\qab&\ge\qaa \label{eq:reg2-1}\\
    \qab + \paa &\ge \pab +\qaa \label{eq:reg2-2}\\
    \pbb+\qaa &\ge \qab + \paa\label{eq:reg2-3}\\
    \pbb &\ge \paa\label{eq:reg2-4}
    \end{align}

    In this region, $\paar{2} = \qaar{2} = 0$. Hence, \Cref{eq:reg2-4} trivially holds. By~\Cref{lemma: non-iid-items-probabilities-monotonicity},  \Cref{eq:reg2-1} is immediate, and further more we have that $\pbbr{2}\ge \qbbr{2} \ge \qabr{2}$ implying \Cref{eq:reg2-3}. Finally, by construction $\pabr{2} = \qabr{2}$ implying \Cref{eq:reg2-2}.  
\end{proof}

\begin{proof}[Proof of \Cref{lemma: region3-induce}]
    First, we focus on the virtual values. Using \Cref{tab:thresholds} we can confirm that  
    \[
H_1(a,b)>0,\quad
H_2(b,a)>0,\quad
H_1(a,a)=0,\quad
H_2(a,a)<0,
\]

Starting with $H_1(a,b)>0,$ the flow must satisfy $x=1-p-\frac{a}{b-a}pq \le \frac{a}{b-a}p(1-q)$  which holds when $p\ge\frac{b-a}{b},$ matching the boundary of the region. Next, for the virtual value of  $H_2(b,a)$ to be positive we must have $x=1-p-\frac{a}{b-a}pq \ge (1-p)(1-q)-\frac{a}{b-a}(1-p)q$, equivalently $p\le \frac{b}{b+a}.$ Finally, for $H_2(a,a)<0,$ the flow must be $x=1-p-\frac{a}{b-a}pq\ge \frac{a}{b-a}pq - p(1-q),$ which simplifies to $pq\ge\frac{b-a}{b+a}.$

We calculate the interim allocation probabilities using \Cref{eq: pbb,eq: qbb,eq: pab,eq: qab} for the virtual values that are strictly positive. Since $H_2(a,a)<0$, we get $\qaar{3}=0.$ Finally, \Cref{eq: BIC-variantI-eq} dictates that $\paar{3} = \pabr{3}-\qabr{3}.$  Applying \Cref{lemma: non-iid-items-payments} gives us the matching payments.

\end{proof}

\begin{proof}[Proof of \Cref{lemma: region3-BIC}]
In Region 3, the flow $x=1-p-\frac{a}{b-a}pq$ is less than $(1-p)(1-q),$ for $p\ge\frac{b-a}{b}.$  Therefore, the mechanism falls under Variant~$I$. The BIC constraints according to \Cref{lemma: non-iid-items-BIC}
    \begin{align}    
    \qbb\ge\qab&\ge\qaa \label{eq:reg3-1}\\
    \qab + \paa &\ge \pab +\qaa \label{eq:reg3-2}\\
    \pbb+\qaa &\ge \qab + \paa\label{eq:reg3-3}\\
    \pbb &\ge \paa\label{eq:reg3-4}
    \end{align}

In this region, $\qaar{3} = 0$. Note that $\paar{3}\le \paa$ (\Cref{eq: paa}), since we flip a biased coin before allocating uniformly at random.  By~\Cref{lemma: non-iid-items-probabilities-monotonicity}, we immediately have~\Cref{eq:reg3-1,eq:reg3-4}.By construction $\paar{3}=\pabr{3}-\qabr{3}$, satisfying~\Cref{eq:reg3-2} with equality. Plugging in the previous equality to~\Cref{eq:reg3-3} we get $\pbbr{3} \ge \qabr{3} + \paar{3} \implies \pbbr{3} \ge\pabr{3}$, which follows again by~\Cref{lemma: non-iid-items-probabilities-monotonicity}.    
    
\end{proof}

\begin{proof}[Proof of \Cref{lemma: region4-induce}]
First we observe that for $p\ge\frac{b}{b+a}$ $ q\le\frac{b-a}{b}$ the flow $x=(1-p)(1-q)-\frac{a}{b-a}(1-p)q$ is feasible since $0\le x\le (1-p)(1-q).$ Next, we show that according to \Cref{tab:thresholds} we can confirm that  
    \[
H_1(a,b)\ge0,\quad
H_2(b,a)=0,\quad
H_1(a,a)\ge0,\quad
H_2(a,a)\le 0,
\]

We start by showing $H_1(a,a)\ge0$. The flow $x=(1-p)(1-q)-\frac{a}{b-a}(1-p)q$ must be at least $1-p-\frac{a}{b-a}pq$. That implies $p\ge \frac{b}{b+a}$ matching the boundary of the region. Monotonicity of the virtual values (\Cref{lemma: non-iid-items-monotone}) implies that $H_1(a,b)\ge H_1(a,a)\ge0.$. Finally, for $H_2(a,a)\le 0$ holds when $x=(1-p)(1-q)-\frac{a}{b-a}(1-p)q \ge \frac{a}{b-a}pq-p(1-q)$ equivalently $q\le\frac{b-a}{b}$, matching again the boundary of the region. Applying \Cref{lemma: non-iid-items-payments} gives us the matching payments.

\end{proof}

\begin{proof}[Proof of \Cref{lemma: region4-BIC}]
In Region 4, the flow $x=(1-p)(1-q)-\frac{a}{b-a}(1-p)q$ is at least zero for $q\le\frac{b-a}{b}.$  Therefore, the mechanism falls under Variant~$I$. The BIC constraints according to \Cref{lemma: non-iid-items-BIC}
    \begin{align}    
    \qbb\ge\qab&\ge\qaa \label{eq:reg4-1}\\
    \qab + \paa &\ge \pab +\qaa \label{eq:reg4-2}\\
    \pbb+\qaa &\ge \qab + \paa\label{eq:reg4-3}\\
    \pbb &\ge \paa\label{eq:reg4-4}
    \end{align}

In this region, $\qaar{4} = 0$. Note that $\qabr{4}\le \qab$ (\Cref{eq: qab}), since we flip a biased coin before allocating uniformly at random.  By~\Cref{lemma: non-iid-items-probabilities-monotonicity}, we immediately have~\Cref{eq:reg4-1,eq:reg4-4}. By construction $\qabr{4} = \pabr{4}-\paar{4}$, satisfying~\Cref{eq:reg4-2} with equality. Plugging in the previous equality to~\Cref{eq:reg4-3} we get $\pbbr{3} \ge \qabr{3} + \paar{3} \implies \pbbr{3} \ge\pabr{3}$, which follows again by~\Cref{lemma: non-iid-items-probabilities-monotonicity}.    
    
\end{proof}

\begin{proof}[Proof of~\Cref{lemma: region5-induce}]
First we observe that for $p\ge\frac{b}{b+a}$ $ q\le\frac{b-a}{b}$ the flow $x=\frac{a}{b-a}pq-p(1-q)$ is feasible since $0\le x\le (1-p)(1-q).$ Next, we show that according to \Cref{tab:thresholds} we can confirm that  
    \[
H_1(a,b)\ge0,\quad
H_2(b,a)\ge0,\quad
H_1(a,a)=0,\quad
H_2(a,a)\le 0,
\]

We start by showing $H_1(a,a)\ge0$. The flow $x=\frac{a}{b-a}pq-p(1-q)$ must be at least $1-p-\frac{a}{b-a}pq$. That implies $pq\ge \frac{b-a}{b+a}$ matching the boundary of the region. Monotonicity of the virtual values (\Cref{lemma: non-iid-items-monotone}) implies that $H_1(a,b)\ge H_1(a,a)\ge0$ and  $H_2(b,a)\ge H_2(a,a)\ge0$ where the final inequality holds by construction. Applying \Cref{lemma: non-iid-items-payments} gives us the matching payments.

\end{proof}

\begin{proof}[Proof of \Cref{lemma: region5-BIC}]
In Region 5, the flow $x=\frac{a}{b-a}pq-p(1-q)$ is less than $(1-p)(1-q),$ for $\frac{1-q}{pq}>\frac{a}{b-a}$. Therefore, the mechanism falls under Variant~$I$. The BIC constraints according to \Cref{lemma: non-iid-items-BIC}
    \begin{align}    
    \qbb\ge\qab&\ge\qaa \label{eq:reg5-1}\\
    \qab + \paa &\ge \pab +\qaa \label{eq:reg5-2}\\
    \pbb+\qaa &\ge \qab + \paa\label{eq:reg5-3}\\
    \pbb &\ge \paa\label{eq:reg5-4}
    \end{align}

By construction, \Cref {eq:reg5-2} is satisfied with equality. ~\Cref{eq:reg5-1,eq:reg5-4} are direct applications of ~\Cref{lemma: non-iid-items-probabilities-monotonicity}. Finally, using the fact that $\qaar{5} = \qabr{5}-\pabr{5}+\paar{5},$ \Cref{eq:reg3-3} becomes $\pbbr{5} \ge\pabr{5}$ which follows again by~\Cref{lemma: non-iid-items-probabilities-monotonicity}.     
    
\end{proof}

\begin{proof}[Proof of \Cref{lemma: region6-BIC}]

In Region 6, the flow $x=(1-p)(1-q)$. Therefore, the mechanism falls under Variant~$II$. The BIC constraints according to \Cref{lemma: non-iid-items-BIC}(\Cref{eq: BIC-variantII})

\begin{align}
    \pbb\ge\pab&\ge\paa, \label{eq: reg6-1}\\
    \pab + \qaa &\ge \qab +\paa,\label{eq: reg6-2} \\
    \qbb+\paa &\ge \pab + \qaa,\label{eq: reg6-3}\\
    \qbb &\ge \qaa\label{eq: reg6-4}
\end{align}

In this region, $\qaar{6} = 0$. Hence~\Cref{eq: reg6-4} is trivially true. By definition, $\pab \ge \qab +\paa$ must hold to be in Region 6, since $\qaar{6} = 0$ implies that~\Cref{eq: reg6-2} must also hold. Due to~\Cref{lemma: non-iid-items-probabilities-monotonicity}, ~\Cref{eq: reg6-4}. Finally, we must show that $\qbb+\paa \ge \pab + \qaa$. By definition of~\Cref{mech: Region6}, $\qaar{6}=0$ $\pabr{6} = p^{n-1}\qbbr{6},$ Therefore, $\qbbr{6}\ge\pabr{5 }\ge \pabr{6}-\paar{6},$ which concludes the proof.   
    
\end{proof}

\begin{proof}[Proof of \Cref{lemma: equal-revenue}]
Recall that the expected revenue of a mechanism is $\mathrm{Rev} = \sum_{i\in [n]}\sum_{v_i\in\Vcal_i} \Pr[v_i]\cdot p_i(v).$ The profiles for each agent are $\Vcal_i=\{(b,b),(b,a),(a,b),(a,a)\}.$ Note that the mechanism is symmetric (i.e., treats all agents identically), we focus on the revenue generated by one agent 
\[\mathrm{Rev}_i = (1-p)(1-q)\cdot p(b,b) + (1-p)q\cdot p(b,a) + p(1-q)\cdot p(a,b) + pq\cdot p(a,a) \]

We are interested in computing the difference between the revenue of the flow-induced mechanism $\mathrm{Rev}_i^{(*)}$ and the modified hierarchy mechanism $\mathrm{Rev}_i^{(6)}.$ Between the two mechanisms $\pab$ and $\paa$ are different. Therefore, we get that 
\begin{align*}
    p^{(*)}(b,b) -p^{(6)}(b,b) = (b-a)(\pabr{6}-\pabr{*}),\\
    p^{(*)}(b,a) -p^{(6)}(b,a) = (b-a)(\paar{6}-\paar{*}),\\
    p^{(*)}(a,b) -p^{(6)}(a,b) = a(\pabr{*}-\pabr{6}),\\
    p^{(*)}(a,a) -p^{(6)}(a,a) = a(\paar{*}-\paar{6}).
\end{align*}

Combining the above, we get that

\begin{align*}
    \mathrm{Rev}_i^{(*)} - \mathrm{Rev}_i^{(6)} &= (1-p)(1-q)\cdot(p^{(*)}(b,b) -p^{(6)}(b,b)) + (1-p)q\cdot (p^{(*)}(b,a) -p^{(6)}(b,a) ) \\
    &\;\quad\quad+ p(1-q)\cdot ( p^{(*)}(a,b) -p^{(6)}(a,b)) + pq\cdot ( p^{(*)}(a,b) -p^{(6)}(a,b))\\
     &= (1-p)(1-q)(b-a)\bigl(\pabr{6}-\pabr{*}\bigr) + (1-p)q (b-a)\bigl(\paar{6}-\paar{*}\bigr) \\
    &\;\quad\quad+ p(1-q)a\bigl(\pabr{*}-\pabr{6}\bigr) + pq a\bigl(\paar{*}-\paar{6}\bigr)\\
    &= \bigl(\pabr{6}-\pabr{*}\bigr)  ((1-p)(1-q)(b-a)-ap(1-q)) \\
    &\;\quad\quad+ \bigl (\paar{6}-\paar{*}) \bigr)((b-a)(1-p)q - apq)\\
    &= \bigl(( a\pabr{6}-\pabr{*}\bigr)  (1-q)(b-a-pb) \\
    &\;\quad\quad+  \bigl(\paar{6}-\paar{*}\bigr)q(b-a-pb)\\
    &= (b-a-pb)\Bigl((1-q)\bigl(\pabr{6}-\pabr{*})  + q(\paar{6}-\paar{*}))\Bigr)\\
    &=(b-a-pb)\left( (1-q)\bigl(\pabr{6}-\pabr{*}\bigr) +q \bigl(\paar{6}-\paar{*}\bigr)\right)\\
    &= (b-a-pb)\left( (1-q)\left( \frac{1}{n}p^{n-1}\frac{1-q^n}{1-q} - \frac{1}{n}p^{n-1} \right) +q \left(\frac{1}{n}p^{n-1}q^{n-1}-\frac{1}{n}p^{n-1}\right)\right)\\
    &= \frac{1}{n}p^{n-1}(b-a-pb)(1-q^n -(1-q)+q^n-q)\\
    &=0
\end{align*}

\end{proof}

\begin{proof}[Proof of \Cref{lemma: region7-BIC}]

In Region 7, the flow $x=(1-p)(1-q)$. Therefore, the mechanism falls under Variant~$II$. The BIC constraints according to \Cref{lemma: non-iid-items-BIC}(\Cref{eq: BIC-variantII})

\begin{align}
    \pbb\ge\pab&\ge\paa, \label{eq: reg7-1}\\
    \pab + \qaa &\ge \qab +\paa,\label{eq: reg7-2} \\
    \qbb+\paa &\ge \pab + \qaa,\label{eq: reg7-3}\\
    \qbb &\ge \qaa\label{eq: reg7-4}
\end{align}

In this region, all the virtual values are positive. Due to~\Cref{lemma: non-iid-items-probabilities-monotonicity}, \Cref{eq: reg7-1,eq: reg7-4} directly follow. Recall that $\paar{7}=\qaar{7}$, \Cref{eq: reg7-2} and \Cref{eq: reg7-3} reduce to $\pab \ge \qab$ and $ \qbb\ge \pab $ respectively, which again is a direct application of~\Cref{lemma: non-iid-items-probabilities-monotonicity}.
    
\end{proof}

\end{document}